\documentclass[11pt]{article}

\usepackage{fullpage}
\usepackage{float}
\usepackage{bm}
\usepackage{algorithm}
\usepackage{amsmath}
\usepackage{amsthm}
\usepackage{amssymb}
\usepackage{amsfonts}
\usepackage[inline]{enumitem}
\usepackage{array}
\usepackage{multirow}
\usepackage{bbm}
\usepackage{enumitem}
\usepackage{dashrule}
\usepackage{tikz}
\usetikzlibrary{shapes.geometric, arrows}
\usepackage{caption}
\usepackage{changepage}
\usepackage{setspace}
\usepackage[en-US]{datetime2}
\usepackage[noadjust]{cite}
\usepackage{soul,xcolor}
\definecolor{DarkBlue}{rgb}{0.2,0.2,0.6}
\definecolor{DarkRed}{rgb}{0.368,0.097,0.078}
\usepackage[colorlinks = true,
    linkcolor = DarkRed,
    anchorcolor = black,
    citecolor = DarkBlue,
    filecolor = black,
    urlcolor = black]{hyperref}

\newcommand{\red}[1]{{\bf{\color{red}{#1}}}}
\renewcommand{\paragraph}[1]{\medskip\noindent{\bf #1}}

\newtheorem{theorem}		{Theorem} 	[section]
\newtheorem{lemma}			[theorem]	{Lemma}
\newtheorem{definition}		[theorem]	{Definition}
\newtheorem{assumption}		[theorem]	{Assumption}
\newtheorem{corollary}		[theorem]	{Corollary}

\newtheorem{remark}		    [theorem]	{Remark}
\newtheorem{example}	    [theorem]	{Example}

\DeclareMathSymbol{\N}{\mathbin}{AMSb}{"4E}
\DeclareMathSymbol{\Z}{\mathbin}{AMSb}{"5A}
\DeclareMathSymbol{\R}{\mathbin}{AMSb}{"52}
\DeclareMathSymbol{\Q}{\mathbin}{AMSb}{"51}
\DeclareMathSymbol{\erert}{\mathbin}{AMSb}{"50}
\DeclareMathSymbol{\I}{\mathbin}{AMSb}{"49}
\DeclareMathSymbol{\C}{\mathbin}{AMSb}{"43}
\DeclareMathSymbol{\E}{\mathbin}{AMSb}{"45}

\newcommand{\Lap}{{\rm Lap}}
\newcommand{\Out}{\mathsf{Out}}
\newcommand{\Space}{{\rm Space}}

\newcommand{\Stitch}{{\rm Stitch}}

\newcommand{\StepSize}{{\rm StepSize}}
\newcommand{\Output}{{\rm Output}}
\newcommand{\Next}{{\rm Next}}
\newcommand{\ST}{{\rm ST}}
\newcommand{\TN}{{\rm \mu}}
\newcommand{\STmon}{{\rm M}}
\newcommand{\STwrapper}{{\rm W}}

\newcommand{\TDE}{{\rm TDE}}
\newcommand{\TDEj}{{\rm TDE\text{-}j}}
\newcommand{\AllTDE}{{\rm All\text{-}TDE}}
\newcommand{\AllST}{{\rm All\text{-}ST}}
\newcommand{\ALG}{{\rm ALG}}
\newcommand{\DE}{{\rm DE}}
\newcommand{\FV}{{\rm FV}}
\newcommand{\PhaseSize}{{\rm PhaseSize}}
\newcommand{\MaxUpdateSize}{{\rm MuSize}}
\newcommand{\PrivateMed}{\texttt{PrivateMed}}
\newcommand{\StitchFrozenVals}{\texttt{StitchFrozenVals}}
\newcommand{\NoCapping}{\texttt{NoCapping}}
\newcommand{\Capp}{\texttt{Capp}}

\newcommand{\AAA}{\mathcal A}

\newcommand{\FFF}{\mathcal F}

\newcommand{\PPP}{\mathcal P}

\newcommand{\SSS}{\mathcal S}
\newcommand{\TTT}{\mathcal T}
\newcommand{\VVV}{\mathcal V}

\newcommand{\eps}{\varepsilon}

\DeclareMathOperator*{\argmin}{arg\,min} 

\def\argmin{\mbox{\rm argmin}}
\newcommand{\poly}{\mathop{\rm poly}}
\newcommand{\polylog}{\mathop{\rm polylog}}

\title{A Framework for Adversarial Streaming via Differential Privacy and Difference Estimators\footnote{In a previous version of the paper we claimed that our framework (stated in algorithm \texttt{RobustDE}) is directly applicable for tracking $F_2$ in the turnstile model. This is incorrect, as we are not aware of a construction of an appropriate difference estimator for $F_2$ in the turnstile model.}}

\author{Idan Attias\thanks{Ben-Gurion University.} \and Edith Cohen\thanks{Google Research and Tel Aviv University.} \and Moshe Shechner\thanks{Tel Aviv University. Partially supported by the Israel Science Foundation (grant 1871/19).} \and Uri Stemmer\thanks{Tel Aviv University and Google Research. Partially supported by the Israel Science Foundation (grant 1871/19)
and by Len Blavatnik and the Blavatnik Family foundation.} }

\date{September 25, 2022}

\begin{document}
\maketitle

%%%%%%%%%%%%%%%
% Abstract
%%%%%%%%%%%%%%%
\begin{abstract}
Classical streaming algorithms operate under the (not always reasonable) assumption that the input stream is fixed in advance. Recently, there is a growing interest in designing {\em robust streaming algorithms} that provide provable guarantees even when the input stream is chosen adaptively as the execution progresses. 
We propose a new framework for robust streaming that combines techniques from two recently suggested frameworks by Hassidim et al.~[NeurIPS 2020] and by Woodruff and Zhou~[FOCS 2021]. 
These recently suggested frameworks rely on very different ideas, each with its own strengths and weaknesses. We combine these two frameworks into a single hybrid framework that obtains the ``best of both worlds'', thereby solving a question left open by Woodruff and Zhou. 
\end{abstract}

%%%%%%%%%%%%%%%
% Introduction
%%%%%%%%%%%%%%%
\section{Introduction}\label{sec:introduction}

Streaming algorithms are algorithms for processing large data streams while using only a limited amount of memory, significantly smaller than what is needed to store the entire data stream. Data streams occur in many applications including computer networking, databases, and natural language processing.  The seminal work of Alon, Matias, and Szegedy~\cite{alon1999space} initiated an extensive theoretical study and further applications of streaming algorithms. 

In this work we focus on streaming algorithms that aim to maintain, at any point in time, an approximation for the value of some (predefined) real-valued function of the input stream. Such streaming algorithms are sometimes referred to as {\em strong trackers}. For example, this predefined function might count the number of distinct elements in the stream. Formally,

\begin{definition}\label{def:obliv}
Let $\AAA$ be an algorithm that, for $m$ rounds, obtains an element from a domain $X$ and outputs a real number.
Algorithm $\AAA$ is said to be a {\em strong tracker} for a function $\FFF:X^*\rightarrow\R$ with accuracy $\alpha$, failure probability $\delta$, and stream length $m$  if the following holds for every sequence $\vec{u}=(u_1,\dots,u_m)\in X^m$. Consider an execution of $\AAA$ on the input stream $\vec{u}$, and denote the answers given by $\AAA$ as $\vec{z}=(z_1,\dots,z_m)$. Then,
$$
\Pr\left[\forall i\in[m]:\; z_i\in(1\pm\alpha)\cdot\FFF(u_1,\dots,u_i) \right]\geq1-\delta,
$$
where the probability is taken over the coins of algorithm $\AAA$.
\end{definition}

While Definition~\ref{def:obliv} is certainly not the only possible definition of streaming algorithms, it is rather standard. Note that in this definition we assume that the input stream $\vec{u}$ is fixed in advance. In particular, we assume that the choice for the elements in the stream is {\em independent} from the internal randomness of $\AAA$. This assumption is crucial for the analysis (and correctness) of many of the existing streaming algorithms. We refer to algorithms that utilize this assumption as {\em oblivious} streaming algorithms. 
In this work we are interested in the setting where this assumption does not hold, often called the {\em adversarial setting}.

\subsection{The adversarial streaming model}

The adversarial streaming model, in various forms, was considered by~\cite{MironovNS11,GHRSW12,GHSWW12,AhnGM12,AhnGM12b,HardtW13,BenEliezerY19,ben2020framework,hassidim2020adversarially,kaplan2021separating,BravermanHMSSZ21,Cohen0NSSS22,CohenNSS22}. We give here the formulation presented by Ben-Eliezer et al.~\cite{ben2020framework}. 
The adversarial setting is modeled by a two-player game between a (randomized) \texttt{StreamingAlgorithm} and an \texttt{Adversary}. At the beginning, we fix a function $\FFF:X^*\rightarrow\R$. Then the game proceeds in rounds, where in the $i$th round:

\begin{enumerate}
	\item The \texttt{Adversary} chooses an update $u_i\in X$ for the stream, which can depend, in particular, on all previous stream updates and outputs of \texttt{StreamingAlgorithm}.
	\item The \texttt{StreamingAlgorithm} processes the new update $u_i$ and outputs its current response $z_i\in\R$.
\end{enumerate}

The goal of the \texttt{Adversary} is to make the \texttt{StreamingAlgorithm} output an incorrect response
$z_i$ at some point $i$ in the stream. For example, in the distinct elements problem,
the adversary's goal is that at some step $i$, the estimate $z_i$ will fail to be a $(1+\alpha)$-approximation
of the true current number of distinct elements.

In this work we present a new framework for transforming an oblivious streaming algorithm into an adversarially-robust streaming algorithm. Before presenting our framework, we first elaborate on the existing literature and the currently available frameworks.

\subsection{Existing framework: Ben-Eliezer et al.~\texorpdfstring{\cite{ben2020framework}}{BJWY20}}

To illustrate the results of~\cite{ben2020framework}, let us consider the {\em distinct elements} problem, in which the function $\FFF$ counts the number of distinct elements in the stream. Observe that, assuming that there are no deletions in the stream, this quantity is monotonically increasing. Furthermore, since we are aiming for a multiplicative error, the number of times we need to modify the estimate we release is quite small (it depends logarithmically on the stream length $m$). Informally, the idea of~\cite{ben2020framework} is to run several independent copies of an oblivious algorithm (in parallel), and to use each copy to release answers over a part of the stream during which the estimate remains constant. In more detail, the generic transformation of~\cite{ben2020framework} (applicable not only to the distinct elements problem) is based on the following definition. 

\begin{definition}[Flip number \cite{ben2020framework}]
Given a function $\FFF$, the {\em $(\alpha,m)$-flip number} of $\FFF$, denoted as $\lambda_{\alpha,m}(\FFF)$, or simply $\lambda$ in short, is the maximal number of times that the value of $\FFF$ can change (increase or decrease) by a factor of $(1+\alpha)$ during a stream of length $m$.
\end{definition}

\begin{remark}
In the technical sections of this work, we sometimes refer to the {\em flip number of the given stream} (w.r.t.\ the target function), which is a more fine-tuned quantity.
\end{remark}

\begin{example}
Assuming that there are no deletions in the stream (a.k.a.\ the {\em insertion only model}), the {\em $(\alpha,m)$-flip number} of the distinct elements problem is at most $O\left( \frac{1}{\alpha} \log m \right)$. However, if deletions are allowed (a.k.a.\ the {\em turnstile model}), then the flip number of this problem could be as big as $\Omega(m)$.
\end{example}

The generic construction of~\cite{ben2020framework} for a function $\FFF$ is as follows. 
\begin{enumerate}
	\item Instantiate $\lambda$ independent copies of an oblivious streaming algorithm for the function $\FFF$, and set $j=1$.
	\item When the next update $u_i$ arrives:
	\begin{enumerate}
		\item Feed $u_i$ to {\em all} of the $\lambda$ copies.
		\item Release an estimate using the $j$th copy (rounded to the nearest power of $(1+\alpha)$). If this estimate is different than the previous estimate, then set $j\leftarrow j+1$.
	\end{enumerate}
\end{enumerate}
Ben-Eliezer et al.~\cite{ben2020framework} showed that this can be used to transform an oblivious streaming algorithm for $\FFF$ into an adversarially robust streaming algorithm for $\FFF$. In addition, the overhead in terms of memory is only $\lambda$, which is small in many interesting settings.  

The simple, but powerful, observation of Ben-Eliezer et al.~\cite{ben2020framework}, is that by ``using every copy at most once'' we can break the dependencies between the internal randomness of our algorithm and the choice for the elements in the stream. Intuitively, this holds because the answer is always computed using a ``fresh copy'' whose randomness is independent from the choice of stream items. 
\subsection{Existing framework: Hassidim et al.~\texorpdfstring{\cite{hassidim2020adversarially}}{HKM+20}}
Hassidim et al.~\cite{hassidim2020adversarially} showed that, in fact, we can use every copy of the oblivious algorithm much more than once. In more detail, the idea of Hassidim et al.\ is to protect the internal randomness of each of the copies of the oblivious streaming algorithm using {\em differential privacy}~\cite{dwork2006calibrating}. Hassidim et al.\ showed that this still suffices in order to break the dependencies between the internal randomness of our algorithm and the choice for the elements in the stream.
This resulted in an improved framework where the space blowup is only $\sqrt{\lambda}$ (instead of $\lambda$). Informally, the framework of \cite{hassidim2020adversarially} is as follows.

\begin{enumerate}
	\item Instantiate $\sqrt{\lambda}$ independent copies of an oblivious streaming algorithm for the function $\FFF$.
	\item When the next update $u_i$ arrives:
	\begin{enumerate}
		\item Feed $u_i$ to {\em all} of the $\sqrt{\lambda}$ copies.
		\item Aggregate all of the estimates given by the $\sqrt{\lambda}$ copies, and compare the aggregated estimate to the previous estimate. If the estimate had changed ``significantly'', output the new estimate. Otherwise output the previous output. 
	\end{enumerate}
\end{enumerate}

In order to efficiently aggregate the estimates in Step 2b, this framework crucially relied on the fact that all of the copies of the oblivious algorithm are ``the same'' in the sense that they compute (or estimate) exactly the same function of the stream. This allowed Hassidim et al.\ to efficiently aggregate the returned estimates using standard tools from the literature on differential privacy. The intuition is that differential privacy allows us to identify global properties of the data, and hence, aggregating several numbers (the outcomes of the different oblivious algorithms) is easy if they are very similar.

%%%%%%%%%%%%%%%%%%%%%%%%%%%%%%
% Existing framework: WZ
%%%%%%%%%%%%%%%%%%%%%%%%%%%%%%
\subsection{Existing framework: Woodruff and Zhou~\texorpdfstring{\cite{woodruff2020tight}}{WZ21}}\label{sec:intro_WZ}
Woodruff and Zhou \cite{woodruff2020tight} presented an adversarial streaming framework that builds on the framework of Ben-Eliezer at el.~\cite{ben2020framework}. 
The new idea of \cite{woodruff2020tight} is that, in many interesting cases, the oblivious algorithms we execute can be modified to track {\em different} (but related) functions, that require less space while still allowing us to use (or combine) several of them at any point in time in order to estimate $\FFF$.

To illustrate this, consider a part of the input stream, say from time $t_1$ to time $t_2$, during which the target function $\FFF$ doubles its value and is monotonically increasing. More specifically, suppose that we already know (or have a good estimation for) the value of $\FFF$ at time $t_1$, and we want to track the value of $\FFF$ from time $t_1$ till $t_2$. Recall that in the framework of \cite{ben2020framework} we only modify our output once the value of the function has changed by more than a $(1+\alpha)$ factor. As $\FFF(t_2)\leq2\cdot\FFF(t_1)$, we get that between time $t_1$ and $t_2$ there are roughly $1/\alpha$ time points at which we need to modify our output. 
In the framework of \cite{ben2020framework}, we need a fresh copy of the oblivious algorithm for each of these $1/\alpha$ time points. For concreteness, let us assume that every copy uses space $1/\alpha^2$ (which is the case if, e.g., $\FFF=F_2$), and hence the framework of \cite{ben2020framework} requires space $1/\alpha^3$ to track the value of the function $\FFF$ from $t_1$ till $t_2$. 

In the framework of \cite{woodruff2020tight}, on the other hand, this will cost only $1/\alpha^2$. We now elaborate on this improvement. As we said, from time $t_1$ till $t_2$ there are $1/\alpha$ time points on which we need to modify our output. Let us denote these time points as $t_1 = w_0 < w_1<w_2<\dots<w_{1/\alpha} = t_2$.\footnote{Note that these time points are not known to the algorithm in advance. Rather, the algorithm needs to discover them ``on the fly''. To simplify the presentation, in Section~\ref{sec:intro_WZ} we assume that these time points are known in advance.
} In the framework of \cite{woodruff2020tight}, the oblivious algorithms we execute are tracking {\em differences} between the values of $\FFF$ at specific times, rather than tracking the value of $\FFF$ directly. (These algorithms are called {\em difference estimators}, or DE in short.) In more detail, suppose that for every $j\in\{0,1,2,3,\dots,\log\frac{1}{\alpha}\}$ and every $i\in\{2^j, 2{\cdot}2^j, 3{\cdot}2^j, 4{\cdot}2^j,\dots, \frac{1}{\alpha}\}$ we have an oblivious algorithm (a {\em difference estimator}) for estimating the value of $[\FFF(w_i) - \FFF(w_{i-2^j})]$. We refer to the index $j$ as the {\em level} of the oblivious algorithm. So there are $\log\frac{1}{\alpha}$ different levels, where we have a different number of oblivious algorithms for each level. (For level $j=0$ we have $1/\alpha$ oblivious algorithms and for level $j=\log\frac{1}{\alpha}$ we have only a single oblivious algorithm.)

Note that given all of these oblivious algorithms, we could compute an estimation for the value of the target function $\FFF$ at each of the time points $w_1,\dots,w_{1/\alpha}$ (and hence for every time $t_1\leq t\leq t_2$) by summing the estimations of (at most) one oblivious algorithm from each level.\footnote{Specifically, in order to reach the estimated value of $\mathcal{F}$ at time $w_{\tilde{t}}$ one can add the estimations of difference estimators of levels corresponds to the binary representation of $\tilde{t}$. That is, at most one of each level $j$.}  For example, an estimation for the value of $\FFF\left(w_{\frac{3}{4\alpha}+1}\right)$ can be obtained by combining estimations as follows:
$$
\FFF\left(w_{\frac{3}{4\alpha}+1}\right) = \FFF\left(w_0\right) + \left[ \FFF\left(w_{\frac{1}{2\alpha}}\right) {-} \FFF\left(w_0\right) \right] + \left[ \FFF\left(w_{\frac{3}{4\alpha}}\right) {-} \FFF\left(w_{\frac{1}{2\alpha}}\right) \right] + \left[ \FFF\left(w_{\frac{3}{4\alpha}+1}\right) {-} \FFF\left(w_{\frac{3}{4\alpha}}\right) \right].
$$

As we sum at most $\log\frac{1}{\alpha}$ estimations, this decomposition increases our estimation error only by a factor of $\log\frac{1}{\alpha}$, which is acceptable. The key observation of \cite{woodruff2020tight} is that the space complexity needed for an oblivious algorithm at level $j$ decreases when $j$ decreases (intuitively because in lower  levels we need to track smaller differences, which is easier). So, even though in level $j{=}10$ we have more oblivious algorithms than in level $20$, these oblivious algorithms are cheaper than in level $20$ such that the overall space requirements for levels $j{=}10$ and level $j{=}20$ (or any other level) is the same. Specifically, \cite{woodruff2020tight} showed that (for many problems of interest, e.g., for $F_2$) the space requirement of a difference estimator at level $j$ is $O( 2^{j} /\alpha  )$. 
We run  $O(2^{-j} / \alpha)$ oblivious algorithms for level $j$, and hence, the space needed for level $j$ is $O(2^{-j}/\alpha \cdot 2^{j}/\alpha)=O(1/\alpha^{2})$. As we have $\log(1/\alpha)$ levels,
the overall space we need to track the value of $\FFF$ from time $t_1$ till $t_2$ is $\tilde{O}(1/\alpha^{2})$. This should be contrasted with the space required by \cite{ben2020framework} for this time segment, which is $O(1/\alpha^3)$.

\subsection{Our results}

The framework of \cite{woodruff2020tight} is very effective for the insertion-only model. However, there are two challenges that need to be addressed in the turnstile setting: (1) We are not aware of non-trivial constructions for difference estimators in the turnstile setting, and hence, the framework of \cite{woodruff2020tight} is not directly applicable to the turnstile setting.\footnote{Moshe Shechner and Samson Zhou. Personal communication, 2022.} (2) Even assuming the existence of a non-trivial difference estimator, the framework of \cite{woodruff2020tight} obtains sub-optimal results in the turnstile setting.

To overcome the first challenge, we introduce a new {\em monitoring technique}, that aims to identify time steps at which we cannot guarantee correctness of our difference estimators (in the turnstile setting), and reset the system at these time steps. This will depend on the specific application at hand (the target function) and hence, we defer the discussion on our monitoring technique to Section~\ref{sec:Applications} where we discuss applications of our framework.

We now focus on the second challenge (after assuming the existence of non-trivial difference estimators).
To illustrate the sub-optimality of the framework of \cite{woodruff2020tight}, let us consider a simplified turnstile setting in which the input stream can be partitioned into $k$ time segments during each of which the target function is monotonic, and increases (or decreases) by at most a factor of 2 (or 1/2). Note that $k$ can be  very large in the turnstile model (up to $O(m)$). 
With the framework of \cite{woodruff2020tight}, we would need space $\tilde{O}\left( \frac{k}{\alpha^2} \right)$ to track the value of $F_2$ throughout such an input stream. The reason is that, like in the framework of \cite{ben2020framework}, the robustness guarantees are achieved by making sure that every oblivious algorithm is ``used only once''. This means that we cannot reuse the oblivious algorithms across the different segments, and hence, the space complexity of \cite{woodruff2020tight} scales linearly with the number of segments $k$. 

To mitigate this issue, we propose a new construction that combines the frameworks of \cite{woodruff2020tight} with the framework of \cite{hassidim2020adversarially}. Intuitively, in our simplified example with the $k$ segments, we want to reuse the oblivious algorithms across different segments, and protect their internal randomness with differential privacy to ensure robustness.
However, there is an issue here. Recall that the framework of \cite{hassidim2020adversarially} crucially relied on the fact that all of the copies of the oblivious algorithm are ``the same'' in the sense that they compute the same function exactly. This allowed \cite{hassidim2020adversarially} to efficiently aggregate the estimates in a differentially private manner. However, in the framework of \cite{woodruff2020tight}, the oblivious algorithms we maintain are fundamentally {\em different} from each other, tracking different functions. Specifically, every difference estimator is tracking the value of $[\FFF(t)-\FFF(e)]$ for a {\em unique enabling time} $e<t$ (where $t$ denotes the current time). That is, every difference estimator necessarily has a different enabling time, and hence, they are not tracking the same function, and it is not clear how to aggregate their outcomes with differential privacy.

\paragraph{Toggle Difference Estimator (TDE).} To overcome the above challenge, we present an extension to the notion of a difference estimator, which we call a {\em Toggle Difference Estimator} (see Definition \ref{def:TDE}). Informally, a toggle difference estimator is a difference estimator that allows us to modify its {\em enabling time} on the go. This means that a TDE can track, e.g., the value of $[\FFF(t)-\FFF(e_1)]$ for some (previously given) enabling time $e_1$, and then, at some later point in time, we can instruct the same TDE to track instead the value of $[\FFF(t)-\FFF(e_2)]$ for some other enabling time $e_2$. We show that this extra requirement from the difference estimator comes at a very low cost in terms of memory and runtime. Specifically, in Section~\ref{sec:TDE} we present a generic (efficiency preserving) method for generating a TDE from a DE.

\medskip

Let us return to our example with the $k$ segments. Instead of using every oblivious algorithm only once, we reuse them across the different segments, where during any single segment all the TDE's are instructed to track the appropriate differences that are needed for the current segment. This means that during every segment we have many copies of the ``same'' oblivious algorithm. More specifically, for every different {\em level} (as we explained above) we have many copies of an oblivious algorithm for that level, which is (currently) tracking the difference that we need. This allows our space complexity to scale with $\sqrt{k}$ instead of linearly with $k$ as the framework of \cite{woodruff2020tight}. To summarize this discussion, our new notion of TDE allows us to gain both the space saving achieved by differential privacy (as in the framework of \cite{hassidim2020adversarially}) and the space saving achieved by tracking the target function via differences (as in the  framework of \cite{woodruff2020tight}).

\begin{remark}
The presentation given above (w.r.t.\ our example with the $k$ segments) is oversimplified. Clearly, in general, we have no guarantees that an input (turnstile) stream can be partitioned into $k$ such segments. This means that in the actual construction we need to calibrate our TDE's across time segments in which the value of the target function is {\em not} monotone. See Section~\ref{sec:constructionOverview} for a more detailed overview of our construction and the additional modifications we had to introduce.
\end{remark}

We are now ready to state our main result (for the formal statement see Theorem \ref{thm:FrameworkSpace}). We present a framework for adversarial streaming for turnstile streams with bounded flip number $\lambda$, for any function $\mathcal{F}$ for which the following algorithms exist:
\begin{enumerate}
	\item An $\alpha$-accurate oblivious streaming algorithm $\mathsf{E}$ with space complexity $\text{Space}(\mathsf{E})$.
	\item An oblivious TDE streaming algorithm $\mathsf{E}_{\TDE}$ (satisfying some conditions).
\end{enumerate}

Under these conditions, our framework results in an $O(\alpha)$-accurate adversarially-robust algorithm with space\footnote{Here $\tilde{O}$ stands for omitting poly-logarithmic factors of $\lambda, \alpha^{-1}, \delta^{-1}, n, m$.}
$\tilde{O}\left(\sqrt{\alpha\cdot \lambda}\cdot \text{Space}(\mathsf{E}) \right)$. 
In contrast, under the same conditions, the framework of \cite{woodruff2020tight} requires space $\tilde{O}\left(\alpha \cdot \lambda \cdot \text{Space}(\mathsf{E})\right)$.

As we mentioned, we are not aware of non-trivial constructions for difference estimators that work in the turnstile setting. Nevertheless, in Section \ref{sec:Applications} we show that our framework is applicable to the problem of estimating $F_2$ (the second moment of the stream). To this end, we introduce the following notion that allows us to control the number of times we need to reset our system (which happens when we cannot guarantee correctness of our difference estimators).

\begin{definition}[Twist number]\label{def:twistIntro} 
The {\em $(\alpha,m)$-twist number} of a stream $\SSS$ w.r.t.\ a functionality $\FFF$, denoted as $\TN_{\alpha,m}(\SSS)$, is the maximal $\mu\in[m]$ such that $\SSS$ can be partitioned into $2\mu$ disjoint segments $\SSS = \PPP_0 \circ \VVV_0 \circ\dots \circ \PPP_{\TN-1} \circ \VVV_{\TN-1}$ (where $\{\PPP_i\}_{i\in[\TN]}$ may be empty) s.t.\ for every $i\in[\mu]$:
\begin{enumerate}
    \item $\FFF(\VVV_i) > \alpha \cdot \FFF(\PPP_0 \circ \VVV_0 \circ \dots\circ \VVV_{i-1} \circ \PPP_i)$
    \item $|\FFF(\PPP_0 \circ\VVV_0 \dots \circ \PPP_i \circ \VVV_i) - \FFF(\PPP_0\circ\VVV_0 \circ \dots \circ \PPP_i)| \leq \alpha\cdot \FFF(\PPP_0 \circ \VVV_0 \circ \dots \circ \PPP_i)$
\end{enumerate}
\end{definition}

Intuitively, a stream has twist number $\mu$ if there are $\mu$ disjoint segments $\VVV_0,\dots,\VVV_{\mu-1}\subseteq\SSS$ such that the value of the function on each of these segments is large (Condition 1), but still these segments do not change the value of the function on the entire stream by too much (Condition 2). In Section \ref{sec:Applications} we leverage this notation and present the following result for $F_2$ estimation.

\begin{theorem}[$F_2$ Robust estimation, informal]\label{thm:F2SpaceInformal}
There exists an adversarially robust $F_2$ estimation algorithm for turnstile streams of length $m$ with a bounded  $(O(\alpha), m)$-flip number $\lambda$ and a bounded $(O(\alpha), m)$-twist number $\mu$ that guarantees $\alpha$-accuracy (w.h.p.)\ using space complexity 
$\tilde{O}\left(\frac{\sqrt{\alpha\lambda+\mu}}{\alpha^{2}}\right)$.

\end{theorem}

This should be contrasted with the result of \cite{hassidim2020adversarially}, who obtain space complexity  $\tilde{O}\left(\frac{\sqrt{\lambda}}{\alpha^2}\right)$ for robust $F_2$ estimation in the turnstile setting. Hence, our new result is better whenever $\mu\ll\lambda$. 

\begin{example}
For $F_2$ estimation in {\em insertion-only} streams, it holds that $\mu=0$ even though $\lambda$ can be large. This is the case because, in insertion only streams, Conditions 1 and 2 from Definition~\ref{def:twistIntro} cannot hold simultaneously. Specifically, denote $p=\PPP_0\circ\dots\circ\PPP_i$ and $v=\VVV_i$, and suppose that Condition 2 holds, i.e., $\|p\circ v\|^2-\|p\|^2\leq\alpha\cdot\|p\|^2$. Hence, in order to show that Condition 1 does {\em not} hold, it suffices to show that $\|v\|^2\leq\|p\circ v\|^2-\|p\|^2$, i.e., show that $\|v\|^2+\|p\|^2\leq\|p\circ v\|^2$,  i.e., show that
$(v_1^2+p_1^2)+\dots+(v_n^2+p_n^2)\leq(v_1+p_1)^2+\dots+(v_n+p_n)^2$, which trivially holds whenever $p_i,v_i\geq0$.
\end{example}

%%%%%%%%%%%%%%%%%%%%%%%%%%%%%%
% Other related works
%%%%%%%%%%%%%%%%%%%%%%%%%%%%%%
\subsection{Other related works}

Related to our work is the line of work on {\em adaptive data analysis}, aimed at designing tools for guaranteeing statistical validity in cases where the data is being accessed adaptively \cite{DworkFHPRR15,BassilyNSSSU21,JungLN0SS20,HardtU14,SteinkeU15,NissimSSSU18,NissimS19,ShenfeldL19,abs-2106-10761,KSS22}. 
Recall that the difficulty in the adversarial streaming model arises from potential dependencies between the inputs of the algorithm and its internal randomness. As we mentioned, our construction builds on a technique introduced by~\cite{hassidim2020adversarially} for using differential privacy to protect not the input data, but rather the internal randomness of algorithm. Following \cite{hassidim2020adversarially}, this technique was also used by \cite{gupta2021adaptive,BKMNSS22} for designing robust algorithms in other settings.

%%%%%%%%%%%%%%%%%%%%%%%%%%%%%%%%%%%
% Preliminaries and Definitions
%%%%%%%%%%%%%%%%%%%%%%%%%%%%%%%%%%%
\section{Preliminaries} \label{sec:preliminaries}
In this work we consider input streams which are represented as a sequence of {\em updates}, where every {\em update} is a tuple containing an element (from a finite domain) and its (integer) weight. Formally,
\begin{definition}[Turnstile stream]
A stream of length $m$ over a domain $[n]$,\footnote{For an integer $n\in \N$ denote $[n]=\{0,1,\dots,n-1\}$ (that is $|[n]| = n$).} consists of a sequence of updates $\langle s_0, \Delta_0 \rangle ,\dots, \langle s_{m-1}, \Delta_{m-1} \rangle$ where $s_i\in [n]$ and $\Delta_i \in \Z$.
Given a stream $\mathcal{S}\in([n]\times\Z)^m$ and integers $0\leq t_1\leq t_2\leq m-1$, we write $\mathcal{S}^{t_1}_{t_2} = (\langle s_{t_1}, \Delta_{t_1} \rangle ,\dots, \langle s_{t_2}, \Delta_{t_2} \rangle)$ to denote the sequence of updates from time $t_1$ till $t_2$. We also use the abbreviation $\mathcal{S}_{t}=\mathcal{S}^{1}_{t}$ to denote the first $t$ updates.
\end{definition}
Let $\mathcal{F}:([n]\times\Z)^{*}\rightarrow \R$ be a function (for example $\mathcal{F}$ might count the number of distinct elements in the stream).  At every time step $t$, after obtaining the next element in the stream $\langle s_t, \Delta_t \rangle$, our goal is to output an approximation for $\mathcal{F}(\mathcal{S}_t)$. To simplify presentation we also denote $\mathcal{F}(t)=\mathcal{F}(\mathcal{S}_t)$ for $t\in [m]$. We assume throughout the paper that $\log(m)=\Theta(\log(n))$ and that $\mathcal{F}$ is bounded polynomially in $n$. 

In Section \ref{sec:introduction}, for the purpose of presentation, it was useful to refer to the quantity a {\em flip number of a function}. Our results are stated w.r.t a more refined quantity: a {\em flip number of a stream}.
\begin{definition}[Flip number of a stream \cite{ben2020framework}]
Given a function $\FFF$ and a stream $\SSS$ of length $m$, the {\em $(\alpha,m)$-flip number} of $\SSS$, denoted as $\lambda_{\alpha}(\SSS)$,  is the maximal number of times that the value of $\FFF$ can change (increase or decrease) by a factor of $(1+\alpha)$ during the stream $\SSS$.
\end{definition}

\paragraph{Toggle Difference Estimator.} For the purpose of our framework, we present an extension to the notion of a {\em difference estimator (DE)} from \cite{woodruff2020tight}, which we call a {\em toggle difference estimator (TDE)}. A difference estimator for a function $\mathcal{F}$ is an oblivious streaming algorithm, defined informally as follows: The difference estimator is initiated on time $t = 1$ and has a dynamically defined {\em enabling time} $1\leq e\leq m$. Once that enabling time is set, the difference estimator outputs an estimation for $\left(\mathcal{F}(\mathcal{S}_t) - \mathcal{F}(\mathcal{S}_e)\right)$ for all times $t>e$ (provided some conditions on that difference). That is, once the difference estimator's {\em enabling time} is set, it cannot be changed. And so, if an estimation is needed for some other {\em enabling time}, say $e^{\prime} \neq e$, then an additional instance of a difference estimator is needed. Our framework requires from such an estimator to be able to provide estimations for multiple {\em enabling times}, as long as the estimation periods do not overlap. This is captured in the following definition.

\begin{definition}[Toggle Difference Estimator]\label{def:TDE}
Let $\mathcal{F}:([n]\times\Z)^{*}\rightarrow \R$ be a function, and let $m,p\in \N$ and  $\gamma,\alpha, \delta\in(0,1)$ be parameters. Let $\mathsf{E}$ be an algorithm with the following syntax. In every time step $t
\in[m]$,  algorithm $\mathsf{E}$ obtains an update $\langle s_t, \Delta_t , b_t \rangle\in([n]\times\Z \times\{0,1\})$ and outputs a number $z_t$. Here $\langle s_t, \Delta_t \rangle$ denotes the current update, and $b_t$ is an indicator for when the current time $t$ should be considered as the {\em new enabling time}. We consider input streams $\mathcal{S}\in([n]\times\Z \times\{0,1\})^m$ such that there are at most $p$ time steps $t$ for which $b_t=1$, and denote these time steps as  $1\leq e^1< e^2< \dots< e^{p}<m$. Also, for a time step $t\in[m]$ we denote $e(t)=\max\{ e^i : e^i\leq t  \}$.

Algorithm $\mathsf{E}$ is a $(\gamma,\alpha, p, \delta)$-{\em toggle difference estimator} for $\mathcal{F}$ if the following holds for every such input stream $\mathcal{S}$. With probability at least $1-\delta$, for every $t\in[m]$ such that
\begin{equation}
    |\mathcal{F}(\mathcal{S}_{t}) - \mathcal{F}(\mathcal{S}_{e(t)}) | \leq \gamma \cdot \mathcal{F}(\mathcal{S}_{e(t)}) \label{req:TDEdiffRange}
\end{equation}
the algorithm outputs a value $z_t$ such that $z_t\in \left(\mathcal{F}(\mathcal{S}_{t}) - \mathcal{F}(\mathcal{S}_{e(t)})\right) \pm \alpha \cdot \mathcal{F}(\mathcal{S}_{e(t)})$.
\end{definition}
This definition generalizes the notion of a difference estimator (DE) from \cite{woodruff2020tight}, in which $p=1$. In Section~\ref{sec:TDE} we show that this extension comes at a very low cost in terms of the space complexity. Note that on times $t$ s.t.\ the requirements specified w.r.t.\ $\gamma$ do not hold, there is no accuracy guarantee from the TDE algorithm.

%%%%%%%%%%%%%%%%%%%%%%%%%%%%
% Preliminaries from DP
%%%%%%%%%%%%%%%%%%%%%%%%%%%%
\subsection{Preliminaries from Differential Privacy}\label{sec:prelimsDP}
Differential privacy \cite{dwork2006calibrating} is a mathematical definition for privacy that aims to enable statistical analyses of databases while providing strong guarantees that individual-level information does not leak. Consider an algorithm $\mathcal{A}$ that operates on a database in which every row represents the data of one individual. Algorithm $\mathcal{A}$ is said to be {\em differentially private} if its outcome distribution is insensitive to arbitrary changes in the data of any single individual. Intuitively, this means that algorithm $\mathcal{A}$ leaks very little information about the data of any single individual, because its outcome would have been distributed roughly the same even if without the data of that individual. Formally,
%%%%%%%%%%%%%%%%%%%%%%%%%%%%
% Differential Privacy def
%%%%%%%%%%%%%%%%%%%%%%%%%%%%
\begin{definition}[\cite{dwork2006calibrating}]Let $\mathcal{A}$ be a randomized algorithm that operates on databases. Algorithm $\mathcal{A}$ is $(\eps,\delta)$-{\em differentially private} if for any two databases $S,S^{\prime}$ that differ on one row, and any event $T$, we have
$$
\Pr \left[ \mathcal{A}(S)\in T \right] \leq e^{\eps} \cdot \Pr \left[ \mathcal{A}(S^{\prime})\in T \right] + \delta.
$$
\end{definition}

See Appendix~\ref{sec:MoreDP} for additional preliminaries on differential privacy.

%%%%%%%%%%%%%%%%%%%%%%%%%%%%%%%%%%%%%%%%%%%%%%%%%%%%%%%
% Main result - A Framework for Adversarial Streaming
%%%%%%%%%%%%%%%%%%%%%%%%%%%%%%%%%%%%%%%%%%%%%%%%%%%%%%%
\section{A Framework for Adversarial Streaming}\label{sec:turnstileFramework}

Our transformation from an oblivious streaming algorithm ${\mathsf{E}}_{\ST}$ for a function $\mathcal{F}$ into an adversarially robust algorithm requires the following two conditions.
\begin{enumerate}
    \item The existence of a toggle difference estimator ${\mathsf{E}}_{\TDE}$ for $\mathcal{F}$, see Definition~\ref{def:TDE}.\label{con:existDEforF}
    \item Every single update can change the value of $\mathcal{F}$ up to a factor of $(1\pm \alpha^{\prime})$ for some $\alpha^{\prime}=O(\alpha)$. Formally, throughout the analysis we assume that for every stream $\mathcal{S}$ and for every update $u=\langle s, \Delta\rangle$ it holds that\label{con:streamMaxUpdateSize}
$$
(1-\alpha^{\prime})\mathcal{F}(\mathcal{S}) \leq \mathcal{F}(\mathcal{S},u) \leq (1+\alpha^{\prime})\mathcal{F}(\mathcal{S}).
$$
\end{enumerate}

\begin{remark}
These conditions are identical 
to the conditions required by \cite{woodruff2020tight}. Formally, they require only a difference estimator instead of a toggle difference estimator, but we show that these two objects are equivalent. See Section~\ref{sec:TDE}. 
\end{remark}

\begin{remark}
Condition 2 can be met for many functions of interest, by applying our framework on portions of the stream during which the value of the function is large enough. For example, when estimating $L_2$ with update weights $\pm1$, whenever the value of the function is at least $\Omega(1/\alpha)$, a single update can increase the value of the function by at most a $(1+\alpha)$ factor. Estimating $L_2$ whenever the value of the function is smaller than $O(1/\alpha)$ can be done using an existing (oblivious) streaming algorithm with error $\rho=O(\alpha)$. To see that we can use an oblivious algorithm in this setting, note that the additive error of the oblivious streaming algorithm is at most $O(\frac{\rho}{\alpha})\ll1$. Hence, by rounding the answers of the oblivious algorithm we ensure that its answers are {\em exactly} accurate (rather than approximate). As the oblivious algorithm returns exact answers in this setting, it must also be adversarially robust.
\end{remark}

%%%%%%%%%%%%%%%%%%%%%%%%%%%%%%%%%%%%%%%
% Construction Overview
%%%%%%%%%%%%%%%%%%%%%%%%%%%%%%%%%%%%%%%

\subsection{Construction Overview}\label{sec:constructionOverview}

Our construction builds on the constructions of \cite{woodruff2020tight} and \cite{hassidim2020adversarially}. At a high level, the structure of our construction is similar to that of \cite{woodruff2020tight}, but our robustness guarantees are achieved using differential privacy, similarly to \cite{hassidim2020adversarially}, and using our new concept of TDE.

Our algorithm can be thought of as operating in {\em phases}. In the beginning of every {\em phase}, we aggregate the estimates given by our strong trackers with differential privacy, and ``freeze'' this aggregated estimate as the base value for the rest of the phase. Inside every phase, 
we privately aggregate (and ``freeze'') estimates given by our TDE's. 
More specifically, throughout the execution we aggregate TDE's of different types/levels (we refer to the level that is currently being aggregated as the {\em active} level).
At any point in time we estimate the (current) value of the target function by summing specific ``frozen'' differences together with the base value.

We remark that, in addition to introducing the notion of TDE's, we had to incorporate several modifications to the framework of \cite{woodruff2020tight} in order to make it compatible with our TDE's and with differential privacy. In particular, \cite{woodruff2020tight} manages phases by placing fixed thresholds (powers of 2) on the value of the target function; starting a new phase whenever the value of the target function crosses the next power of 2. If, at some point in time, the value of the target function drops below the power of 2 that started this phase, then this phase ends, and they go back to the previous phase. This is possible in their framework because the DE's of the previous phase still exist in memory and are ready to be used. In our framework, on the other hand, we need to share all of the TDE's across the different phases, and we cannot go back to ``TDE's of the previous phase'' because these TDE's are now tracking other differences. We overcome this issue by modifying the way in which differences are combined inside each phase.

In Algorithm~\ref{alg:PseudoRobustDE} we present a simplified version of our main construction, including inline comments to improve readability. The complete construction is given in Algorithm~\ref{alg:ROEF}.

%%%%%%%%%%%%%%%%%%%%%%%%%%%%%%%%%%%%%%%%%%
% Main algorithm: Simplified presentation
%%%%%%%%%%%%%%%%%%%%%%%%%%%%%%%%%%%%%%%%%%
\begin{algorithm}

\caption{Simplified presentation of \texttt{RobustDE}}
\label{alg:PseudoRobustDE}

{\bf Input:} Stream $\mathcal{S} = \{\langle s_t, \Delta_t\rangle \}_{t\in [m]}$, accuracy parameter $\alpha$, and a bound on the flip number $\lambda$.\\[-10px]

{\bf Estimators used:} Strong tracker {\em $\mathsf{E}_{\ST}$} 
and toggle-difference-estimator $\mathsf{E}_{\TDE}$ for $\mathcal{F}$.

\hdashrule[0.5ex][x]{\linewidth}{0.5pt}{1.5mm}  

{\bf Initialization.} Let $\beta = \lceil \log(\alpha^{-1}) \rceil$ denote the number of levels.
    For every level $j\in \{0, \dots, \beta - 1 \}$ initialize $O(\sqrt{2^{-j}\lambda})$ copies of the TDE algorithm $\mathsf{E}_{\TDE}$ with parameter $\gamma=O(2^{j}\alpha)$ and set their initial estimation to be $\mathsf{Z}_{j}=0$. Also initialize $O(\sqrt{\alpha\lambda})$ copies of the strong tracker $\mathsf{E}_{\ST}$, and set their initial estimation to be $\mathsf{Z}_{\ST}=0$. 
    (We denote $j=\beta$ for the level of the strong trackers, and denote $\mathsf{Z}_{\beta}=\mathsf{Z}_{\ST}$.)  Initialize counter $\tau$.
    
\hdashrule[0.5ex][x]{\linewidth}{0.5pt}{1.5mm}    

\begin{enumerate}[leftmargin=18pt,rightmargin=10pt,itemsep=1pt,topsep=3pt]
\item[{\color{gray} \%}] {\color{gray} As the execution progresses, we update the variables $\mathsf{Z}_{j}$, in which we maintain aggregations of the estimates given by our strong trackers and difference estimators. We make sure that, at any point in time, we can compute an estimation for the target function by carefully combining the values of these variables. This careful combination is handled using the counter $\tau$.}

\item[{\color{gray} \%}] {\color{gray} It is convenient to partition the time points into {\em phases}, where during each phase the value of $\mathsf{Z}_{\ST}$ remains constant. Intuitively, in the beginning of every phase we compute a very accurate estimation for the target function using $\mathsf{E}_{\ST}$, and then in the rest of the phase, we augment that estimation with weaker estimations given by our TDE's (and accumulate errors along the way). Our error is reset at the beginning of every phase.}

\end{enumerate}

\hdashrule[0.5ex][x]{\linewidth}{0.5pt}{1.5mm}    

For every time step $t\in[m]$: 
\begin{enumerate}
[leftmargin=25pt,rightmargin=10pt,itemsep=1pt,topsep=3pt]
    \item Get the update $\langle s_t, \Delta_t \rangle$ and feed it to all estimators.
	\item Select the relevant estimator level according to $\tau$ into $j$ (where $j=\beta$ in case the relevant level is that of the strong trackers, which happens only if  $\tau=O(1/\alpha)$).

	\item Let $\mathsf{Z}$ denote the sum of previously computed aggregates for levels $i>j$ such that the $i$th bit of $\tau$ is 1 (in binary representation). That is, $\mathsf{Z}=\sum_{i>j:\tau[i]=1} \mathsf{Z}_i$.
	
    \item[{\color{gray} \%}] {\color{gray} Note that if $j=\beta$ then $\mathsf{Z}=0$.}
	
	\item\label{step:simpleSV} Summed with $\mathsf{Z}$, every estimator from level $j$ suggests an estimation for $\FFF(t)$. If the previous output $\Out$ is ``close enough'' to (most of) these suggestions, then goto Step~\ref{algHLStep:output}. Otherwise, modify $\Out$ as follows.
	
	\begin{enumerate} [leftmargin=25pt,rightmargin=10pt,itemsep=1pt,topsep=0pt]
		\item\label{step:simpleMedian} $\mathsf{Z}_{j}\leftarrow$ Differentially private approximation for the median of the outputs given by the estimators at level $j$. \label{algHLStep:UpdateFrozenVals}
		
		\item Re-enable all TDE's at levels $i< j$.
		\item[{\color{gray} \%}] {\color{gray} 
		That is, from now on, all TDE's at levels $i<j$ are estimating differences from future time steps to the current time step.}

		\item If $j = \beta$, then set $\tau \leftarrow 0$. Otherwise set $\tau \leftarrow \tau + 1$.
		
		\item[{\color{gray} \%}] {\color{gray} That is, if $j = \beta$, which happens only if $\tau = O(1/\alpha)$, then we start a new phase.}
		
		\item $\Out\leftarrow \mathsf{Z}+\mathsf{Z}_{j}$.
		
	\end{enumerate}
	\item\label{algHLStep:output} Output $\Out$ 
\end{enumerate}
\end{algorithm}

%%%%%%%%%%%%%%%%%%%%%%%%%%%%%%%%%%%%%%%%%%
% Analysis Overview
%%%%%%%%%%%%%%%%%%%%%%%%%%%%%%%%%%%%%%%%%%
\subsection{Analysis Overview} 
At a high level, the analysis can be partitioned into five components (with one component being significantly more complex then the others). We now elaborate on each of these components. The complete analysis is given in Appendix~\ref{sec:formalAnalysis}.

\subsubsection{First component: Privacy analysis}

In Section~\ref{sec:formalPrivacy} we show that our construction satisfies differential privacy w.r.t.\ the collection of random strings on which the oblivious algorithms operate. Recall that throughout the execution we aggregate (with differential privacy) the outcome of our estimators from the different levels. Thus, in order to show that the whole construction satisfies privacy (using composition theorems) we need to bound the maximal number of times we aggregate the estimates from the different levels. However, we can only bound this number under the assumption that the framework is accurate (in the adaptive setting), and for that we need to rely on the privacy properties of the framework. So there is a bit of circularity here. To simplify the analysis, we add to the algorithm hardcoded caps on the maximal number of times we can aggregate estimates at the different levels. This makes the privacy analysis straightforward. However, we will later need to show that this hardcoded capping ``never'' happens, as otherwise the algorithm fails.\footnote{We remark that as the hardcoded capping ``never'' happens, we can in fact remove it from the algorithm. One way or another, however, we must derive a high probability bound on the number of times we can aggregate estimates at the different levels.}

\subsubsection{Second component: Conditional accuracy}
 In Section~\ref{sec:formalAccuracy} we show that if the following two conditions hold, then the framework is accurate:
\begin{enumerate}[leftmargin=80pt]
    \item[{\bf Condition (1):}] At any time step throughout the execution, at least 80\% of the estimators in every level are accurate (w.r.t.\ the differences that they are estimating).
    \item[{\bf Condition (2):}] The hardcoded capping never happens.
\end{enumerate}
This is the main technical part in our analysis; here we provide an oversimplified overview, hiding many of the technicalities.  
We first show that if Conditions (1) and (2) hold then the framework is accurate. We show this by proving a sequence of lemmas that hold (w.h.p.)\ whenever Conditions~(1) and~(2) hold. We now elaborate on some of these lemmas. Recall that throughout the execution we ``freeze'' aggregated estimates given by the different levels. The following lemma shows that these ``frozen'' aggregations are accurate (at the moments at which we ``freeze'' them). This Lemma follows almost immediately from Condition (1), as if the vast majority of our estimators are accurate, then so is their private aggregation.

\begin{lemma}[informal version of Lemma~\ref{lem:FrozenAccuracy}]\label{lem:simple1}
In every time step $t\in[m]$ in which we compute a value $\mathsf{Z}_j$ (in Step~\ref{step:modifying} of Algorithm~\texttt{RobustDE}, or Step~\ref{step:simpleMedian} of the simplified algorithm) it holds that $\mathsf{Z}_j$ is accurate. Informally, if the current level $j$ is that of the strong trackers, then $|\mathsf{Z}_j - \mathcal{F}(t)| <  \alpha\cdot \mathcal{F}(t)$, and otherwise 
    $|\mathsf{Z}_j - (\mathcal{F}(t) - \mathcal{F}(e_j))| <  \alpha\cdot \mathcal{F}(e_j)$, where $e_j$ is the last {\em enabling time} of level $j$.
\end{lemma}

During every time step $t\in[m]$, we test whether the previous output is still accurate (and modify it if it is not). This test is done by comparing the previous output with (many) suggestions we get for the current value of the target function. These suggestions are obtained by summing the outputs of the estimators at the currently active level $j$ together with a (partial) sum of the previously frozen estimates (denoted as $\mathsf{Z}$). This is done in Step~\ref{algStep:AboveThresh} of Algorithm~\texttt{RobustDE}, or in Step~\ref{step:simpleSV} of the simplified algorithm. The following lemma, which we prove using Lemma~\ref{lem:simple1}, states that the majority of these suggestions are accurate (and hence our test is valid). 

\begin{lemma}[informal version of Lemma~\ref{lem:estimationErrorBound}]
Fix a time step $t\in [m]$, and let $j$ denote the level of active estimators. 
Then, for at least $80\%$ of the estimators in level $j$, summing their output $z$ with $\mathsf{Z}$ is an accurate estimation for the current value of the target function, i.e., 
$
\left|\mathcal{F}(t) - (\mathsf{Z} + z) \right| \leq \alpha \cdot \mathcal{F}(t).
$
\end{lemma}

So, in every iteration we test whether our previous output is still accurate, and our test is valid. Furthermore, when the previous output is not accurate, we modify it to be $(\mathsf{Z} + \mathsf{Z}_j)$, where $\mathsf{Z}_j$ is the new aggregation (the new ``freeze'') of the estimators at level $j$. 
So this modified output is accurate (assuming that the hardcoded capping did not happen, i.e., Condition (2), as otherwise the output is not modified). We hence get the following lemma.

\begin{lemma}[informal version of Lemma~\ref{lem:outputAccuracy}]
In every time step $t\in [m]$ we have
$$
|\Output(t) - \mathcal{F}(t)| \leq \alpha\cdot \mathcal{F}(t).
$$
\end{lemma}

That is, the above lemma shows that our output is ``always'' accurate. Recall, however, that this holds only assuming that Conditions~(1) and~(2) hold.

\subsubsection{Third component: Calibrating to avoid capping}
In Section~\ref{sec:formalCal} we derive a high probability bound on the maximal number of times we will aggregate estimates at the different levels. In other words, we show that, with the right setting of parameters, we can make sure that Condition~(2) holds. The analysis of this component still assumes that Condition (1) holds.

We first show that between every two consecutive times in which we modify our output, the value of the target function must change noticeably. Formally,

\begin{lemma}[informal version of Lemma~\ref{lem:consequtiveOutputUpdatesProgress}]
Let $t_1< t_2\in [m]$ be consecutive times in which the output is modified (i.e., the output is modified in each of these two iterations, and is not modified between them). Then, 
$|\mathcal{F}(t_2) - \mathcal{F}(t_1)| =\Omega\left( \alpha \cdot \mathcal{F}(t_1) \right)$.
\end{lemma}

We leverage this lemma in order to show that there cannot be too many time steps during which we modify our output. We then partition these time steps and ``charge'' different levels $j$ for different times during which the output is modified. This allows us to prove a probabilistic bound on the maximal number of times we aggregate the estimates from the different levels (each level has a different bound). See Lemma~\ref{lem:outputUpdatesBounds} for the formal details.

\subsubsection{Forth component: The framework is robust}
In Section~\ref{sec:formalRobust} we prove that Condition~(1) holds (w.h.p.). That is, we show that at any time step throughout the execution, at least 80\% of the estimators in every level are accurate.

This includes two parts. First, in Lemma~\ref{lem:boundedEstimationRanges}, we show that throughout the execution, the condition required by our TDE's hold (specifically, see \ref{req:TDEdiffRange} in Definition~\ref{def:TDE}). This means that, {\em had the stream been fixed in advance}, then (w.h.p.)\ we would have that {\em all} of the estimators are accurate throughout the execution. In other words, this shows that if there were no adversary then (a stronger variant of) Condition~(1) holds. 

Second, in Lemma~\ref{lem:AccurateEstimations} we leverage the generalization properties of differential privacy to show that Condition~(1) must also hold in the adversarial setting. This lemma is similar to the analysis of \cite{hassidim2020adversarially}.

\subsubsection{Fifth component: Calculating the space complexity}
In the final part of the analysis, in Section~\ref{sec:formalSpaceComplexity}, we calculate the total space needed by the framework by accounting for the number of estimators in each level (which is a function of the high probability bound we derived on the number of aggregations done in each level), and the space they require. We refer the reader to Appendix~\ref{sec:formalAnalysis} for the formal analysis.

%%%%%%%%%%%%%%%%%%%%%%%%%%%%%%%%%%%%%%%
% DE from TDE - a general method
%%%%%%%%%%%%%%%%%%%%%%%%%%%%%%%%%%%%%%%
\section{Toggle Difference Estimator from a Difference Estimator}\label{sec:TDE}
We present a simple method that transforms any {\em difference estimator} to a {\em toggle difference estimator}. The method works as follows. Let $\DE$ be a difference estimator (given as an subroutine). We construct a $\TDE$ that instantiates two copies of the given difference estimator: $\DE_{{
\rm enable}}$ and $\DE_{{\rm fresh}}$. It also passes its parameters, apart of the enabling times, verbatim to both copies. As $\DE$ is set to output estimations only after receiving an (online) enabling time $e$, the $\TDE$ never enables the copy $\DE_{{\rm fresh}}$. Instead, $\DE_{{\rm fresh}}$ is used as a fresh copy that received the needed parameters and the stream $\mathcal{S}$ and therefore it is always ready to be enabled. Whenever a time $t$ is equal to some enabling time (i.e. $t=e^i$ for some $i\in [p]$), then the $\TDE$ copies the state of $\DE_{{\rm fresh}}$ to $\DE_{{\rm enable}}$ (running over the same space), and then it enables $\DE_{{\rm enable}}$ for outputting estimations. 

\begin{corollary}\label{cor:TDEfromDE}
For any function $\mathcal{F}$, provided that there exist a $(\gamma,\alpha,\delta)$-Difference Estimator for $\mathcal{F}$ with space $S_{\DE}(\gamma,\alpha,\delta,n,m)$, then there exists a $(\gamma,\alpha,\delta,p)$-Toggle Difference Estimator for $\mathcal{F}$ with space $ S_{\TDE}(\gamma,\alpha,\delta,p,n,m) = 2\cdot S_{\DE}(\gamma,\alpha,\delta/p,n,m)$
\end{corollary}
Note that for a $\DE$ whose space dependency w.r.t.\ the failure parameter $\delta$ is logarithmic, the above construction gives a $\TDE$ with at most a logarithmic blowup in space, resulting from the $p$ enabling times.

%%%%%%%%%%%%%%%%%%%%%%%%%%%%%%%%%%%%%%%%%%%%%%%%%
% Applications
%%%%%%%%%%%%%%%%%%%%%%%%%%%%%%%%%%%%%%%%%%%%%%%%%
\section{Applications}\label{sec:Applications}

Our framework is applicable to functionalities that admit a strong tracker and a difference estimator. As \cite{woodruff2020tight} showed, difference estimators exist for many functionalities of interest in the {\em insertion only} model, including estimating frequency moments of a stream, estimating the number of distinct elements in a stream, identifying heavy-hitters in a stream and entropy estimation. 
However, as we mentioned, we are not aware of non-trivial DE constructions in the turnstile model. In more detail, \cite{woodruff2020tight} presented DE for the turnstile setting, but these DE require additional assumptions and do not exactly fit our framework (nor the framework of \cite{woodruff2020tight}).

To overcome this challenge we introduce a new {\em monitoring technique} which we use as a wrapper around our framework. This wrapper allows us to check whether the additional assumptions required by the DE hold, and reset our system when they do not. As a concrete application, we present the resulting bounds for $F_2$ estimation. 

\begin{definition}[Frequency vector] \label{def:freqVec}
The {\em frequency vector} of a stream $S = (\langle s_1, \Delta_1 \rangle ,\dots, \langle s_m, \Delta_m \rangle)\in ([n]{\times}\{\pm1\})^{m}$ is the vector $u\in \Z^{n}$ whose $i$th coordinate is 
$u[i] = \sum_{j\in [m], s_j = i}{\Delta_j}.$ We write $u^{(t)}$ to denote the frequency vector of the stream $S_t$, i.e., restricted to the first $t$ updates. Given two time points $t_1\leq t_2\in[m]$ we write $u^{(t_1,t_2)}$ to denote the frequency vector of the stream $S_{t_1}^{t_2}$, i.e., restricted to the updates between time $t_1$ and $t_2$.
\end{definition}

In this section we focus on estimating $F_2$, the second moment of the frequency vector. That is, after every time step $t$, after obtaining the next update $\langle s_t, \Delta_t \rangle\in([n]{\times}\{\pm1\})$, we want to output an estimation for
$$
\left\|u^{(t)}\right\|_2^2 = \sum_{i=1}^n \left|u^{(t)}[i]  \right|^2.
$$

Woodruff and Zhou \cite{woodruff2020tight} presented a $(\gamma, \alpha, \delta)$-difference estimator for $F_2$ that works in the turnstile model, under the additional assumption that for any time point $t$ and enabling time $e\leq t$ it holds that
\begin{equation}
  \left\| u^{(e,t)} \right\|_2^2\leq \gamma\cdot \left\| u^{(e)} \right\|_2^2. \label{req:F2DEratio_body}
\end{equation}

In general, we cannot guarantee that this condition holds in a turnstile stream. To bridge this gap, we introduce the notion of {\em twist number} (see Definition~\ref{def:twistIntro}) in order to control the number of times during which this condition does not hold (when this condition does not hold we say that a {\em violation} has occurred). Armed with this notion, our approach is to run our framework (algorithm \texttt{RobustDE}) alongside a {\em validation algorithm} (algorithm \texttt{Guardian}) that identifies time steps at which algorithm \texttt{RobustDE} loses accuracy, meaning that a violation has occurred. We then restart algorithm \texttt{RobustDE} in order to maintain accuracy. As we show, our notion of twist number allows us to bound the total number of possible violation, and hence, bound the number of possible resets. This in turn allows us to bound the necessary space for our complete construction. The details are given in Appendix~\ref{sec:MissingAppli}; here we only state the result.

\begin{theorem}
There exists an adversarially robust $F_2$ estimation algorithm for turnstile streams of length $m$ with a bounded  $(O(\alpha), m)$-flip number and $(O(\alpha), m)$-twist number with parameters $\lambda$ and $\mu$ correspondingly, that guarantees $\alpha$-accuracy with probability at least $1-1/m$ in all time $t\in [m]$ using space complexity of
$$\tilde{O}\left(\frac{\sqrt{\alpha\lambda+\mu}}{\alpha^{2}}\log^{3.5}(m)\right)\text{.}$$
\end{theorem}

As we mentioned, this should be contrasted with the result of \cite{hassidim2020adversarially}, who obtain space complexity  $\tilde{\mathcal{O}}\left(\frac{\sqrt{\lambda}}{\alpha^2}\right)$ for robust $F_2$ estimation in the turnstile setting. Hence, our new result is better whenever $\mu\ll\lambda$.

%%%%%%%%%%%%%%%
% bibliography
%%%%%%%%%%%%%%%
\bibliographystyle{alpha}
\bibliography{Robust_Streaming_Framework}

%%%%%%%%%%%%%%%
% appendix
%%%%%%%%%%%%%%%
\appendix

%%%%%%%%%%%%%%%%%%%%%%%%%%%%%%%%%%%%%%%%%%%%%%
% Section: Additional preliminaries
% Required theorems from DP
%%%%%%%%%%%%%%%%%%%%%%%%%%%%%%%%%%%%%%%%%%%%%%
\section{Additional Preliminaries from Differential Privacy}\label{sec:MoreDP}

%%%%%%%%%%%%%%%%%%%%%%%%%%%%
% The Laplace Mechanism
%%%%%%%%%%%%%%%%%%%%%%%%%%%%
\paragraph{The Laplace Mechanism.} The most basic constructions of differentially private algorithms are
via the Laplace mechanism as follows.
\begin{definition}[The Laplace distribution] A random variable has probability distribution $\Lap(b)$ if its probability density function is $f(x)=\frac{1}{2b}\exp\left(-\frac{|x|}{b} \right)$, where $x\in\R$.
\end{definition}
\begin{definition}[Sensitivity] A function $f:X^*\rightarrow \R$ has sensitivity $\ell$ if for every two databases $S,S^\prime\in X^*$ that difer in one row it holds that $|f(S)-f(S^\prime)|\leq\ell$.
\end{definition}
\begin{theorem}[Laplace mechanism \cite{dwork2006calibrating}]\label{thm:LaplaceMechanism} Let $f:X^*\rightarrow \R$ be a sensitivity $\ell$ function. The mechanism that on input $S\in X^*$ returns $f(S)+\Lap(\frac{\ell}{\eps})$ preserves $(\eps,0)$-differential privacy.
\end{theorem}

%%%%%%%%%%%%%%%%%%%%%%%%%%%%
% The sparce vector
%%%%%%%%%%%%%%%%%%%%%%%%%%%%
\paragraph{The sparse vector technique.} Consider a large number of low-sensitivity functions $f_1,f_2,\dots$ which are given (one by one) to a data curator (holding a database $S$). Dwork et al.~\cite{DBLP:conf/stoc/DworkNRRV09} presented a simple tool, called \texttt{AboveThreshold} (see Algorithm~\ref{alg:AboveThreshold}), for privately identifying the first index $i$ such that the value of $f_i(S)$ is ``large''.

\begin{algorithm*}[ht]
\caption{\texttt{AboveThreshold}}\label{alg:AboveThreshold}
{\bf Input:} Database $S\in X^*$, privacy parameter $\eps$, threshold $t$ and a stream of sensitivity-1 queries $f_i:X^*\rightarrow \R$.

\begin{enumerate}
	\item Let $\hat{t}\leftarrow t+ \Lap(\frac{2}{\eps})$
	\item In each round $i$, when receiving $f_i$ do the following:
	\begin{enumerate}
		\item Let $\hat{f}_i\leftarrow f_i(S)+\Lap(\frac{4}{\eps})$.
		\item If $\hat{f}_i > \hat{t}$, then output $\top$ and halt.
		\item Otherwise, output $\bot$ and proceed to the next iteration.
	\end{enumerate}
\end{enumerate}
\end{algorithm*}

\begin{theorem}[\cite{DBLP:conf/stoc/DworkNRRV09,DBLP:conf/focs/HardtR10}]\label{thm:AboveThreshold} Algorithm \texttt{AboveThreshold} is $(\eps,0)$-differentially private.
\end{theorem}

%%%%%%%%%%%%%%%%%%%%%%%%%%%%
% \PrivateMed
%%%%%%%%%%%%%%%%%%%%%%%%%%%%
\paragraph{Privately approximating the median of the data.} Given a database $S\in X^{*}$, consider the task of {\em privately} identifying an {\em approximate median} of $S$. Specifically, for an error parameter $\Gamma$, we want to identify an element $x\in X$ such that there are at least $|S|/2-\Gamma$ elements in $S$ that are bigger or equal to $x$, and there are at least $|S|/2-\Gamma$ elements in $S$ that are smaller or equal to $x$. The goal is to keep $\Gamma$ as small as possible, as a function of the privacy parameters $\eps, \delta$, the database size $|S|$, and the domain size $|X|$. 

There are several advanced constructions in the literature with error that grows very slowly as a function of the domain size (only polynomially with $\log^{*}|X|$). \cite{DBLP:journals/toc/BeimelNS16, DBLP:conf/focs/BunNSV15, DBLP:conf/stoc/BunDRS18, DBLP:conf/colt/KaplanLMNS20} In our application, however, the domain size is already small, and hence, we can use simpler constructions (where the
error grows logarithmically with the domain size).

\begin{theorem}[folklore]\label{thm:PrivateMed} There exists an $(\eps,0)$-differentially private algorithm that given a database $S\in X^{*}$ outputs an element $x\in X$ such that with probability at least $1-\delta$ there are at least $|S|/2-\Gamma$ elements in $S$ that are bigger or equal to $x$, and there are at least $|S|/2-\Gamma$ elements in S that are smaller or equal to $x$, where $\Gamma = O\left( \frac{1}{\eps}\log\left( \frac{|X|}{\delta} \right) \right)$.
\end{theorem}

%%%%%%%%%%%%%%%%%%%%%%%%%%%%
% Composition
%%%%%%%%%%%%%%%%%%%%%%%%%%%%
\paragraph{Composition of differential privacy.} The following theorems argue about the privacy guarantees of an algorithm that accesses its input database using several differentially private mechanisms.
\begin{theorem}[Simple composition~\cite{DBLP:conf/eurocrypt/DworkKMMN06, DBLP:conf/stoc/DworkL09}]\label{thm:compSimple} Let $0<\eps \leq 1$, and let $\delta \in [0,1]$. A mechanism that permits $k$ adaptive interactions with mechanisms that preserve $(\eps,\delta)$-differential privacy (and does not access the database otherwise) ensures $(k\eps, k\delta)$-differential privacy.
\end{theorem}
\begin{theorem}[Advanced composition~\cite{DBLP:conf/focs/DworkRV10}]\label{thm:PrivacyCompo} Let $0<\eps, \delta^{\prime} \leq 1$, and let $\delta \in [0,1]$. A mechanism that permits $k$ adaptive interactions with mechanisms that preserve $(\eps,\delta)$-differential privacy (and does not access the database otherwise) ensures $(\eps^{\prime},k\delta +\delta^{\prime})$-differential privacy, for $\eps^{\prime}=\sqrt{2k\ln(1/\delta^{\prime})}\cdot \eps + 2k\eps^2$.
\end{theorem}

%%%%%%%%%%%%%%%%%%%%%%%%%%%%%%%%%%%
% Privacy implies Generalization
%%%%%%%%%%%%%%%%%%%%%%%%%%%%%%%%%%%
\paragraph{Generalization properties of differential privacy.} Dwork et al. \cite{DworkFHPRR15} and Bassily et al. \cite{BassilyNSSSU21} showed that if a predicate $h$ is the result of a differentially private computation on a random sample, then the empirical average of $h$ and its expectation over the underlying distribution are guaranteed to be close.
\begin{theorem}[\cite{DworkFHPRR15, BassilyNSSSU21}]\label{thm:PrivacyImpGen} Let $\eps \in (0,1/3)$, $\delta\in (0,\eps/4)$ and $n \geq \frac{1}{\eps^2}\log(\frac{2\eps}{\delta})$. Let $\mathcal{A}: X^{n}\rightarrow 2^{X}$ be an $(\eps,\delta)$-differentially private algorithm that operates on a database of size $n$ and outputs a predicate $h:X\rightarrow\{0,1\}$. Let $\mathcal{D}$ be a distribution over $X$, let $S$ be a database containing $n$ i.i.d elements from $\mathcal{D}$, and let $h\leftarrow \mathcal{A}(S)$. Then
$$
\Pr_{\stackrel{S\sim\mathcal{D}}{h\leftarrow \mathcal{A}(S)}} 
\left[ \left|
\frac{1}{|S|}\sum_{x\in S} h(x) - \mathop{\mathbb{E}}_{x\sim \mathcal{D}}[h(x)] \right| \geq 10\eps 
\right] < \frac{\delta}{\eps}
$$
\end{theorem}

%%%%%%%%%%%%%%%%%%%%%%%%%%%%%%%%%%%%%%%%%%%%%%
% Section: Main result
% Construction and formal analysis of RobustDE
%%%%%%%%%%%%%%%%%%%%%%%%%%%%%%%%%%%%%%%%%%%%%%
\section{The Formal Analysis}\label{sec:formalAnalysis}

In this section we present the full construction (Algorithm \ref{alg:ROEF}) and its formal analysis.

%%%%%%%%%%%%%%%%%%%%%%%%%%%%%%%%%%%%%%%%%%%%%%
% Turnstile Hybrid Framework: Construction
%%%%%%%%%%%%%%%%%%%%%%%%%%%%%%%%%%%%%%%%%%%%%%
\begin{algorithm*}
\caption{\texttt{RobustDE($\SSS, \alpha,\delta,\lambda,\mathsf{E}_{\ST},\mathsf{E}_{\TDE}$)}}
\makeatletter\def\@currentlabel{\texttt{RobustDE}}\makeatother
\label{alg:ROEF}

{\bf Input:} A stream $\mathcal{S} = \{\langle s_t, \Delta_t\rangle \}_{t\in [m]}$ accuracy parameter $\alpha$, failure probability $\delta$ and a bound on the flip number $\lambda$.%\\[-15px]

{\bf Estimators used:} Strong tracker {\em $\mathsf{E}_{\ST}$} 
and toggle-difference-estimator $\mathsf{E}_{\TDE}$ for the function $\mathcal{F}$.%\\[-15px]

{\bf Subroutines used:} \ref{alg:StitchFrozenValues}, \ref{alg:ActiveLVL}.%\\[-10px]

\hdashrule[0.5ex][x]{\linewidth}{0.5pt}{1.5mm}

{\bf Constants calculation:}
\begin{enumerate}
[leftmargin=25pt,rightmargin=10pt,itemsep=1pt,topsep=0pt]
	
	\item $\StepSize(\alpha)\leftarrow O(\alpha)$, $\alpha_{\ST} \leftarrow O(\alpha)$, $\alpha_{\TDE} \leftarrow O(\alpha/\log(\alpha^{-1}))$, $\Gamma\leftarrow\Theta(1)$.
	
	\item Phase params: let $\PhaseSize \leftarrow \lfloor \frac{1}{\StepSize(\alpha)}\rfloor$, $\beta\leftarrow \lceil \log (\PhaseSize) \rceil$, $J\leftarrow [\beta+1]\cup\{\STwrapper\}$.

    \item For $j\in [\beta+1]$:$P_j \leftarrow O\left(2^{-j}\lambda\right)$, $P_{\STwrapper} \leftarrow P_{\beta}$; For $j\in [\beta]$:$\gamma_j\leftarrow\Omega(2^j\alpha)$; For $j\in J$:$\eps_j = \tilde{O}(P_j^{-0.5})$. 
	
	\item Estimator sets: For $j\in J$ set $\mathsf{K}_{j}\leftarrow \tilde{\Omega}(\eps_j^{-1})$ and let $\bar{\mathsf{E}}_{j}=\{\mathsf{E}_{j}^k \}_{k\in[\mathsf{K}_{j}]}$.
	
\end{enumerate} 

\hdashrule[0.5ex][x]{\linewidth}{0.5pt}{1.5mm}

{\bf Initialization:}
\begin{enumerate}[resume,leftmargin=25pt,rightmargin=0pt,itemsep=0pt,topsep=0pt]

	\item Start all estimators $\{\bar{\mathsf{E}}_{j}\}_{j\in J}$.
    
	\item Set thresholds: For $j\in J$, $\mathsf{T}_{j}\leftarrow \mathsf{K}_{j}/2 + \Lap(2\cdot \eps^{-1}_{j})$.
	
	\item For $j\in J$ set $\Capp_j\leftarrow 0$; For $j\in [\beta+1]$ set $\mathsf{Z}_{j} \leftarrow 0$; $\tau \leftarrow \PhaseSize$.
\end{enumerate}

\hdashrule[0.5ex][x]{\linewidth}{0.5pt}{1.5mm}

{\bf For ${\boldsymbol{t\in[m]}}$:}
\begin{enumerate}[resume,leftmargin=25pt,rightmargin=10pt,itemsep=1pt,topsep=3pt]

    \item\label{algStep:setUpdate} Get the update $\langle s_t, \Delta_t \rangle$ from $\mathcal{S}$ and feed $\langle s_t, \Delta_t \rangle$ into all estimators.

    \item\label{algStep:STwrapper} If $\left| \left\{ k\in [\mathsf{K}_{\STwrapper}] : z_{\STwrapper}^{k} \notin \Big(\frac{1}{\Gamma}\cdot \mathsf{Z}_{\ST} \;,\; \Gamma\cdot\mathsf{Z}_{\ST} \Big)  \right\} \right| + \Lap(4\cdot \eps^{-1}_{\STwrapper})> \mathsf{T}_{\STwrapper}$ then set $\tau=0$, redraw $\mathsf{T}_{\STwrapper}$.

    \item\label{step:choose-j} $j \leftarrow \texttt{ActiveLVL}(\tau)$\label{algStep:levelSelect}
    \item For $k\in [\mathsf{K}_{j}]$ get estimation $z_{j}^{k} \leftarrow \mathsf{E}_{j}^k$\label{algStep:GetEstimation}
    \item If $j = \beta$: $\mathsf{Z} \leftarrow 0$\label{algStep:STBase} \qquad \qquad \qquad \qquad \qquad \qquad \qquad \quad{} {\bf \% Estimation offset of \ST}
    \item Else: \qquad $\mathsf{Z} \leftarrow \StitchFrozenVals(\tau-2^{j}+1)$\label{algStep:FrozenBase} \qquad {\bf \% Estimation offset of \TDE}
    \item\label{algStep:AboveThresh} If $\tau=0$ or $\left| \left\{ k\in [\mathsf{K}_j] : \mathsf{Z}+z_{j}^{k} \notin \Output(t-1) \pm \mathsf{Z}_{\ST}\cdot \StepSize(\alpha)  \right\} \right| + \Lap(4\cdot \eps^{-1}_{j})> \mathsf{T}_{j}$
    \begin{enumerate}[leftmargin=25pt,rightmargin=10pt,itemsep=1pt,topsep=0pt]
        \item\label{step:modifying} $\mathsf{Z}_{j}\leftarrow \PrivateMed(\{z_{j}^{k}\}_{k\in [\mathsf{K}_{j}]})$ \label{algStep:PrivMed}
        \item Redraw $\mathsf{T}_{j}$
        \item For $j^{\prime}\in [j], k \in [\mathsf{K}_{j^{\prime}}]$ set $e_{j^{\prime}}^k \leftarrow t$\label{algStep:EnableTDEs}
        \item $\Capp\text{}_j \leftarrow \Capp\text{}_j +1$
        \item If $\tau[\beta] = 0$: $\tau \leftarrow  2^{\beta}$\label{algStep:tauUpdateStartPhase} \qquad \qquad \qquad \qquad \qquad{} \quad{}{} {\bf \% Starting phase}
        \item ElseIf $\tau < 2^{\beta} + \PhaseSize$: $\tau \leftarrow \tau + 1$\label{algStep:tauUpdateInnerPhase} \qquad{} \quad{} {\bf \% Inner phase update}
        \item If $\tau = 2^{\beta} + \PhaseSize$: $\tau[\beta] \leftarrow 0$\label{algStep:tauUpdateEndPhase} \qquad \qquad \qquad {\bf \% Ending phase}
	\end{enumerate}
    \item $\NoCapping \leftarrow \bigwedge_{j \in J} \Capp_j < P_j$
	\item If $\NoCapping = \text{True}$: Output(t) $\leftarrow \StitchFrozenVals(\tau)$\label{algStep:outputStitch}
\end{enumerate}
\end{algorithm*}

%%%%%%%%%%%%%%%%%%%%%%%%%%%%%%%%%%%%%%%
% Formal proofs - Privacy
%%%%%%%%%%%%%%%%%%%%%%%%%%%%%%%%%%%%%%%
\subsection{Privacy analysis}\label{sec:formalPrivacy}
The following lemma shows that algorithm \texttt{RobustDE} is private w.r.t.\ the random bit-strings of the $\beta+1$ type datasets.
\begin{lemma}\label{lem:EstimatorsRandomStringsPrivacy}
For level $j\in [\beta+1]\cup \{\STwrapper\}$ let $\mathcal{R}_j$ be its corresponding random bit-strings dataset. Algorithm \ref{alg:ROEF} satisfies $(\eps,\delta^{\prime})$-DP w.r.t.\ a dataset $\mathcal{R}_j$ by configuring $\eps_{j}=O\left(\eps/\sqrt{P_j\log(1/\delta^{\prime})}\right)$.
\end{lemma}
\begin{proof}[Proof sketch] 
Let us focus on some level $j\in[\beta+1]\cup \{\STwrapper\}$. We analyze the privacy guarantees w.r.t.\ $\mathcal{R}_j$ by arguing separately for every sequence of time steps during which we do not modify our output using level $j$ (the sequence ends in a time point at which we do modify the output using level $j$).\footnote{Note that this sequence contains time points at which we modify our output using different levels.} 
Let us denote by $P_j$ the number of such sequences we allow for each level $j$ after-which algorithm \texttt{RobustDE} is not generating any output due to the {\em capping} counters. 
Throughout every such time sequence, we access the dataset $\mathcal{R}_j$ via the sparse vector technique and once  (at the end of the sequence) using the private median algorithm.  We calibrate the privacy parameters of these algorithms to be $\eps_{j}=O\left(\eps/\sqrt{P_j\log(1/\delta^{\prime})}\right)$ such that, by using composition theorems across all of the $P_j$ sequences, our algorithm satisfies $(\eps,\delta^{\prime})$-differential privacy w.r.t.\ $\mathcal{R}_j$.
\end{proof}

%%%%%%%%%%%%%%%%%%%%%%%%%%%%%%%%%%%%%%%
% Formal proofs - Accuracy
%%%%%%%%%%%%%%%%%%%%%%%%%%%%%%%%%%%%%%%
\subsection{Conditional accuracy}\label{sec:formalAccuracy} 

We first prove the framework accuracy assuming that throughout the run of algorithm \texttt{RobustDE}, for all $t\in [m]$ $80\%$ of the estimations given from the estimators are accurate (each level w.r.t its accuracy parameter $\alpha_j$). This assumption will be proved independently on 
\ref{lem:accurateEstimationsAllTime}.
The following is its formal definition:
%%%%%%%%%%%%%%%%%%%%%%%%%%%%%%%%%%%%%%%
% accurate estimators assumption
%%%%%%%%%%%%%%%%%%%%%%%%%%%%%%%%%%%%%%%
\begin{assumption}[Accurate estimations]\label{ass:accurateEstimators} Fix a time step $t\in [m]$. 
Let $j\in [\beta+1]\cup \{\STwrapper \}$ be a level of estimators.
Recall that $\mathsf{K}_j$ denotes the number of the estimators in level $j$, and let $z^{1}_{j}, \dots, z^{\mathsf{K}_j}_{j}$ denote the estimations given by these estimators.
Then:
\begin{enumerate}
    \item For $j \in \{\beta, \STwrapper\}$  \qquad 
    $|\{k\in [\mathsf{K}_j] : |z^{k}_{j}-\mathcal{F}(t)|< \alpha_{\ST}\cdot\mathcal{F}(t) \}| \geq (8/10) \mathsf{K}_j$
    \item For $j < \beta$, \qquad \qquad
    $|\{k\in [\mathsf{K}_j] : |z^{k}_{j}-(\mathcal{F}(t) - \mathcal{F}(e_j))|< \alpha_{\TDE}\cdot\mathcal{F}(e_j) \}| \geq (8/10) \mathsf{K}_j$
\end{enumerate}
\end{assumption}
The framework accuracy is proved on three steps. The first step is arguing that given the assumption on the accuracy of the given estimations, every frozen value is accurate w.r.t the function that it has estimated (Lemma \ref{lem:FrozenAccuracy}). 
The second step is a follow up to the first: 
focusing on the active level of estimators, call it $j$, and the value of the function $\mathcal{F}$ at the time these estimators were enabled ($\mathcal{F}(e_j)$), then combining frozen values of relevant levels results in an accurate estimation for $\mathcal{F}(e_j)$ (denoted on step \ref{algStep:FrozenBase} as $\mathsf{Z}$). That is achieved by applying \ref{lem:FrozenAccuracy} on each of these frozen values of the relevant levels, and accounting for the accumulated error.
And so, we have that $\mathsf{Z}\approx \mathcal{F}(e_j)$. In addition, the assumption of the accuracy promise that $80\%$ of (in particular) level $j$ estimators are accurate thus their estimations $z_j \approx\mathcal{F}(t) - \mathcal{F}(e_j)$. Combining it we get $\mathsf{Z} + z_j \approx \mathcal{F}(t)$ which is established on Lemma \ref{lem:estimationErrorBound}.
Applying once again Lemma \ref{lem:FrozenAccuracy} for times $t$ that the estimators of level $j$ were aggregated into the frozen value $\mathsf{Z}_j$ (step \ref{algStep:PrivMed}) results in $\mathsf{Z}+\mathsf{Z}_j \approx \mathcal{F}(t)$.
By further observe that on these steps, the output is modifies into $\mathsf{Z}+\mathsf{Z}_j$, we get that on an output modification steps we guarantee a more refined accuracy then $\alpha$. See Corollary \ref{cor:outputUpdateAccuracy} in-which the mentioned observation is elaborated.
Finally, using \ref{lem:FrozenAccuracy}, \ref{lem:estimationErrorBound}, \ref{cor:outputUpdateAccuracy}, we prove that given that the estimators are accurate (Assumption \ref{ass:accurateEstimators}), then on all $t\in [m]$ before capping stage we have an accurate output (Lemma \ref{lem:outputAccuracy}). 
We prove an additional lemma in this section: Lemma \ref{lem:maxPhaseProgress}. Lemma \ref{lem:maxPhaseProgress} states that during any phase the value of function $\mathcal{F}$ changes (increases or decreases) by at most a constant factor. That lemma is used in the proofs of \ref{lem:estimationErrorBound}, \ref{cor:outputUpdateAccuracy}, \ref{lem:outputAccuracy}.

%%%%%%%%%%%%%%%%%%%%%%%%%%%%%%%%%%%%%%%
% Accuracy of frozen values
%%%%%%%%%%%%%%%%%%%%%%%%%%%%%%%%%%%%%%%
\begin{lemma}[Accuracy of frozen values] \label{lem:FrozenAccuracy}
Let $t\in [m]$ be a time step such that
\begin{enumerate}
    \item Assumption \ref{ass:accurateEstimators} holds for every $t'\leq t$.
    \item $\NoCapping{=} \text{True}$ during time $t$.
    \item Algorithm {\bf \PrivateMed} was activated during time $t$ (on Step~\ref{algStep:PrivMed}\text{}).
\end{enumerate}
Let $j\in [\beta+1]$ be the level of estimators used in time $t$. 
Let $\mathsf{Z}_j$ be the value returned by {\bf \PrivateMed}, and suppose that $
\mathsf{K}_{j} = \Omega \left( \frac{1}{\eps}
\sqrt{P_{j} \cdot \log \left(\frac{1}{\delta^{\prime}} \right) }\log\left( \frac{P_{j}}{\delta^{M}\alpha} \log(n)\right)
\right)$. 
Then, with probability at least $1-\delta^{M}/P_j$ we have that:
\begin{enumerate}
    \item For $j = \beta$, \qquad $|\mathsf{Z}_j - \mathcal{F}(t)| <  \alpha_{\ST}\cdot \mathcal{F}(t)\text{.}$
    \item For $j < \beta$, \qquad 
    $|\mathsf{Z}_j - (\mathcal{F}(t) - \mathcal{F}(e_j))| <  \alpha_{\TDE}\cdot \mathcal{F}(e_j)\text{.}$
\end{enumerate}
\end{lemma}
\begin{proof}
In the case that step \ref{algStep:PrivMed} was executed, mechanism {\bf \PrivateMed} was activated on the estimations $z_j^{1},\dots,z_j^{\mathsf{K}_{j}}$ of level $j$ estimators to get a new value for $\mathsf{Z}_{j}$. By theorem \ref{thm:PrivateMed}, assuming that\footnote{We assume there exist some constant $c$ for which all estimates returned by the oblivious estimators type $\mathsf{E}_j$ are within the range of $[-n^{c},-1/n^{c}]\cup \{0\} \cup [1/n^{c}, n^{c}]$. Rounding these estimates to their nearest values of $(1\pm \MaxUpdateSize(\alpha))$ has only a small effect on the error. Such rounding on for range yields at most $X=O(\alpha^{-1}\log(n))$ possible values.}
$$
\mathsf{K}_{j} = \Omega \left( \frac{1}{\eps}
\sqrt{P_{j} \cdot \log \left(\frac{1}{\delta^{\prime}} \right) }\log\left( \frac{P_{j}}{\delta^{M}\alpha} \log(n)\right)
\right)\text{,}
$$
then with probability at least $1-\delta^{M}/P_{j}$ Algorithm {\bf \PrivateMed} returns an approximate median $\mathsf{Z}_{j}$ to the estimations $z_j^{1},\dots,z_j^{\mathsf{K}_{j}}$, satisfying
$$
\left| \left\{
k\in \mathsf{K}_{j} : z_j^{k} \geq \mathsf{Z}_{j}
\right\}\right| \geq \frac{4\cdot \mathsf{K}_{j}}{10}
\qquad \text{and}
\qquad 
\left| \left\{
k\in \mathsf{K}_{j} : z_j^{k} \leq \mathsf{Z}_{j}
\right\}\right| \geq \frac{4\cdot \mathsf{K}_{j}}{10}.
$$
By assumption \ref{ass:accurateEstimators}, $(8/10)\cdot \mathsf{K}_{j}$ of the estimations $z^{k}$ satisfy the condition $|z^k - \mathcal{F}(t)| < \alpha_{\ST} \cdot \mathcal{F}(t)$ (or $|z^k -(\mathcal{F}(t) - \mathcal{F}(e_j))| <\alpha_{\TDE}\cdot \mathcal{F}(e_j)$ respectively), the approximate median $\mathsf{Z}_j$ must also satisfy this condition.
\end{proof}

%%%%%%%%%%%%%%%%%%%%%%%%%%%%%%%%%%%%%%%
% Definition: 'Good Execution'
%%%%%%%%%%%%%%%%%%%%%%%%%%%%%%%%%%%%%%%
\begin{definition}[Good execution]\label{def:goodRun} Throughout the execution of algorithm \ref{alg:ROEF}, for $j\in [\beta+1]\cup \{\STwrapper\}$ the algorithm draws at most $4m$ noises from Laplace distribution with parameter $\eps_j$. In addition, denoting by $P_j$ for $j\in [\beta+1]$ the number of times that algorithm \ref{alg:ROEF} activates $\PrivateMed$ on estimations of level $j$.
Denote $\delta^{N}=\delta/(4\cdot (\beta+2))$, $\delta^M = \delta/(4\cdot(\beta+1))$. 
Set the algorithm parameters as follows $
\mathsf{K}_{j} = \Omega \left( \frac{1}{\eps_j} \log\left( \frac{P_{j}}{\delta^{M}\alpha} \log(n)\right)
\right)
$, $\eps_j = O\left( \eps/\sqrt{P_{j} \cdot \log \left(\frac{1}{\delta^{\prime}} \right) }
 \right)$
We define a {\em good execution} as follows:
\begin{enumerate}
    \item All noises for all types $j\in [\beta+1]\cup \{\STwrapper\}$ are at most $O\left(\frac{1}{\eps_j}\log\left( \frac{ m}{\delta^{N}}\right)\right)$ in absolute value.
    \item For all $j\in [\beta +1]$, all first $P_j$ frozen values of level $j$ are accurate. That is, if $t$ is the time of the frozen value computation then:
    \begin{itemize}
        \item For $j = \beta$, \qquad $|\mathsf{Z}_j - \mathcal{F}(t)| <  \alpha_{\ST}\cdot \mathcal{F}(t)\text{.}$
        \item For $j < \beta$, \qquad 
    $|\mathsf{Z}_j - (\mathcal{F}(t) - \mathcal{F}(e_j))| <  \alpha_{\TDE}\cdot \mathcal{F}(e_j)\text{.}$
    \end{itemize}
\end{enumerate}
\end{definition}
We configure level $j$ Laplace noise with parameter $\eps_j$. By the properties of Laplace distribution, with probability at least $1-\delta/4$, all noises for all types $j\in [\beta+1]\cup \{\STwrapper \}$ are at most $\frac{4}{\eps_j}\log\left( \frac{4m}{\delta^{N}}\right)$ in absolute value.
By Lemma \ref{lem:FrozenAccuracy} the second requirement of Definition \ref{def:goodRun} occur w.p. at least $1-\delta/4$. That implies a good execution w.p. at least $1-\delta/2$.
We continue with the analysis assuming a {\em good execution} (\ref{def:goodRun}).

%%%%%%%%%%%%%%%%%%%%%%%%%%%%%%%%%%%%%%%
% Formal proofs - Max Phase Size
%%%%%%%%%%%%%%%%%%%%%%%%%%%%%%%%%%%%%%%
\paragraph{Max phase progress} Algorithm \ref{alg:ROEF} is coded with mechanism that guarantees a maximal progress of a phase (Step \ref{algStep:STwrapper}). 
A {\em phase} is starting whenever $\tau[\beta]$ is set to $0$ (in either Step \ref{algStep:STwrapper} or Step \ref{algStep:tauUpdateStartPhase}). 
Denoting by $t_p$ the time a phase has started, that code guarantees that in anytime $t$ throughout the phase, the ratio between the values of the function for times $t_p ,t$ is roughly bounded from above by $\Gamma$ and from below by $\Gamma ^{-1}$. That gives a bound on the ratio of the value of the function between any two times that are on the same phase of $\Theta(\Gamma^2)$. That bound is given in the following lemma formally.

%%%%%%%%%%%%%%%%%%%%%%%%%%%%%%%%%%%%%%%
% Constant change bound within a phase
%%%%%%%%%%%%%%%%%%%%%%%%%%%%%%%%%%%%%%%
\begin{lemma}[Max phase progress]\label{lem:maxPhaseProgress}
Let $t_1< t_2\in [m]$ be time steps such that
\begin{enumerate}
    \item Assumption \ref{ass:accurateEstimators} holds for every $t'\leq t_2$.
    \item $\tau\neq 0$ for every $t_1<t'\leq t_2$. \quad {\color{gray} \% In particular, this happens if $t_1,t_2$ are in the same phase.}
\end{enumerate}
Then, for any such $t_1, t_2 \in [m]$, assuming $\mathsf{K}_{\STwrapper} = \Omega\left( \sqrt{P_{\STwrapper}\cdot\log\left(\frac{1}{\delta^{\prime}}\right)}\cdot  \log\left( \frac{m}{\delta^{N}}\right)  \right)\text{,}$
For $\delta^{N}=O\left(\frac{\delta}{\log(\alpha^{-1})} \right)$
and $P_{\STwrapper} = O(\alpha\cdot \lambda)$
then assuming a {\em good execution} (see Definition \ref{def:goodRun}) we have
$$
\frac{\min\{ \mathcal{F}(t_1), \mathcal{F}(t_2)\} }
{\max\{ \mathcal{F}(t_1), \mathcal{F}(t_2)\} }
 \leq \left(\frac{1+\alpha_{\ST}}{1-\alpha_{\ST}}\right)\cdot \Gamma^2 =\Theta\left(\Gamma^2\right)
$$
Where $\alpha_{\ST}\leq 1$ is the accuracy parameter of estimators $E_{\STwrapper\text{}}, E_{\ST}$ and $\Gamma$ is some constant.
\end{lemma}
\begin{proof} 
Denote by $t_p$ the time that s phase started (that is, in which $\tau$ was set to $0$). We first bound the ratio between $\mathcal{F}(t_p)$ and $\mathcal{F}(t)$ for time $t$ in the same phase s.t. $\mathcal{F}(t_p)\leq \mathcal{F}(t)$ (the case $\mathcal{F}(t_p)\geq \mathcal{F}(t)$ is similar). Since $\tau$ was not set to $0$ it means that on Step \ref{algStep:STwrapper} the condition was not triggered. that is:
$$
\left| \left\{ k\in [\mathsf{K}_{\STwrapper}] : z_{k}^{\STwrapper} \in \Big(\frac{1}{\Gamma}\cdot \mathsf{Z}_{\ST} \;,\; \Gamma\cdot\mathsf{Z}_{\ST} \Big)  \right\} \right| \geq \frac{\mathsf{K}_{{\STwrapper}}}{2} -  \frac{4}{\eps_{\STwrapper}}\log\left( \frac{4m}{\delta^{N}}  \right)
\geq \frac{4\cdot \mathsf{K}_{{\STwrapper}}}{10}
$$

where the first inequality holds in a {\em good execution} (Definition \ref{def:goodRun}) and the last inequality follows by asserting that
$$\mathsf{K}_{\STwrapper} = \Omega\left( \sqrt{P_{\STwrapper}\cdot\log\left(\frac{1}{\delta^{\prime}}\right)}\cdot  \log\left( \frac{ m}{\delta^{N}}  \right)  \right)\text{,}$$ 
where $P_{\STwrapper}$ is the number of times the condition in Step \ref{algStep:STwrapper} may trigger during the run. Note that for a stream $\SSS$ with $\lambda_{\alpha^{\prime}}(\SSS)$ flip number, We have at most $P_{\STwrapper} =  O(\alpha\cdot \lambda_{\alpha^{\prime}}(\SSS) )$ times in which the function value is changed by a constant factor. 
So, for at least $4\mathsf{K}_{\STwrapper}/10$ of the estimations $z_{k}^{\STwrapper}$ we have that 
$z_{k}^{\STwrapper} < \Gamma\cdot \mathsf{Z}_{\ST}$. and we have:
\begin{align}
(1-\alpha_{\ST})\cdot \mathcal{F}(t) \stackrel{*}{\leq}  z_{k}^{\STwrapper} < \Gamma\cdot \mathsf{Z}_{\ST} \label{eq:maxProgAboveSt}
\end{align}
Where (*) follows from  Assumption~\ref{ass:accurateEstimators} . 
Similarly for times $t$ s.t. $\mathcal{F}(t_p)\geq \mathcal{F}(t)$ we get:
\begin{align}
(1+\alpha_{\ST})\cdot \mathcal{F}(t) \geq  z_{k}^{\STwrapper} > \Gamma^{-1} \cdot \mathsf{Z}_{\ST} \label{eq:maxProgBelowSt}
\end{align}
Overall, for any times $t_1, t_2$ that belong to the same phase we get from Equations \ref{eq:maxProgAboveSt},  \ref{eq:maxProgBelowSt}:
$$
\frac{\min\{ \mathcal{F}(t_1), \mathcal{F}(t_2)\} }
{\max\{ \mathcal{F}(t_1), \mathcal{F}(t_2)\} }
 \leq \left(\frac{1+\alpha_{\ST}}{1-\alpha_{\ST}}\right)\cdot \Gamma^2 =\Theta\left(\Gamma^2\right)
$$
\end{proof}

%%%%%%%%%%%%%%%%%%%%%%%%%%%%%%%%%%%%%%%
% Accuracy of stitching
%%%%%%%%%%%%%%%%%%%%%%%%%%%%%%%%%%%%%%%
\begin{lemma}[Estimation error]\label{lem:estimationErrorBound} 
Let $t\in [m]$ be a time step such that
\begin{enumerate}
    \item Assumption \ref{ass:accurateEstimators} holds for every $t'\leq t$.
    \item $\NoCapping{=} \text{True}$ during time $t$.
\end{enumerate}
Let $j\in [\beta+1]$ be the level of estimators used in time step $t$, and let $\mathsf{Z}$ be the value computed in Step \ref{algStep:FrozenBase}. 
Let $z^{1}_{j}, \dots, z^{\mathsf{K}_j}_{j}$ denote the estimations given by the estimators in level $j$. 
Then assuming a {\em good execution} (see Definition \ref{def:goodRun}), for at least $80\%$ of the indices $k\in[\mathsf{K}_j]$ we have
$$
\left|\mathcal{F}(t) - (\mathsf{Z} + z^{k}_j) \right| \leq \alpha_{\Stitch}\cdot \mathcal{F}(t),
$$
where $\alpha_{\Stitch} = \Gamma\cdot(\alpha_{\ST} + \beta\cdot \alpha_{\TDE})$.
\end{lemma}
\begin{proof} The estimation offset $\mathsf{Z}$ is computed in a different manner for the cases that $j = \beta$ and $j < \beta$, as the first is an offset of the strong tracker and the second is an offset of a \TDE\text{} of some level. We prove separately for these cases:

\paragraph{Case $j<\beta$.} On step \ref{algStep:FrozenBase} $\mathsf{Z}$ is computed using the subroutine \ref{alg:StitchFrozenValues}. The parameter that is passed to that subroutine is $\tau - 2^{j-1}+1$. Since $\StitchFrozenVals$ sums the frozen values with indices corresponding to the bits of the parameter that are set to $1$, such a parameter results in summing frozen values correspond to levels $j^{\prime}>j$ where $j$ is the active level. And so, $\mathsf{Z}$ consist of frozen values of levels $j^{\prime}>j$. These are levels that consisted the output that was modified on the time that level $j$ was enabled, that is on time $e_j$. 
Thus summing the estimations from level $j$ (that is one of the estimations $\{z_j^{k}\}_{k\in[\mathsf{K}_j]}$) to that value $\mathsf{Z}$ results in the current internal estimations to the value of the function $\mathcal{F}(t)$.
In order to bound $|\mathcal{F}(t) - (\mathsf{Z} + z_j^{k})|$ , we break the value $\mathcal{F}(t)$ into a telescopic series of differences, each difference correspond to freezing and enabling time of a certain level $j^{\prime}$ from the frozen levels that compose $\mathsf{Z}$. 
Let $J_{\mathsf{Z}}$ be the set of indexes of these levels and denote  $j_1> j_2 > \dots j_{\mathsf{Z}}$ their order (therefore if $J_{\mathsf{Z}}\neq \emptyset$, then $j_1 = \beta$ which is the level of the strong tracker).
\begin{align*}
\mathcal{F}(t) =&	\mathcal{F}(t)
					+ \left( \mathcal{F}(f_{j_1})			- \mathcal{F}(f_{j_1}) 		\right)
					+ \left( \mathcal{F}(f_{j_2})			- \mathcal{F}(f_{j_2})		\right)
		 + 	\dots	+ \left( \mathcal{F}(f_{j_{\mathsf{Z}}}) 	- \mathcal{F}(f_{j_{\mathsf{Z}}}) 	\right)\\
		\stackrel{1}{=}&	\mathcal{F}(t)	+ \left( \mathcal{F}(f_{j_1})		- \mathcal{F}(e_{j_2}) 							\right)
			+ \left( \mathcal{F}(f_{j_2})		- \mathcal{F}(e_{j_3}) 			\right)
		+ \dots 	+ \left( \mathcal{F}(f_{j_{\mathsf{Z}}}) 	- \mathcal{F}(e_{j}) 		\right)\\	
		\stackrel{2}{=}&	\mathcal{F}(f_{j_1})	+ \left( \mathcal{F}(f_{j_2}) 			- \mathcal{F}(e_{j_2})  	\right)
						+ \left( \mathcal{F}(f_{j_2}) 			- \mathcal{F}(e_{j_2})	\right) + \dots
						+ \left( \mathcal{F}(t) - \mathcal{F}(e_{j}) \right) \\
		\stackrel{3}{=}&	\mathcal{F}(f_{\ST}) + \sum_{i \in J_{\mathsf{Z}}\setminus \{\ST\}}{\left( \mathcal{F}(f_{i}) - \mathcal{F}(e_{i}) \right)}
		+ \left( \mathcal{F}(t) 	- \mathcal{F}(e_{j}) 		\right)
\end{align*}
Where (1) holds by noting that for levels $j_1 > j_2$, on the time that $j_1$ was frozen $j_2$ was enabled (see step \ref{algStep:EnableTDEs}) thus $f_{j_1} = e_{j_2}$. (2) is by reordering the terms and (3) is renaming index $j_1$ as \ST.\\
We now plug this alternative formulation of $\mathcal{F}(t)$ into the following:
\begin{align}
\left| \mathcal{F}(t) - (\mathsf{Z} + z_j^{k}) \right| =& \left| \mathcal{F}(t) - (\StitchFrozenVals(\tau-2^{j-1} +1) + z_j^{k}) \right| \nonumber\\
=&	\left| \mathcal{F}(t) - \left( \mathsf{Z}_{\ST}+\sum_{i\in J_{\mathsf{Z}}\setminus\{\ST\}} \mathsf{Z}_{i}  + z_j^{k}\right) \right| \nonumber\\
\leq&	\left| \mathcal{F}(f_{\ST}) - \mathsf{Z}_{\ST}\right| + \sum_{i\in J_{\mathsf{Z}}\setminus\{\ST\}}\left| \mathsf{Z}_{i} - \left( \mathcal{F}(f_{i}) - \mathcal{F}(e_{i}) \right)\right| + \left| z_j^{k} - \left( \mathcal{F}(t) - \mathcal{F}(e_{j}) \right) \right|\nonumber\\
\stackrel{1}{\leq}& \alpha_{\ST}\cdot\mathcal{F}(f_{\ST})  + \sum_{i\in J_{\mathsf{Z}}\setminus\{\ST\}}{\alpha_{\TDE}\cdot \mathcal{F}(e_{i})}  + \alpha_{\TDE}\cdot \mathcal{F}(e_{j})\nonumber\\
\stackrel{2}{\leq}& \Gamma\cdot\alpha_{\ST}\cdot\mathcal{F}(t)  + \Gamma\cdot \alpha_{\TDE}\cdot\sum_{i\in J_{\mathsf{Z}}\setminus\{\ST\}}{ \mathcal{F}(t)} + \Gamma\cdot \alpha_{\TDE}\cdot \mathcal{F}(t) \nonumber\\		
\stackrel{3}{\leq}& \Gamma\cdot (\alpha_{\ST} + \beta\cdot \alpha_{\TDE})\cdot \mathcal{F}(t)\nonumber\\
		=& \alpha_{\Stitch}\cdot \mathcal{F}(t)\nonumber
\end{align}
Where inequality (1) is by using Assumption \ref{ass:accurateEstimators} directly on the right difference term 
while other difference are due to the accuracy of the frozen values promised on a {\em good execution} (Definition \ref{def:goodRun}),
(2) is due to the ratio bound $\Gamma$, that is promised by Lemma \ref{lem:maxPhaseProgress}, between any function values of two times from the same phase.
(3) is due to the fact that $|J_{\mathsf{Z}}\setminus\{\ST\}| \leq \beta -1$ since $j \notin J_{\mathsf{Z}}$.
Last equality is by denoting $\alpha_{\Stitch} = \Gamma\cdot (\alpha_{\ST} + \beta\cdot \alpha_{\TDE})$.

\paragraph{Case $j=\beta$.} On step \ref{algStep:STBase}, $\mathsf{Z}$ is set to $0$, then directly from Assumption \ref{ass:accurateEstimators} we have:
$$
| \mathcal{F}(t) - (\mathsf{Z} + z_j^{k}) | = 
| \mathcal{F}(t) -  z_j^{k} | \leq \alpha_{\ST}\cdot \mathcal{F}(t) \leq \alpha_{\Stitch}\cdot \mathcal{F}(t)\text{.}
$$
\end{proof}

%%%%%%%%%%%%%%%%%%%%%%%%%%%%%%%%%%%%%%%
% subroutine - StitchFrozenVals
%%%%%%%%%%%%%%%%%%%%%%%%%%%%%%%%%%%%%%%
\begin{algorithm*}[ht]
\caption{\texttt{StitchFrozenVals($\tau$)}}
\makeatletter\def\@currentlabel{\texttt{StitchFrozenVals($\tau$)}}\makeatother
\label{alg:StitchFrozenValues}
{\bf Input:} A counter $\tau$.
{\bf Global Variables:} $\alpha$, $\mathsf{Z}_{\ST}$, $\mathsf{Z}_{j}$ for $j\in [\beta]$.

\begin{enumerate}
	\item $\FV \leftarrow \left\{ j\in [\beta] | \tau[j] = 1 \right\}$
	\item Return $\mathsf{Z}_{\ST} + \sum_{j\in \FV}{\mathsf{Z}_{\TDE,j}}$
\end{enumerate}
\end{algorithm*}

%%%%%%%%%%%%%%%%%%%%%%%%%%%%%%%%%%%%%%%
% Accuracy on update steps
%%%%%%%%%%%%%%%%%%%%%%%%%%%%%%%%%%%%%%%
By now we showed that given that Assumption \ref{ass:accurateEstimators} holds, then on step \ref{algStep:FrozenBase} we have a bound on the estimation error of at least $8/10$ out of level $j$ estimations of $\mathcal{F}(t)$ (Lemma \ref{lem:estimationErrorBound})
and that whenever mechanism $\PrivateMed$ is activated (on step \ref{algStep:PrivMed}), then its output for level $j$, $\mathsf{Z}_j$ is accurate (Lemma \ref{lem:FrozenAccuracy}). Combining these lemmas results in the following corollary:
\begin{corollary}[Accuracy on output modification]\label{cor:outputUpdateAccuracy}
Let $t\in [m]$ be a time step such that
\begin{enumerate}
    \item Assumption \ref{ass:accurateEstimators} holds for every $t'\leq t$.
    \item $\NoCapping{=} \text{True}$ during time $t$.
    \item Algorithm {\bf \PrivateMed} was activated during time $t$ (on Step~\ref{algStep:PrivMed}\text{}).
\end{enumerate}
Then assuming a {\em good execution} (see Definition \ref{def:goodRun})
we have
$$
|\Output(t) - \mathcal{F}(t)| \leq \alpha_{\Stitch}\cdot \mathcal{F}(t),
$$
where $\alpha_{\Stitch} = \Gamma\cdot(\alpha_{\ST} + \beta\cdot \alpha_{\TDE})$.
\end{corollary}
\begin{proof} During an output modification step, we update the value of $\tau$. Denote $\tau^{\text{pre}}, \tau^{\text{post}}$ the values of $\tau$ before and after this update.

\paragraph{Case $j < \beta$.} For the case that the active level that was frozen was a $\TDE$ level, we look on the frozen values after mechanism $\PrivateMed$ was activated on step \ref{algStep:PrivMed}. 
Let $J_{\mathsf{Z}}$ be the set of indexes for levels of the frozen values that compose $\mathsf{Z}$ (computed by \ref{alg:StitchFrozenValues} in step \ref{algStep:FrozenBase}).
Then:
\begin{align*}
    |\mathcal{F}(t) - \Output(t)| & = |\mathcal{F}(t) - \StitchFrozenVals(\tau^{\text{post}})| \\
    &= |\mathcal{F}(t) - \StitchFrozenVals(\tau^{\text{pre}} + 1)| \\
    &= |\mathcal{F}(t) - (\StitchFrozenVals(\tau^{\text{pre}} + 1) -\mathsf{Z}_j +\mathsf{Z}_j)| \\
    &= |\mathcal{F}(t) - (\StitchFrozenVals(\tau^{\text{pre}} -2^{j} + 1)  +\mathsf{Z}_j)| \\
    &\stackrel{1}{=} |\mathcal{F}(t) - (\mathsf{Z}  +\mathsf{Z}_j)| \\
    &\stackrel{}{=} |-\mathsf{Z}  + \mathcal{F}(e_j) + \mathcal{F}(t) - \mathcal{F}(e_j) -\mathsf{Z}_j| \\
    &\leq |\mathsf{Z} - \mathcal{F}(e_j)| + |(\mathcal{F}(t) - \mathcal{F}(e_j)) - \mathsf{Z}_j| \\
    &\leq	\left| \mathcal{F}(f_{\ST}) - \mathsf{Z}_{\ST}\right| + \sum_{i\in J_{\mathsf{Z}}\setminus\{\ST\}}\left| \mathsf{Z}_{i} - \left( \mathcal{F}(f_{i}) - \mathcal{F}(e_{i}) \right)\right| + |\mathsf{Z}_j - \left( \mathcal{F}(t) - \mathcal{F}(e_{j}) \right) | \nonumber\\
    &\stackrel{2}{\leq}	 \alpha_{\ST}\cdot \mathcal{F}(f_{\ST}) + \sum_{i\in J_{\mathsf{Z}}\setminus\{\ST\}}\alpha_{\TDE}\cdot \mathcal{F}(e_{i}) + \alpha_{\TDE}\cdot\mathcal{F}(e_{j}) \nonumber\\
    &\stackrel{3}{\leq}	 \Gamma\cdot \alpha_{\ST}\cdot \mathcal{F}(t) + \Gamma\cdot\alpha_{\TDE} \sum_{i\in J_{\mathsf{Z}}\setminus\{\ST\}} \mathcal{F}(t) + \Gamma\cdot\alpha_{\TDE}\cdot \mathcal{F}(t) \nonumber\\
    &\stackrel{4}{\leq} \Gamma\cdot (\alpha_{\ST} + \beta\cdot  \alpha_{\TDE})\mathcal{F}(t)\\
    &= \alpha_{\Stitch}\cdot \mathcal{F}(t)
\end{align*}
Where (1) is true for modification of levels smaller then $\ST$ since the values composing $\mathsf{Z}$ are of levels that did not change after updating $\tau$,
(2) holds in a {\em good execution} (Definition \ref{def:goodRun}),
(3) is by Lemma \ref{lem:maxPhaseProgress},
(4) is due to the fact that $|J_{\mathsf{Z}}\setminus\{\ST\}| \leq \beta -1$ Since $j\notin J_{\mathsf{Z}}$.

\paragraph{Case $j = \beta$.} For the case that the active level that was frozen was an $\ST$ level, then the output is $\mathsf{Z}_{\ST}$ which is accurate for the case of a {\em good execution} (Definition \ref{def:goodRun}). We have 
$$|\mathcal{F}(t) - \mathsf{Z}_{\ST}|\leq \alpha_{\ST}\cdot \mathcal{F}(t) \leq \alpha_{\Stitch}\cdot \mathcal{F}(t)$$
\end{proof}
The following lemma is arguing about the output accuracy on all time $t\in [m]$ (not only on output modification steps).
%%%%%%%%%%%%%%%%%%%%%%%%%%%%%%%%%%%%%%%
% Accuracy on update steps
%%%%%%%%%%%%%%%%%%%%%%%%%%%%%%%%%%%%%%%
\begin{lemma}[Output accuracy]\label{lem:outputAccuracy}
Let $t\in [m]$ be a time step such that
\begin{enumerate}
    \item Assumption \ref{ass:accurateEstimators} holds for every $t'\leq t$.
    \item $\NoCapping{=} \text{True}$ during time $t$.
\end{enumerate}
Then assuming a {\em good execution} (see Definition \ref{def:goodRun}) we have
$$
|\Output(t) - \mathcal{F}(t)| \leq \alpha \cdot \mathcal{F}(t),
$$
provided that $\alpha_{\ST} = O(\alpha)$, $\alpha_{\TDE} = O(\alpha/\log(\alpha^{-1}))$ and for all $j\in [\beta+1]\cup\{\STwrapper\}$, $\mathsf{K}_{j} = \Omega \left( \frac{1}{\eps_j}\log\left( \frac{m}{\delta^{N}}\right) \right)$.
\end{lemma}
\begin{proof} We prove for two cases of execution types: one is a step execution without an output-modification, and the second is an execution that generates an output-modification.

\paragraph{Case 1 (no output-modification):} If on time $t$ we do not modify the output (the condition in step \ref{algStep:AboveThresh} was not satisfied), then assuming a {\em good execution} imply a bounded noise magnitude and we have that:
$$\left| \left\{ k\in [\mathsf{K}_j] : \mathsf{Z}+z_{j}^{k} \in \Output(t-1) \pm \mathsf{Z}_{\ST}\cdot \StepSize(\alpha)  \right\} \right| \geq \frac{\mathsf{K}_{j}}{2} -  \frac{4}{\eps_j}\log\left( \frac{4m}{\delta^{N}}  \right)
\geq \frac{4\cdot \mathsf{K}_{j}}{10}$$
where the last inequality follows by asserting that 
$$
\mathsf{K}_{j} = \Omega \left( 
\frac{1}{\eps_j}\log\left( \frac{m}{\delta^{N}}\right)
\right)
= \Omega \left( \frac{1}{\eps}
\sqrt{P_{j} \cdot \log \left(\frac{1}{\delta^{\prime}} \right) }\log\left( \frac{m}{\delta^{N}}\right)
\right)
$$
So, for at least $4\mathsf{K}_{j}/10$ of the estimations $z_{j}^k$ we have that 
$|(\mathsf{Z}+z_{j}^{k}) - \Output(t-1)| \leq \mathsf{Z}_{\ST}\cdot \StepSize(\alpha)$.
On the other hand, 
by the assumption on the accuracy of the estimators (Assumption \ref{ass:accurateEstimators}) 
we have that the requirement for Lemma \ref{lem:estimationErrorBound} met, therefore for at least $8\mathsf{K}_{j}/10$ of the estimations $z_{j}^k$ we have that $|\mathcal{F}(t) - (\mathsf{Z} + z^{k}_j)| \leq \alpha_{\Stitch}\cdot \mathcal{F}(t)$
Therefore, there must exist an index $k$ that satisfies both conditions, and so:
\begin{align*}
|\mathcal{F}(t) - \Output(t-1)| &\leq |\Output(t-1) - (\mathsf{Z}+z_{j}^{k})| + |\mathcal{F}(t) - (\mathsf{Z} + z^{k}_j)|\\
&\leq \StepSize(\alpha)\cdot \mathsf{Z}_{\ST} + \alpha_{\Stitch}\cdot \mathcal{F}(t) \\
&\stackrel{1}{\leq} \Gamma\cdot \StepSize(\alpha)\cdot \mathcal{F}(t) + \alpha_{\Stitch}\cdot \mathcal{F}(t) \\
&\stackrel{2}{\leq} \alpha\cdot \mathcal{F}(t) \\
\end{align*}
where (1) is due to Lemma \ref{lem:maxPhaseProgress}
and (2) holds for $\alpha_{\Stitch}\leq\frac{1}{10}\StepSize(\alpha)$, and $\StepSize(\alpha)\leq (2\Gamma)^{-1}\alpha$. That imply $\alpha_{\Stitch} \leq \alpha/(20\Gamma)$. Since $\alpha_{\Stitch} = \Gamma\cdot (\alpha_{\ST} + \beta \alpha_{\TDE})$ it is sufficient to set $\alpha_{\ST} = O(\alpha)$ and $\alpha_{\TDE} = O(\alpha/\beta) = O(\alpha/\log(\alpha^{-1}))$ to get
we have that $|\mathcal{F}(t) - \Output(t-1)|\leq \frac{3}{4}\alpha$. 
Therefore the output is accurate for not updating the output.

\paragraph{Case 2 (an output-modification):} If on time $t$ we {\em do} modify the output (the condition in step \ref{algStep:AboveThresh} was  satisfied) then Algorithm {\bf \PrivateMed} was activated during time $t$ (on Step~\ref{algStep:PrivMed}\text{}). In that case the requirements of Corollary \ref{cor:outputUpdateAccuracy} are met and we have (for $\alpha_{\Stitch} \leq \alpha$):
$$
|\Output(t) - \mathcal{F}(t)| 
\leq \alpha_{\Stitch}\cdot \mathcal{F}(t) 
\leq \alpha \cdot \mathcal{F}(t) 
$$
\end{proof}

%%%%%%%%%%%%%%%%%%%%%%%%%%%%%%%%%%%%%%%
% Formal proofs - Calibrating P_j
%%%%%%%%%%%%%%%%%%%%%%%%%%%%%%%%%%%%%%%
\subsection{Calibrating to avoid capping}\label{sec:formalCal}
In this section we calculate the needed calibration of parameters $P_j$ of \texttt{RobustDE} in order to avoid capping before the input stream ends.
In order to avoid capping we need to calibrate for each of the estimators levels a sufficient privacy budget. That  budget is derived from the number of output modification associated with each of these levels. 
At a high level, the calculation on these numbers per level is done as follows: recall our framework operates in phases. In each phase we bound the number of output modification for each of the estimators levels $j\in [\beta +1]\cup \{\STwrapper\}$. In addition we also bound the total number of phases. And so, the total number of output modification associated with each level results by multiplication of these bounds. 
This calculation is analysed w.r.t the framework level selection management (subroutine \ref{alg:ActiveLVL} and the state of $\tau$).
The following definition captures the number of output modifications we wish to bound w.r.t an input stream for algorithm \texttt{RobustDE}:
%%%%%%%%%%%%%%%%%%%%%%%%%%%%%%%%%%%%%%%
% Capping definitions
%%%%%%%%%%%%%%%%%%%%%%%%%%%%%%%%%%%%%%%
\begin{definition}
For every level $j\in [\beta +1]\cup \{\STwrapper\}$ and every time step $t\in[m]$, let $C_j(t)$ denote the number of time steps $t'\leq t$ during which 
\begin{enumerate}
    \item Level $j$ was selected.
    \item The output is modified.
\end{enumerate}
\end{definition}

The lemma that bounds these quantities is Lemma \ref{lem:outputUpdatesBounds}, and it is the main lemma  of this section. 
A central part in that lemma is to upper bound the number of output modifications done by algorithm \texttt{RobustDE} for some stream segment. Lemma \ref{lem:consequtiveOutputUpdatesProgress} is useful for that.
Additional lemma is needed: A phase can also be terminated before its predefined length (i.e. \PhaseSize) by a phase reset. And so, in the analysis we focus on each of the stream segments between such resets. The analysis of these segments requires a bound on their flip number. This is enabled via Lemma \ref{lem:substreamFlipNumber} that bounds the flip number of a sub stream:

\begin{lemma}[flip number of sub stream] \label{lem:substreamFlipNumber}
Let $\SSS$ be a stream with a $(\alpha,m)$-flip number denoted by $\lambda_{\SSS}$.
Let $\TTT = \{t_i\}_{i\in [\lambda_{\SSS}]}$, $t_i < t_{i+1}$, $t_i \in [m]$ be a set of time steps s.t. for all $i\in [\lambda_{\SSS}]$,
$|\mathcal{F}(t_{i}) - \mathcal{F}(t_{i+1})| \geq \alpha \cdot \mathcal{F}(t_{i})$.
Let $\PPP$ be a sub-stream of $\SSS$ from time $r_1$ to time $r_2>r_1$, $r_1, r_2 \in [t]$. 
Let $\lambda^{\prime} = |\{j\in [\lambda_{\SSS}] : \{t_j\}_{j\in [\lambda_{\SSS}]}, r_1\leq t_j<r_2\}|$ and let $\lambda_{\PPP}$ be the $(\alpha, m)$-flip number of $\PPP$. Then:
$$
\lambda^{\prime} \leq \lambda_{\PPP} \leq \lambda^{\prime} + 2
$$
\end{lemma}

\begin{proof}
Fix a stream $\SSS$ with a $(\alpha, m)$-flip number denoted as $\lambda_{\SSS}$, and fix some set of time steps 
$\TTT = \{t_i\}_{i\in [\lambda_{\SSS}]}$, $t_i < t_{i+1}$, $t_i \in [m]$ s.t. for all $i\in [\lambda_{\SSS}]$, $|\mathcal{F}(t_{i}) - \mathcal{F}(t_{i+1})| \geq \alpha \cdot \mathcal{F}(t_{i})$ with respect to $\SSS$.
We look on time steps from $\TTT$ that reside in $[r_1, r_2)$, that is $\TTT \cap [r_1, r_2)$.
Then, any set of times in $t_j\in [r_1,r_2)$, $t_j < t_{j+1}$ with $|\mathcal{F}(t_{i}) - \mathcal{F}(t_{i+1})| \geq \alpha \cdot \mathcal{F}(t_{i})$ is at size at most $|\TTT \cap [r_1, r_2)|$. Otherwise it could be used to construct along with $\TTT \setminus [r_1, r_2)$ a set of $\alpha$-jumps times in $\SSS$ larger then $\lambda_{\SSS}$ contradicting the maximality of the flip number of $\SSS$ being $\lambda_{\SSS}$. 
That is:
$$
\lambda^{\prime} \leq \lambda_{\PPP} \text{.}
$$
Now, denote the smallest index by $f = \argmin \{t_j \in \TTT \cap [r_1, r_2)\}$ (for {\em first}) and the largest index by $l = \argmin \{t_j \in \TTT \cap [r_1, r_2)\}$ (for {\em last}). Then we can have at most additional two $\alpha$-jumps. One from time $r_1$ to time $f$ and the second from time $l$ to time $r_2-1$ (regardless of the choice of $\TTT$). That is:
$$
\lambda_{\PPP} \leq \lambda^{\prime} +2 \text{.}
$$
\end{proof}

\begin{lemma}[Function value progress between output-modifications]\label{lem:consequtiveOutputUpdatesProgress}
Let $t_1< t_2\in [m]$ be consecutive times in which the output is modified (i.e., the output is modified in each of these two iterations, and is not modified between them), where 
\begin{enumerate}
    \item Assumption \ref{ass:accurateEstimators} holds for every $t'\leq t_2$.
    \item $\NoCapping{=} \text{True}$ during time $t_2$.
    
    \item $\tau\neq0$ during time $t_2$.
\end{enumerate}
Then, assuming a {\em good execution} (see Definition \ref{def:goodRun}) we have:
\begin{itemize}
    \item $|\mathcal{F}(t_2) - \mathcal{F}(t_1)| \geq \StepSize(\alpha)\cdot \mathsf{Z}_{\ST} - 2\cdot \alpha_{\Stitch}\cdot \max\{\mathcal{F}(t_1), \mathcal{F}(t_2)\}$
    \item $|\mathcal{F}(t_2) - \mathcal{F}(t_1)| \leq \StepSize(\alpha)\cdot \mathsf{Z}_{\ST} +  (2\cdot \alpha_{\Stitch} + \MaxUpdateSize(\alpha))\cdot \max\{\mathcal{F}(t_1), \mathcal{F}(t_2)\}$
\end{itemize}
\end{lemma}
\begin{proof} 
Upper bound and lower bound of the function value between such times is analysed separately. First we analyse the lower bound of such progress and then we analyse the upper bound of it.

\paragraph{Minimum progress.} We look on the time $t_2$. Let $j$ be the level of estimators used in time $t_2$. 
On that time we modify the output, which means that during that time the condition on step \ref{algStep:AboveThresh} was satisfied. In such case, assuming a good execution we have bounded noise. That means that by asserting that $
\mathsf{K}_{j} = \Omega \left( \frac{1}{\eps}
\sqrt{P_{j} \cdot \log \left(\frac{1}{\delta^{\prime}} \right) }\log\left( \frac{m}{\delta^{N}}\right)
\right)\text{,}
$
at least $40\%$ of the estimations of level $j$ admit $|(\mathsf{Z}+z_{j}^{k}) - \Output(t_2-1)| \geq \mathsf{Z}_{\ST}\cdot \StepSize(\alpha)$. Since Assumption \ref{ass:accurateEstimators} holds, the requirements for Lemma \ref{lem:estimationErrorBound} are met and we have for at least $80\%$ of these estimations: $|\mathcal{F}(t_2) - (\mathsf{Z} + z^{k}_j) | \leq \alpha_{\Stitch}\cdot \mathcal{F}(t_2)$. Thus at least one index $k$ admit both inequalities and therefore:
\begin{align}
    |\mathcal{F}(t_2) - \Output(t_2-1)| \geq \StepSize(\alpha)\cdot \mathsf{Z}_{\ST} - \alpha_{\Stitch}\cdot \mathcal{F}(t_2)\label{eq:progressOnUpdate}
\end{align}
Now, observe\footnote{Calibrating $\StepSize(\alpha) > 2\cdot \MaxUpdateSize(\alpha)$ ensures that for any two consecutive time s.t. the output is modified there must be at least one time step between them.} that since there was no output modification between times $t_1, t_2$ then $\Output(t_2 - 1) = \Output(t_1)$. 
Applying Lemma \ref{cor:outputUpdateAccuracy} time $t_1$ (and setting $\Output(t_1)\leftarrow \Output(t_2-1)$) we get:
\begin{align}
|\Output(t_2-1) - \mathcal{F}(t_1)| \leq \alpha_{\Stitch}\cdot \mathcal{F}(t_1)    \label{eq:outputAccuracyOnUpdate}
\end{align}
Combining the established equations \ref{eq:progressOnUpdate}, \ref{eq:outputAccuracyOnUpdate} we get:
\begin{align*}
|\mathcal{F}(t_2) - \mathcal{F}(t_1)| &\geq \StepSize(\alpha)\cdot \mathsf{Z}_{\ST} - (\alpha_{\Stitch}\cdot \mathcal{F}(t_1) + \alpha_{\Stitch}\cdot \mathcal{F}(t_2)) \\
&\geq \StepSize(\alpha)\cdot \mathsf{Z}_{\ST} - 2\cdot \alpha_{\Stitch}\cdot \max\{\mathcal{F}(t_1), \mathcal{F}(t_2)\} \text{.}
\end{align*}

\paragraph{Maximum progress.} We now focus on times $t_2, t_2-1$. Let $j$ be the level of estimator that is used on time $t_2-1$.
Since on time $t_2-1$ we did not modify the output, then the condition on step \ref{algStep:AboveThresh} did not trigger. That means that for at least $40\%$ of the estimations $z^{k}_j$ of level $j$ the following holds: $| (\mathsf{Z}+z^{k}_j) - \Output(t_2-1) | \leq \mathsf{Z}_{\ST} \cdot \StepSize(\alpha) $. Since in addition the output did not change between $t_1$ to $t_2-1$ (thus $\Output(t_1) = \Output(t_2-1)$) and by applying Corollary \ref{cor:outputUpdateAccuracy} on time $t_1$ we have that the output on time $t_1$ was $\alpha_{\Stitch}$-accurate. And so we get that for at least (same) $40\%$ estimations that are used on time $t_2-1$ the following holds:
\begin{align}
    \left| (\mathsf{Z}+z^{k}_j) - \mathcal{F}(t_1) \right| \leq \mathsf{Z}_{\ST} \cdot \StepSize(\alpha) + \alpha_{\Stitch}\cdot \mathcal{F}(t_1) \label{eq:t2-1Tot_1Bound}
\end{align}
Now, applying on time $t_2 -1$ Lemma \ref{lem:estimationErrorBound} (since assumption \ref{ass:accurateEstimators} holds) we get that for at least $80\%$ of the estimations $z^{k}_j$ of level $j$ the following holds: $| (\mathsf{Z} + z^{k}_j) - \mathcal{F}(t_2 -1) | \leq \alpha_{\Stitch} \cdot \mathcal{F}(t_2 -1)$.
Recalling the maximum update size assumption (see Condition \ref{con:streamMaxUpdateSize}) we get a bound on the progress of the value of the function $\mathcal{F}$ between the adjacent times $t_2, t_2-1$: $| \mathcal{F}(t_2) - \mathcal{F}(t_2-1) | \leq \MaxUpdateSize(\alpha) \cdot \mathcal{F}(t_2)$.
And so, for at least $80\%$ of the estimation $z^{k}_j$ used on time $t_2-1$ the following holds:
\begin{align}
    \left| (\mathsf{Z} + z^{k}_j) - \mathcal{F}(t_2) \right| 
    \leq 
    \alpha_{\Stitch} \cdot \mathcal{F}(t_2 -1) + \MaxUpdateSize(\alpha)\cdot \mathcal{F}(t_2) \label{eq:t_2-1Tot_2Bound}
\end{align}
Equations \ref{eq:t2-1Tot_1Bound}, \ref{eq:t_2-1Tot_2Bound} hold for $40\%$ and $80\%$ of the estimations $z^{k}_j$ of time $t_2-1$ respectively. And so, for at least one of these estimations both equations hold and we get:
\begin{align*}
|\mathcal{F}(t_2) - \mathcal{F}(t_1)| &\leq \StepSize(\alpha)\cdot \mathsf{Z}_{\ST} + \alpha_{\Stitch}\cdot(\mathcal{F}(t_1) +  \mathcal{F}(t_2-1)) + \MaxUpdateSize(\alpha)\cdot \mathcal{F}(t_2) \\
&\leq \StepSize(\alpha)\cdot \mathsf{Z}_{\ST} + (2\cdot \alpha_{\Stitch}+\MaxUpdateSize(\alpha))\cdot \max\{\mathcal{F}(t_1), \mathcal{F}(t_2)\}
\end{align*}
\end{proof}
The following lemma is using Lemmas  \ref{lem:substreamFlipNumber},  \ref{lem:consequtiveOutputUpdatesProgress} to bound the total number of output modification $C_j$ for each estimators levels $j\in [\beta+1]\cup \{\STwrapper\}$.
\begin{remark}
Recall that once $\NoCapping{=} \text{False}$, then the output never changes. Therefore, if during some time $\hat{t}$ we have that $\NoCapping{=} \text{False}$, then $C_j(\hat{t})=C_j(\hat{t}+1)$.
\end{remark}

%%%%%%%%%%%%%%%%%%%%%%%%%%%%%%%%%%%%%%%
% output modifications per level
%%%%%%%%%%%%%%%%%%%%%%%%%%%%%%%%%%%%%%%
\begin{lemma}[Output modifications of each level]\label{lem:outputUpdatesBounds} 
Let $\SSS$ be the input stream of length $m$ for algorithm \ref{alg:ROEF} with a flip number $\lambda_{\alpha^{\prime}}(\SSS)$ and let $t\in [m]$ be a time step such that Assumption \ref{ass:accurateEstimators} holds for every $t'\leq t$.  Then, assuming a {\em good execution} (see Definition \ref{def:goodRun}), for every level $j\in [\beta+1]\cup \{\STwrapper\}$ we have
$$
C_j(t) \leq O\left(\frac{\lambda_{\alpha^{\prime}}(\SSS)}{2^j}\right),
$$
where $\alpha^{\prime} = (1/2)\cdot \StepSize(\alpha) = O(\alpha)$.
\end{lemma}

\begin{proof} We bound the number of output modifications for each level by bounding the number of phases and then multiplying it with the number of output modifications of each level within a phase. The later is done in the last part of the proof while bounding the number of phases is the main part of the proof.

\paragraph{Bounding the number of phases.} Whenever a phase starts, the previous phase is terminated. We elaborate on the two cases of phase termination and count them separately.
A phase starts whenever an ST level estimators  (i.e.\ $\beta$) are selected in \ref{alg:ActiveLVL}. That happens whenever $\tau[\beta] = 0$ which happens in two cases:
\begin{enumerate}[leftmargin=40pt]
    \item[({\rm C1}):]\quad A {\em phase end}: $\tau\neq 0, \tau[\beta]=0$. When Step \ref{algStep:tauUpdateEndPhase} was executed on previous time step.
    \item[({\rm C2}):]\quad A {\em phase reset}: $\tau=0$. When condition in Step \ref{algStep:STwrapper} is True.
\end{enumerate}
And so in (C1) previous phase reached its end while in (C2) previous phase is terminated before its ending due to a phase reset.

\paragraph{Number of phase resets.} When the condition in Step \ref{algStep:STwrapper} is True it holds that the value of the target function $\mathcal{F}$ has changed by a constant multiplicative factor $\Gamma$ ($\Gamma \geq 2$) compared to what it was in the beginning of the terminated phase. 
By the assumption of the flip number of $\SSS$, this can happen at most $O(\alpha\cdot \lambda_{\alpha', m}(\SSS) )$ times.
That is, the number of phase resets is bounded by:
\begin{align}
O(\alpha\cdot \lambda_{\alpha', m}(\SSS) ) \label{eq:lvlOutputMod_eq1}
\end{align}

\paragraph{Number of output modifications between resets.} We bound the number of output modifications between two consecutive times where a phase reset was executed (C2). Denote two such consecutive times where $\tau=0$ by $r_{i}<r_{i+1}$ and let $\SSS_i$ be the segment of $\SSS$ for the times $[r_{i}, r_{i+1})$ with an $\alpha^{\prime}$-flip number $\lambda_{\alpha^{\prime}}(\SSS_i)$. 
We bound the number of output modifications in $[r_{i}, r_{i+1})$ by looking at two consecutive time steps where the output is modified $t_1 < t_2$, s.t.\ $r_{i}\leq t_1 < t_2 < r_{i+1}$. That is, the output is modified in times $t_1,t_2$ and is not modified between them. Then we have
\begin{align*}
|\mathcal{F}(t_2) - \mathcal{F}(t_1)| &\stackrel{1}{\geq} 
\StepSize(\alpha)\cdot \mathsf{Z}_{\ST} - 2\cdot \alpha_{\Stitch}\cdot \max\{\mathcal{F}(t_1), \mathcal{F}(t_2)\}\\
&\stackrel{2}{\geq} \StepSize(\alpha)\cdot \frac{1-\alpha_{\ST}}{\Gamma}\cdot\max\{\mathcal{F}(t_1), \mathcal{F}(t_2)\}  - 2\cdot \alpha_{\Stitch}\cdot \max\{\mathcal{F}(t_1), \mathcal{F}(t_2)\} \\
&\stackrel{3}{\geq} (1/2) \cdot \StepSize(\alpha)\cdot \max\{\mathcal{F}(t_1), \mathcal{F}(t_2)\}
\end{align*}
(1) is by Lemma~\ref{lem:consequtiveOutputUpdatesProgress}, 
(2) holds due to the ratio checked in step \ref{algStep:STwrapper} thus the ratio holds for any time of that phase
(3) holds whenever $\alpha_{\ST} \leq 1/3$ and $\alpha_{\Stitch}\leq  (1/12\Gamma)\cdot \StepSize(\alpha)$.
That is, in every time of such output modification the true value of the target function is changed by a multiplicative factor of at least $(1\pm\alpha')$. 
Thus, for every segment $\SSS_i$ algorithm \texttt{RobustDE} can have at most $\lambda_{\alpha^{\prime}}(\SSS_{i})$ such output modifications. Since in every such segment we have a single phase reset, and a single additional output modification that results from it, we have:
\begin{align}
C(\SSS_i) \leq \lambda_{\alpha^{\prime}}(\SSS_{i}) + 1 \label{eq:lvlOutputMod_eq2}
\end{align}

\paragraph{Output modifications in a phase.} We now show that by algorithm \texttt{RobustDE} management of phases start/ end time, a phase that ends without a reset termination (that is in case (C1)), has $O(\PhaseSize)$ number of output modifications (phases of case (C2) are shorter). 
That management is done on Steps \ref{algStep:tauUpdateStartPhase}, \ref{algStep:tauUpdateInnerPhase}, \ref{algStep:tauUpdateEndPhase} according to $\tau$ which indicates the number of output modifications in a phase.
We now elaborate on that management for case (C1):
\begin{enumerate}
    \item Step \ref{algStep:tauUpdateStartPhase} (Stating a new phase) Setting \ST\text{} bit in $\tau$ to $1$ to indicate a new value for $\mathsf{Z}_{\ST}$. This also set all lower bits of $\tau$ to $0$ which indicated that there are no frozen value for levels $j<\beta$.
    \item Step \ref{algStep:tauUpdateInnerPhase} (Inner phase step) Increment the value of $\tau$ by $+1$ to indicate additional step of the current phase.
    \item Step \ref{algStep:tauUpdateEndPhase} (Ending phase) Setting the \ST\text{} bit of $\tau$ to $0$ to indicate that the phase has ended and next estimator level used will be \ST\text{}.
\end{enumerate}
That is, the counting cycle of a phase is managed on the lower bits ($[0,\beta)$) of $\tau$: These bits are set to zero on the beginning of the phase. Then for each output modification $\tau$ is incremented by $1$. The cycle ends when the value of these bits equals \PhaseSize.
Accounting the output modification done on a phase  start, the number of output modifications in a phase that ends in case (C1) is $\PhaseSize + 1$.

\paragraph{Total number of phases.} For $\kappa$ number of phase resets executed in times $r_0<r_1\dots<r_{\kappa}$ we have $\kappa + 1$ sub-streams of $\SSS$ corresponding to times $[r_i,r_{i+1})$ denoted by $\SSS_i$ for $i\in [\kappa +1]$. Denote by $\phi, \phi_i$ the number of phases in $\SSS, \SSS_i$ correspondingly.
The following holds:
\begin{align*}
\phi &\leq \sum_{i\in [\kappa+1]} \phi_i 
\stackrel{1}{=} \sum_{i\in [\kappa+1]} \left\lceil \frac{C(\SSS_i)}{\PhaseSize+1} \right\rceil 
\leq \sum_{i\in [\kappa+1]} \left( \frac{C(\SSS_i)}{\PhaseSize+1} + 1 \right) \\
&\stackrel{2}{\leq} (\kappa + 1) + \frac{\sum_{i\in [\kappa+1]}  \lambda_{\alpha^{\prime}}(\SSS_i)+1}{\PhaseSize+1}  
\stackrel{3}{\leq} (\kappa + 1) + \frac{\lambda_{\alpha^{\prime}}(\SSS) + 3\kappa}{\PhaseSize+1}  
\stackrel{4}{=} O\left(\frac{\lambda_{\alpha^{\prime}}(\SSS)}{\PhaseSize}  \right)
\end{align*}
where (1) is true since on each segment $\SSS_i$ there is no phase reset and we start a new phase every $\PhaseSize+1$ number of steps, %
(2) holds by Equation \ref{eq:lvlOutputMod_eq2},
(3) is true by Lemma \ref{lem:substreamFlipNumber} 
and (4) is true since by Equation \ref{eq:lvlOutputMod_eq1} we have that $\kappa = O\left(\alpha \lambda_{\alpha^{\prime}(\SSS)} \right)$ and $\PhaseSize = O(\alpha^{-1})$.

\paragraph{Number of ouput modification for level $j$.} The levels $j\in [\beta + 1]$ are selected in \ref{alg:ActiveLVL} according to $\tau,$ s.t. $j$ is the LSB of $\tau + 1$. Since on every output modification we increment the value of $\tau$ by $+1$ then level $j=0$ is  selected every second time, level $j=1$ is selected every forth time, level $j=2$ is selected every eighth time and so on. That is a total of $O(\PhaseSize/2^{j})$ for level $j$. 
Multiplying the established bound for $\phi$ (the total number of phases) with that bound of the number of output modification of level $j$ we get:
$$
C_j  =  O\left(\frac{\lambda_{\alpha^{\prime}}(\SSS)}{\PhaseSize}\right) \cdot O \left( \frac{\PhaseSize}{2^{j}} \right)= O\left(\frac{\lambda_{\alpha^{\prime}}(\SSS)}{2^j}\right)
$$
\end{proof}
%
%%%%%%%%%%%%%%%%%%%%%%%%%%%%%%%%%%%%%%%
% Corollary: No capping
%%%%%%%%%%%%%%%%%%%%%%%%%%%%%%%%%%%%%%%
\begin{corollary}\label{cor:NoCapping}
Provided that $\lambda > \lambda_{\alpha^{\prime}}(\SSS)$ then algorithm \ref{alg:ROEF} will not get to {\em capping} state by calibrating: 
$$
P_j = \Omega\left(\frac{\lambda}{2^j}\right)\text{.}
$$
\end{corollary}

%%%%%%%%%%%%%%%%%%%%%%%%%%%%%%%%%%%%%%%
% subroutine - ActiveLVL
%%%%%%%%%%%%%%%%%%%%%%%%%%%%%%%%%%%%%%%
\begin{algorithm*}[ht]

\caption{\texttt{ActiveLVL($\tau$)}}
\makeatletter\def\@currentlabel{\texttt{ActiveLVL($\tau$)}}\makeatother
\label{alg:ActiveLVL}
{\bf Input:} A counter $\tau$. {\bf Global parameter:} $\beta$.
\begin{enumerate}
    \item If $\tau[\beta] = 0$ : Return $\beta$ \qquad \qquad \qquad \text{\% Selecting the \ST\text{} level}
	\item Else Return The LSB of $(\tau+1)$ \qquad \text{\% Selecting a TDE level}
\end{enumerate}
\end{algorithm*}

%%%%%%%%%%%%%%%%%%%%%%%%%%%%%%%%%%%%%%%
% Framework robustness
%%%%%%%%%%%%%%%%%%%%%%%%%%%%%%%%%%%%%%%
\subsection{The framework is robust}\label{sec:formalRobust}
We move on to show that the framework \texttt{RobustDE} is robust for adaptive inputs. Lemma \ref{lem:AccurateEstimations} (adaptation of Lemma 3.2 \cite{hassidim2020adversarially}) uses tools from differential privacy to show that if the framework preserve privacy with respect to the random strings of the estimators, then the estimators yield accurate estimations. 
Yet in our case the accuracy of estimators of levels $j<\beta$ (i.e.\ TDE levels) have an additional requirement: they must also estimate differences that are within their range (see requirement \ref{req:TDEdiffRange}).
We show in Lemma \ref{lem:boundedEstimationRanges} that indeed whenever estimator of level $j$ is being used by the framework, then it is estimating a difference that is within its accuracy range.
\begin{lemma}[bounded estimation ranges]\label{lem:boundedEstimationRanges}
Let $t\in [m]$ be a time step such that
\begin{enumerate}
    \item Level $j \in [\beta]$ was selected (a \TDE).
    \item Assumption \ref{ass:accurateEstimators} holds for every $t'< t$.
\end{enumerate}
Denote by $e_j$ the last time step during which level $j$ estimators were enabled. Then, assuming a {\em good execution} (see Definition \ref{def:goodRun}), the following holds:
$$
|\mathcal{F}(t) - \mathcal{F}(e_j)| \leq \gamma_j\cdot  \mathcal{F}(e_j)
$$
where $\gamma_j =  \frac{1+\alpha_{\ST}}{1-\alpha_{\ST}} \Gamma^2 \cdot 2^{j+1} \cdot \alpha = O(2^j \cdot \alpha)$.
\end{lemma}
\begin{proof}
Let $j\in [\beta]$ be some $\TDE$ level and let $t$ be a time s.t. level $j$ is selected. Then the number of output modifications between the time $e_j$ (the enabling time of level $j$ estimators) and the time $t$ is $2^j-1$.
Denote the times during which the output was modified between the time $e_j$ the time $t$ by $\{t_l\}_{l\in [2^j-1]}$ where $t_{l=0} = e_j$. 
We first bound the difference of the current value of the function $\mathcal{F}$ to its value on the last output modification:
\begin{align}
    |\mathcal{F}(t)-\mathcal{F}(t_{2^{j}-1})| &\leq 
    |\mathcal{F}(t) - \mathcal{F}(t-1)| + |\mathcal{F}(t-1)-\Output(t_{2^{j}-1})| + |\Output(t_{2^{j}-1})-\mathcal{F}(t_{2^{j}-1})|\nonumber\\
    &\leq \MaxUpdateSize(\alpha)\cdot \mathcal{F}(t) + \alpha\cdot \mathcal{F}(t-1) + \alpha_{\Stitch}\cdot \mathcal{F}(t_{2^{j}-1}) \label{eq:LastStepUpperBound}
\end{align}
where the last inequality is due to \ref{lem:outputAccuracy} and \ref{cor:outputUpdateAccuracy}, and the assumption of bounded update size (condition \ref{con:streamMaxUpdateSize}).
The following holds:
\begin{align*}
    |\mathcal{F}(t) - \mathcal{F}(e_j)| \stackrel{1}{\leq}& \sum_{l\in [2^j-1]}|\mathcal{F}(t_{l+1})-\mathcal{F}(t_{l})| + |\mathcal{F}(t)-\mathcal{F}(t_{2^{j}-1})| \\
    \stackrel{2}{\leq}& \sum_{l\in [2^j-1]} \left(\StepSize(\alpha)\cdot \mathsf{Z}_{\ST} + (2\cdot \alpha_{\Stitch} + \MaxUpdateSize(\alpha))\cdot \max\{\mathcal{F}(t_{l}), \mathcal{F}(t_{l-1})\}\right) \\
     & + \MaxUpdateSize(\alpha)\cdot \mathcal{F}(t) + \alpha\cdot \mathcal{F}(t-1) + \alpha_{\Stitch}\cdot \mathcal{F}(t_{2^{j}-1})\\
     \stackrel{}{\leq} &
     (2^j-1)\StepSize(\alpha)\cdot \mathsf{Z}_{\ST} + 
     (2^j-1)(2\alpha_{\Stitch}+\MaxUpdateSize(\alpha)) \max_{l\in [2^j-1]}\{\mathcal{F}(t_{l})\}\\
     & + \MaxUpdateSize(\alpha)\cdot \mathcal{F}(t) + \alpha\cdot \mathcal{F}(t-1) + \alpha_{\Stitch}\cdot \mathcal{F}(t_{2^{j}-1})\\
    \stackrel{3}{\leq} &
     (2^j-1)(\StepSize(\alpha) + 2\alpha_{\Stitch} +\MaxUpdateSize(\alpha))\frac{1+\alpha_{\ST}}{1-\alpha_{\ST}} \Gamma^2 \cdot \mathcal{F}(e_j) \\
     &+ (\MaxUpdateSize(\alpha) + \alpha + \alpha_{\Stitch})\frac{1+\alpha_{\ST}}{1-\alpha_{\ST}} \Gamma^2 \cdot \mathcal{F}(e_j) \\
     \stackrel{4}{=}& (\alpha + 
     (2^j-1)\StepSize(\alpha) +
     (3(2^j-1) - 1)\MaxUpdateSize(\alpha) ) \frac{1+\alpha_{\ST}}{1-\alpha_{\ST}} \Gamma^2 \cdot \mathcal{F}(e_j)\\
     \stackrel{5}{\leq}& (\alpha + 2(2^j-1)\StepSize(\alpha)) \frac{1+\alpha_{\ST}}{1-\alpha_{\ST}} \Gamma^2 \cdot \mathcal{F}(e_j)\\
     \leq&  2^{j+1}\cdot \alpha \cdot \frac{1+\alpha_{\ST}}{1-\alpha_{\ST}} \Gamma^2 \cdot \mathcal{F}(e_j)\\
    \stackrel{6}{\leq}& \gamma_j \cdot   \mathcal{F}(e_j)
\end{align*}
where (1) is by decomposing ($\mathcal{F}(t) - \mathcal{F}(e_j)$) according to $\{t_l\}_{l\in [2^j-1]}$, 
(2) holds by plugging in Equation \ref{eq:LastStepUpperBound} and by applying Lemma \ref{lem:consequtiveOutputUpdatesProgress} on each of the differences in the term,
(3) is due to Lemma \ref{lem:maxPhaseProgress},
(4) is by setting $\alpha_{\Stitch} = \MaxUpdateSize(\alpha)$ ,
(5) is by setting $\StepSize(\alpha) \geq 4\cdot \MaxUpdateSize(\alpha)$,
(6) is by denoting 
$\gamma_j =  \frac{1+\alpha_{\ST}}{1-\alpha_{\ST}} \Gamma^2 \cdot 2^{j+1} \cdot \alpha$.
\end{proof}

\begin{lemma}[Accurate Estimations (Lemma 3.2 \cite{hassidim2020adversarially})]\label{lem:AccurateEstimations} The following holds for a {\em good execution} (see Definition \ref{def:goodRun}). Let $t\in [m]$ be a time step such that:
\begin{enumerate}
    \item Level $j$ was selected.
    \item Assumption \ref{ass:accurateEstimators} holds for every $t'< t$.
\end{enumerate}
Let $\mathsf{E}(\mathcal{S},\pi)$ be the estimator of level $j$ that was selected on time step $t$, and let $\pi$ be its (possibly dynamic) parameters.
Let $\mathsf{E}(\mathcal{S},\pi)$ have (an oblivious) guarantee that all of its estimates are accurate with accuracy parameter $\alpha_{\mathsf{E}}$
with probability at least $\frac{9}{10}$. 
Then for sufficiently small $\eps$, if algorithm \ref{alg:ROEF} is $(\eps,\delta^{\prime})$-DP w.r.t.\ the random bits of the estimators $\{\mathsf{E}^k\}_{k\in\mathsf{K}}$, then with probability at least $1-\frac{\delta^{\prime}}{\eps}$, for time $t$ we have:
\begin{enumerate}
    \item For $j \in \{\beta, \STwrapper\}$,  \qquad
    $|\{k\in [\mathsf{K}] : |z^{k}-\mathcal{F}(t)|< \alpha_{\mathsf{E}}\cdot\mathcal{F}(t) \}| \geq (8/10) \mathsf{K}$
    \item For $j < \beta$, \qquad \qquad{}
    $|\{k\in [\mathsf{K}] : |z^{k}-(\mathcal{F}(t) - \mathcal{F}(e))|< \alpha_{\mathsf{E}}\cdot\mathcal{F}(e) \}| \geq (8/10) \mathsf{K}$
\end{enumerate}
Where $z^k \leftarrow \mathsf{E}^{k}(\mathcal{S},\pi)$ for a set of size $\mathsf{K} \geq \frac{1}{\eps^2}\log\left( \frac{2\eps}{\delta^{\prime}} \right) $ of the oblivious estimator $\mathsf{E}(\mathcal{S},\pi)$ \end{lemma}

\begin{proof} Since the requirements of Lemma \ref{lem:boundedEstimationRanges} holds, we have that on time $t$, whenever the level $j$ that was selected is a type $\TDE$ estimator (i.e. $j < \beta$), that the {\em accuracy requirement} of these estimators holds. That is, 
$$
|\mathcal{F}(t) - \mathcal{F}(e_j)| \leq \gamma_j\cdot  \mathcal{F}(e_j)
$$
Now, for time $t$ let $\mathcal{S}_t=\left( \langle s_1, \Delta_1 \rangle , \dots,  \langle s_t, \Delta_t \rangle \right)$ be the prefix of the input stream $\mathcal{S}$ for that time, and let $\pi(t)$ be the parameters configured to $\mathsf{E}$ at that time. Let $z_t\leftarrow \mathsf{E}(r,\langle \mathcal{S}_t, \pi(t) \rangle)$ be the estimation returned by the oblivious streaming algorithm $\mathsf{E}$ after the $t$ stream update, when its executed with random string $r$ on the input stream $\mathcal{S}_t$ with parameters $\pi(t)$. Consider the following function (which is differently defined w.r.t the estimator type):
\begin{enumerate}
    \item for $j = \{\beta,\STwrapper\}$, define $f_{\langle \mathcal{S}_t, \pi(t) \rangle}(r) = \mathbbm{1}\left\{ z_t \in \left( 1 \pm \alpha_{\mathsf{E}} \right)\cdot \mathcal{F}(\mathcal{S}_t)  \right\}$
    \item for $j < \beta$, \qquad{} define $f_{\langle \mathcal{S}_t, \pi(t) \rangle}(r) = \mathbbm{1}\left\{ z_t\in \left(\mathcal{F}(\mathcal{S}_{t}) - \mathcal{F}(\mathcal{S}_{e(t)})\right) \pm \alpha_{\mathsf{E}} \cdot \mathcal{F}(\mathcal{S}_{e(t)})  \right\}$
\end{enumerate}
Since Lemma \ref{lem:EstimatorsRandomStringsPrivacy} holds, then by the generalization properties of differential privacy (see Theorem \ref{thm:PrivacyImpGen}), assuming that $\mathsf{K} \geq \frac{1}{\eps^2}\log\left( \frac{2\eps}{\delta^{\prime}} \right)$, with probability at least $1-\frac{\delta^{\prime}}{\eps}$, the following holds for time $t$:
$$
\left| 
\E_{r} \left[f_{\langle \mathcal{S}_t, \pi(t) \rangle}(r) \right] - 
\frac{1}{\mathsf{K}} \sum_{k\in [\mathsf{K}]}f_{\langle \mathcal{S}_t, \pi(t) \rangle}(r_k)
\right| \leq 10\eps
$$
We continue with the analysis assuming that this is the case. Now observe that $\E_{r} \left[f_{\langle \mathcal{S}_t, \pi(t) \rangle}(r) \right] \geq 9/10$ by the utility guarantees of $\mathsf{E}$ 
(because when the stream is fixed and its {\em accuracy requirement} is met its answers are accurate to within a multiplicative error of $(1\pm \alpha_{\mathsf{E}})$ with probability at least $9/10$). 
Thus for $\eps\leq \frac{1}{100}$, for at least of $8/10$ of the executions of $\mathsf{E}$ we have $f_{\langle \mathcal{S}_t, \pi(t) \rangle}(r_k)=1$ which means 
the estimations $z_t$ returned from these executions are accurate. That is, we have that at least $8\mathsf{K}/10$ of the estimations $\left\{ z^k_t \right\}_{k\in [\mathsf{K}]}$ satisfy the accuracy of the estimators of level $j$.
\end{proof}
%
%%%%%%%%%%%%%%%%%%%%%%%%%%%%%%%%%%%%%%%
% Accurate Estimation assumption holds
%%%%%%%%%%%%%%%%%%%%%%%%%%%%%%%%%%%%%%%
Lemmas \ref{lem:boundedEstimationRanges} and Lemma \ref{lem:AccurateEstimations} state that for some time $t\in [m]$ given that assumption \ref{ass:accurateEstimators} holds for all times $t^{\prime}< t$ then it also hold in time $t$ (w.p. $1-\eps/\delta^{\prime}$). 
As a corollary we get the following:
\begin{lemma}[Accuracy assumption holds]\label{lem:accurateEstimationsAllTime}
Fix a time step $t\in [m]$. 
Let $j\in [\beta+1]\cup\{\STwrapper\}$ be the level of estimators used in time $t$,
Recall that $\mathsf{K}_j$ denotes the number of the estimators in level $j$, and let $z^{1}_{j}, \dots, z^{\mathsf{K}_j}_{j}$ denote the estimations given by these estimators.
Then with probability at least $1-\delta/2$, the following holds for all time $t\in [m]$:
\begin{enumerate}
    \item For $j \in \{\beta$, \STwrapper\},  \qquad 
    $|\{k\in [\mathsf{K}_j] : |z^{k}_{j}-\mathcal{F}(t)|< \alpha_{\ST}\cdot\mathcal{F}(t) \}| \geq (8/10) \mathsf{K}_j$
    \item For $j < \beta$, \qquad \qquad{ }
    $|\{k\in [\mathsf{K}_j] : |z^{k}_{j}-(\mathcal{F}(t) - \mathcal{F}(e_j))|< \alpha_{\TDE}\cdot\mathcal{F}(e_j) \}| \geq (8/10) \mathsf{K}_j$
\end{enumerate}
provided that $\eps_j = O\left(\eps\frac{1}{\sqrt{P_j\cdot \log{1/\delta^{\prime}}}}\right)$ for $\delta^{\prime} = O(\eps\delta/m\beta)$.
\end{lemma}

\begin{proof} Fix some time $t\in [m]$ and let $j\in [\beta +1]\cup\{\STwrapper\}$ be the level of estimators used in that time. We set $\delta^{\prime}=\eps\cdot \delta/(2m(\beta+2))$, then by union bound over $m$ possible times we have that with probability $1-\delta/(2(\beta+2))$, $8/10\cdot \mathsf{K}_j$ of the estimators of level $j$ are accurate by Lemma \ref{lem:AccurateEstimations}. Union bound over all different $\beta +2$ estimators level, we get that with probability at least $1-\delta/2$, $8/10$ of estimations are accurate for all time $t\in [m]$ for all levels $j\in [\beta]\cup\{\ST\}\cup\{\STwrapper\}$.
\end{proof}

%%%%%%%%%%%%%%%%%%%%%%%%%%%%%%%%%%%%%%%
% Algorithm is correct on all times
%%%%%%%%%%%%%%%%%%%%%%%%%%%%%%%%%%%%%%%
By Corollary \ref{cor:NoCapping} and Lemma \ref{lem:accurateEstimationsAllTime} we have that algorithm \texttt{RobustDE} will not get to capping state and Assumption \ref{ass:accurateEstimators} holds. That is, the conditions for Lemma \ref{lem:outputAccuracy} are met and we have the following:
\begin{theorem}[Algorithm \ref{alg:ROEF} correctness]\label{thm:algCorrectness} 
Denote $\delta^{*} = \delta/\beta = O(\delta/\log(\alpha^{-1}))$. 
Provided that For all $j\in [\beta+1]$:
\begin{enumerate}
    \item $\gamma_j = \Omega(2^j\cdot \alpha)$
    \item $\eps_{j}=O\left(1/\sqrt{P_j\log(m/\delta^{*})}\right)$
    \item $P_j = \Omega\left(\frac{\lambda}{2^j}\right)$
    \item $\mathsf{K}_{j} = \Omega
\left( \sqrt{P_j\log\left(\frac{m}{\delta^{*}} \right)} \left[
\log\left( \frac{P_{j}}{\delta^{*}\alpha} \log(n)\right)
+
\log\left( \frac{m}{\delta^{*}}\right)
\right]
\right)$
\end{enumerate}
and $\eps_{\STwrapper} = \eps_{\beta}$, $P_{\STwrapper} = P_{\beta}$, $\mathsf{K}_{\STwrapper} = \mathsf{K}_{\beta}$, then for all time $t\in [m]$,
 with probability at least $1-\delta$ we have
$$
|\Output(t) - \mathcal{F}(t)| \leq \alpha \cdot \mathcal{F}(t) \text{.}
$$
\end{theorem}

\begin{proof} The proof follows by two parts: estimators of the framework remain accurate under adaptive inputs and that the framework computes an accurate output from their estimation.

\paragraph{Estimators are accurate w.h.p.} 
Lemma \ref{lem:accurateEstimationsAllTime} holds due to the privacy of the data bases $\{\mathcal{R}_j\}_{j\in [\beta+1]}\cup\{\mathcal{R}_{\STwrapper}\}$ and by making sure the estimations of level $[\beta]$ are done within the estimators accuracy range. The later holds by Lemma \ref{lem:boundedEstimationRanges} with configuring $\gamma_j = \Omega(2^j\cdot \alpha)$.
The privacy of $j\in [\beta+1]\cup\{\STwrapper\}$ databases $\mathcal{R}_j$ in Lemma \ref{lem:EstimatorsRandomStringsPrivacy} is due to calibrating the noise parameters $\eps_{j}=O\left(\eps/\sqrt{P_j\log(1/\delta^{\prime})}\right)$. By Corollary \ref{cor:NoCapping}, it is sufficient to set $P_j = \Omega\left(\frac{\lambda}{2^j}\right)$ (and $P_{\STwrapper} = P_{\beta}$) to have sufficient privacy budget for all databases $\mathcal{R}_j$. 
And so, by Lemma \ref{lem:accurateEstimationsAllTime}, setting $\delta^{\prime}=\eps\cdot \delta/(2m(\beta+2)) = O(\eps\cdot \delta/m\beta)$ yields that at least $80\%$ of estimators of each of the levels $j$ are accurate on all $t\in [m]$ (assumption \ref{ass:accurateEstimators}) w.p. at least $1-\delta/2$. 

\paragraph{Output accuracy.} It remains to show that the requirements of above lemmas are met. That is we have a {\em good run} (Definition \ref{def:goodRun}) w.h.p, and in addition the number of estimators $\mathsf{K}_j$ on each level is calibrated according to the constraints of above lemmas. We begin by calculating a sufficient number of estimators $\mathsf{K}_j$ for the required lemmas: First, 
Lemma \ref{lem:AccurateEstimations} require for all levels to have $\mathsf{K} \geq \frac{1}{\eps^2}\log\left( \frac{2\eps}{\delta^{\prime}} \right)$. 
Lemma \ref{lem:FrozenAccuracy} requires $\mathsf{K}_{j} = \Omega \left( \frac{1}{\eps_j} \log\left( \frac{P_{j}}{\delta^{M}\alpha} \log(n)\right) \right)$, 
Lemma \ref{lem:maxPhaseProgress} requires $\mathsf{K}_{\STwrapper} = \Omega\left( \frac{1}{\eps_{\STwrapper}}  \log\left( \frac{m}{\delta^{N}}\right)  \right)$,
Lemma \ref{lem:outputAccuracy} requires $\mathsf{K}_{j} = \Omega \left( \frac{1}{\eps_j}\log\left( \frac{m}{\delta^{N}}\right) \right)$,
Lemma \ref{lem:consequtiveOutputUpdatesProgress} requires $\mathsf{K}_{j} = \Omega \left( \frac{1}{\eps_j}
\log\left( \frac{m}{\delta^{N}}\right)
\right)$,
and we have overall requirement of:
\begin{equation}
\mathsf{K}_{j} = \Omega
\left( \frac{1}{\eps_j} \left[
\log\left( \frac{P_{j}}{\delta^{M}\alpha} \log(n)\right)
+
\log\left( \frac{m}{\delta^{N}}\right)
\right]
+
\frac{1}{\eps^2}\log\left( \frac{2\eps}{\delta^{\prime}} \right)
\right)\text{.} \label{eq:numberOfEstimatorRequired}
\end{equation}
Now, recall that $\eps=10^{-2}$ (constant). Setting $\delta^{N}=\delta/(4\cdot (\beta+2))$, $\delta^M = \delta/(4\cdot(\beta+1))$ and denote $\delta^{*} = \delta/\beta$ we get that $\delta^{\prime} = O(\delta^{*}/m)$, $\delta^{M}, \delta^{N} = O(\delta^{*})$. And so Equation \ref{eq:numberOfEstimatorRequired} simplified:
\begin{equation}
\mathsf{K}_{j} = \Omega
\left( \sqrt{P_j\log\left(\frac{m}{\delta^{*}} \right)} \left[
\log\left( \frac{P_{j}}{\delta^{*}\alpha} \log(n)\right)
+
\log\left( \frac{m}{\delta^{*}}\right)
\right]
\right)\text{.} 
\end{equation}
The setting $\delta^{N}=\delta/(4\cdot (\beta+2))$, $\delta^M = \delta/(4\cdot(\beta+1))$ also yields that we have a good run w.p. at least $\delta/2$. And so, all the requirement of Lemma \ref{lem:outputAccuracy} are met and we have with probability at least $1-\delta$ the output is accurate in all time $t\in [m]$.
\end{proof}

%%%%%%%%%%%%%%%%%%%%%%%%%%%%%%%%%%%%%%%%%%%%
% Main result - Framework Space Complexity
%%%%%%%%%%%%%%%%%%%%%%%%%%%%%%%%%%%%%%%%%%%%
\subsection{Calculating the space complexity}\label{sec:formalSpaceComplexity} 
Algorithm \texttt{RobustDE} space complexity is determined by its input parameters: accuracy parameter $\alpha>0$, 
the flip number bound $\lambda$ of the input stream $\SSS$ for functionality $\mathcal{F}$, failure probability $\delta \in (0,1]$, and space complexity of the given subroutines $\mathsf{E}_{\ST}$ and $\mathsf{E}_{\TDE}$ (denoted $S_{\ST}(\alpha_{\ST},\delta_{\ST},n,m)$ and $S_{\TDE}(\gamma,\alpha_{\TDE},\delta_{\TDE},n,m)$ correspondingly).
Observe that the pointers to the subroutines are on order of the total number of estimators (i.e. $\sum_{j \in [\beta+1]} \mathsf{K}_{j}$). Therefore, the dominating parameter of the space complexity is the number of the estimators, which multiplied by the space complexity of the estimators will by the dominating term in the total space complexity.
These set sizes $\mathsf{K}_{j}$ are determined by the number of times the corresponding estimators type that were used caused an output modification, which is upper bounded by $P_j$ and calculated in Lemma \ref{lem:outputUpdatesBounds}. In Theorem \ref{thm:algCorrectness} we have the sufficient numbers of set-sizes to calculate the space complexity of algorithm \texttt{RobustDE}:

%%%%%%%%%%%%%%%%%%%%%%%%%%%%%%%%%%%%%%%%%%%%
% Formal proofs - Framework Space Complexity
%%%%%%%%%%%%%%%%%%%%%%%%%%%%%%%%%%%%%%%%%%%%
\begin{theorem}[Framework for Adversarial Streaming - Space]\label{thm:FrameworkSpace}
Provided that there exist: 
\begin{enumerate} 
	\item An oblivious streaming algorithm $\mathsf{E}_{\ST}$ for functionality $\mathcal{F}$, that guarantees that with probability at least $9/10$ all of it's estimates are accurate to within a multiplicative error of $(1\pm \alpha_{\ST})$ with space complexity of $S_{\ST}(\alpha_{\ST}, n,m)$
	\item For every $\gamma,p$ there is a $(\gamma,\alpha_{\TDE},p,\frac{1}{10})$-$\TDE$ for $\mathcal{F}$ using space $\gamma \cdot S_{\TDE}(\alpha_{\TDE},p,n,m)$.
\end{enumerate}  
Then there exist an adversarially robust streaming algorithm for functionality $\mathcal{F}$ that for any stream with a bounded flip number $\lambda_{\alpha/8,m}< \lambda$, s.t.\ with probability at least $1-\delta$ its output is accurate to within a multiplicative error of $(1\pm \alpha)$ for all times $t\in [m]$, and has a space complexity of
$$
O\left(
\sqrt{\alpha \lambda}
\cdot 
\polylog(\lambda, \alpha^{-1}, \delta^{-1}, n, m)
\right)
\cdot 
\left[ 
S_{\ST}(O(\alpha),n,m) + S_{\TDE}(O(\alpha/\log(\alpha^{-1})),\lambda,n,m) 
\right] \text{.}
$$
\end{theorem}

\begin{proof} 
Setting the estimators set sizes $\mathsf{K}_j$ for $j\in [\beta+1]$  according to Theorem \ref{thm:algCorrectness}, 
ensures that with probability at least $1-\delta$, algorithm \texttt{RobustDE} produce an accurate output at all times $t\in [m]$ with the presence of an adaptive adversary controlling the stream.
And so, we calculate the space needed according this setting.
We separate the calculation of levels according to the type of estimators used in each level.

\paragraph{$\TDE$ level $j$ space:} Consider an oblivious toggle difference estimator for the function $\mathcal{F}$, $\mathsf{E}_{\TDE}$, with space complexity $\Space(\mathsf{E}_{\TDE}) = \gamma\cdot S_{\TDE}(\alpha, p, n, m)$.
Recall (see Lemma \ref{lem:boundedEstimationRanges}, Lemma \ref{lem:outputAccuracy} and Corollary \ref{cor:NoCapping}) that for level $j$ estimators we have $\gamma_j = O(\alpha \cdot 2^{j})$, $\alpha_{\TDE} = O(\alpha/\beta) = O(\alpha/\log(\alpha^{-1}))$, $P_j = O(\lambda/2^j)$. Plugging in the parameters of granularity level $j$ yields an oblivious TDE with space complexity of:
\begin{align*}
\Space(\mathsf{E}_{\TDEj}) &= \gamma_j \cdot S_{\TDE}(\alpha_{\TDE}, P_j, n, m) \\
&= O(\alpha \cdot 2^{j}) \cdot S_{\TDE}\left( P_j \right)
\end{align*}
Accounting for the sufficient amount of estimators $\mathsf{K}_{\TDE,j}$:
\begin{align*}
\Space(\TDEj) &= \mathsf{K}_{\TDE,j}\cdot \Space(\mathsf{E}_{\TDE,j})\\
&= \sqrt{P_j\cdot\log\left(\frac{m}{\delta^{*}}\right)}\cdot \left[ \log\left( \frac{m}{\delta^{*}}\right) + \log\left( \frac{P_{j}}{\alpha\delta^{*}} \log(n)\right) \right] \cdot O(\alpha \cdot 2^{j}) \cdot S_{\TDE}\left( P_j \right)\\
&= \sqrt{\frac{\lambda}{2^j}\cdot\log\left(\frac{m}{\delta^{*}}\right)}\cdot \left[ \log\left( \frac{m}{\delta^{*}}\right) + \log\left( \frac{P_{j}}{\alpha\delta^{*}} \log(n)\right) \right] \cdot 
O(\alpha \cdot 2^{j}) \cdot S_{\TDE}\left( P_j \right)\\
&= O\left( 
\alpha \cdot \sqrt{\lambda \cdot 2^{j}}
\cdot \left[
\log\left( \frac{m}{\delta^{*}}\right) + \log\left( \frac{P_{j}}{\alpha\delta^{*}} \log(n)\right) 
\right]
\cdot \sqrt{\log\left(\frac{m}{\delta^{*}}
\right)}
\cdot S_{\TDE}\left( P_j \right)
\right)\\
&\stackrel{(\ast)}{\leq} O\left( 
\alpha \cdot \sqrt{\lambda \cdot 2^{j}}
\cdot \left[
\log\left( \frac{m}{\delta^{*}}\right) + \log\left( \frac{\lambda}{\alpha\delta^{*}} \log(n)\right) 
\right]
\cdot \sqrt{\log\left(\frac{m}{\delta^{*}}
\right)}
\cdot S_{\TDE}\left( \lambda \right)
\right)\\
&= O\left( 
\alpha \cdot \sqrt{\lambda \cdot 2^{j}}
\cdot
\polylog\text{}_{\TDE}(\lambda, \alpha^{-1}, \delta^{-1}, n, m)
\cdot S_{\TDE}\left( \lambda \right)
\right)
\end{align*}
Where $(\ast)$ holds since $S_{\TDE}$ is monotonic increasing in its $p$ parameter and for all $j$ $P_j = O(\lambda/2^j) \leq O(\lambda)$ and on last equality we denoted $\polylog\text{}_{\TDE} = \left[
\log\left( \frac{m}{\delta^{\prime}}\right) + \log\left( \frac{\lambda}{\alpha\delta^{\prime}} \log(n)\right) 
\right]
\cdot \sqrt{\log\left(\frac{m}{\delta^{*}}
\right)}$.

\paragraph{TDEs total space:}
We sum the calculated $\Space(\TDEj)$ over $j\in [\beta]$ for $\beta=\lceil \log(\alpha^{-1})\rceil$ to get the total TDE estimators space:
\begin{align*}
\Space(\AllTDE) &= \sum_{j\in [\beta]}\Space(\TDEj)\\
&\leq\sum_{j\in [\beta]}O\left( 
\alpha \cdot \sqrt{\lambda \cdot 2^{j}}
\cdot
\polylog\text{}_{\TDE}
\cdot S_{\TDE}
\right)\\
&= O\left( 
\alpha \cdot \sqrt{\lambda} \cdot
\polylog\text{}_{\TDE} \cdot S_{\TDE}
\right) \cdot \sum_{j\in [\beta]}\sqrt{2^{j}}\\
&\stackrel{(\ast)}= O\left( 
\alpha \cdot \sqrt{\lambda} \cdot
\polylog\text{}_{\TDE} \cdot S_{\TDE}
\right) \cdot O\left( \alpha^{-0.5} \right)\\
&= O\left( 
\sqrt{\alpha \cdot \lambda}
\cdot
\polylog\text{}_{\TDE} 
\right)
\cdot S_{\TDE}
\end{align*}
where equality $(\ast)$ is true since for $\beta=\lceil \log(\alpha^{-1})\rceil$ the following holds: 
$$\sum_{j\in [\beta]}\sqrt{2^j}=\sum_{j\in [\beta]}(\sqrt{2})^j=\frac{\sqrt{2}}{\sqrt{2}-1}\cdot\left( (\sqrt{2})^{\log(\alpha^{-1})} -1 \right)=\frac{\sqrt{2}}{\sqrt{2}-1}\cdot\left( \alpha^{-0.5} -1 \right) = O(\alpha^{-0.5})$$

\paragraph{ST space:} Here we argue about the space complexity of both $\ST$, $\STwrapper$ estimators. Since all their parameters are identical, we only calculate the space for $\ST$ estimators.
Consider an oblivious strong tracker for the function $\mathcal{F}$, $\mathsf{E}_{\ST}$, with space complexity $\Space(\mathsf{E}_{\ST}) = S_{\ST}(\alpha, n, m)$. 
Recall (see Lemma \ref{lem:outputAccuracy} and Corollary \ref{cor:NoCapping}) that for level $j\in \{ \ST, \STwrapper \}$ estimators we have $\alpha_{ST} = O(\alpha)$, $P_j = O(\alpha \cdot \lambda)$.
Then plugging in the parameters for the $\ST$, $\STwrapper$ type estimator yields an oblivious estimator with space complexity of:
\begin{align*}
\Space(\mathsf{E}_{\ST}) &= S_{\ST}(\alpha_{\ST},n,m)
\end{align*}
Accounting for the sufficient amount of estimators $\mathsf{K}_{\ST}$:
\begin{align*}
\Space(\AllST) &= \mathsf{K}_{\ST}\cdot \Space(\mathsf{E}_{\ST})\\
&= \left( \sqrt{P_{\ST}\cdot\log\left(\frac{m}{\delta^{*}}\right)}\cdot \left[ \log\left( \frac{m}{\delta^{*}}\right) + \log\left( \frac{P_{\ST}}{\alpha\delta^{*}} \log(n)\right) \right] \right)
\cdot S_{\ST}\\
&=O\left( \sqrt{\alpha \cdot \lambda\cdot\log\left(\frac{m}{\delta^{*}}\right)}\cdot \left[ \log\left( \frac{m}{\delta^{*}}\right) + \log\left( \frac{\alpha\lambda}{\alpha\delta^{*}} \log(n)\right) \right] \right)
\cdot S_{\ST}\\
&=O\left(\sqrt{\alpha\cdot \lambda}\cdot \left[ 
\log\left( \frac{m}{\delta^{*}}\right) + \log\left( \frac{\lambda}{\delta^{*}} \log(n)\right)
\right] 
\sqrt{\log\left(\frac{m}{\delta^{*}}\right)}
\right)\cdot S_{\ST}\\
&= O\left(\sqrt{\alpha\cdot \lambda}
\cdot \polylog\text{}_{\ST}(\lambda, \alpha^{-1}, \delta^{-1}, n, m)
\right)
\cdot S_{\ST}
\end{align*}
on last equality we denoted $\polylog\text{}_{\ST} = \left[ 
\log\left( \frac{m}{\delta^{*}}\right) + \log\left( \frac{\lambda}{\delta^{*}} \log(n)\right)
\right] 
\sqrt{\log\left(\frac{m}{\delta^{*}}\right)}$.\\

\paragraph{Algorithm space complexity:} We sum the contribution of all estimators, Strong trackers and TDEs, to get:
\begin{align*}
\Space(\texttt{RobustDE}) =& \Space(\AllST) + \Space(\AllTDE)\\
=&O\left(\sqrt{\alpha\cdot \lambda}
\cdot \polylog\text{}_{\ST}(\lambda, \alpha^{-1}, \delta^{-1}, n, m)
\right)
\cdot S_{\ST} \text{ }+\\
& O\left( 
\sqrt{\alpha \cdot \lambda}
\cdot
\polylog\text{}_{\TDE}(\lambda, \alpha^{-1}, \delta^{-1}, n, m)
\right)
\cdot S_{\TDE}\\
=&O\left(
\sqrt{\alpha\cdot \lambda} \cdot 
\polylog\text{}_{\ALG}(\lambda, \alpha^{-1}, \delta^{-1}, n, m)
\right)\cdot 
\left[
S_{\ST} + 
S_{\TDE}
\right]
\end{align*}
Where 
\begin{enumerate}
	\item $S_{\ST} = S_{\ST}(\alpha_{\ST},n,m) = S_{\ST}(O(\alpha),n,m)$.
	\item $S_{\TDE} = S_{\TDE}(\alpha_{\TDE},\lambda,n,m) = S_{\TDE}(O(\alpha/\log(\alpha^{-1})),\lambda,n,m)$.
	\item $\polylog_{\ALG} =  \left[ 
\log\left( \frac{m}{\delta^{*}}\right) + \log\left( \frac{\lambda}{\alpha\delta^{*}} \log(n)\right)
\right] 
\sqrt{\log\left(\frac{m}{\delta^{*}}\right)}
$, for $\delta^{*} = \delta/\log(\alpha^{-1})$.
\end{enumerate}
\end{proof}

%%%%%%%%%%%%%%%%%%%%%%%%%%%%%%%%%%%%%%%%%%%%%%%%%
% Corollary for functions with space \alpha^{-2}
%%%%%%%%%%%%%%%%%%%%%%%%%%%%%%%%%%%%%%%%%%%%%%%%%
\begin{corollary}
Provided that there exist: 
\begin{enumerate} 
	\item An oblivious streaming algorithm $\mathsf{E}_{\ST}$ for functionality $\mathcal{F}$, that guarantees that with probability at least $9/10$ all of it's estimates are accurate to within a multiplicative error of $(1\pm \alpha_{\ST})$ with space complexity of $ O\left( \frac{1}{\alpha_{\ST}^{2}}\cdot f_{\ST}\right)$ for $f_{\ST}=\polylog(\alpha_{\ST},n,m)$ 
	\item For every $\gamma,p$ there is a $(\gamma,\alpha_{\TDE},\frac{1}{10},p)$-$\TDE$ for $\mathcal{F}$ using space $\gamma \cdot O\left( \frac{1}{\alpha_{\TDE}^{2}}\cdot f_{\TDE} \right)$ for $f_{\TDE}=\polylog(\alpha_{\TDE},p,n,m)$.
\end{enumerate}  
Then there exist an adversarially robust streaming algorithm for functionality $\mathcal{F}$ that for any stream with a bounded flip number $\lambda_{\frac{1}{8}\alpha,m}< \lambda$, s.t.\ with probability at least $1-\delta$ its output is accurate to within a multiplicative error of $(1\pm \alpha)$ for all times $t\in [m]$, and has a space complexity of
$$
O\left(
\frac{\sqrt{\lambda}}{\alpha^{1.5}} \cdot 
\polylog(\lambda, \alpha^{-1}, \delta^{-1}, n, m)
\right)
$$
\end{corollary}

%%%%%%%%%%%%%%%%%%%%%%%%%%%%%%%%%%%%%%%%%%%%%%
% Section: Application - F_2
% Construction and formal analysis of Guardian
%%%%%%%%%%%%%%%%%%%%%%%%%%%%%%%%%%%%%%%%%%%%%%
\section{Formal Details for Applications (Section~\ref{sec:Applications})}\label{sec:MissingAppli}
In this section we give the formal details for the resulting space bounds for $F_2$. As these bounds are a function of a stream characterization, we begin with that.

\paragraph{Characterising the input streams for $F_2$.} 
The $F_2$ $\DE$ construction presented in \cite{woodruff2020tight} has an additional requirement for turnstile streams. We now present this requirement:
\begin{lemma}[Difference estimator for $F_2$ (Lemma 3.2, \cite{woodruff2020tight})] 
There exists a $(\gamma, \alpha, \delta)$-difference estimator for $F_2$ that uses space of $O(\gamma\eps^{-1}\log n (\log \alpha^{-1} + \log \delta^{-1}))$ for streams $\SSS$ that for any time $t>e$, where $e\in [m]$ is the enabling time, admit:
\begin{equation}
  F_2(\mathcal{S}_{t}^{e})\leq \gamma\cdot F_2(\mathcal{S}_{e}) \label{req:F2DEratio}
\end{equation}
\end{lemma}
For $F_2$ estimation of a turnstile stream, it may be the case that requirement \ref{req:F2DEratio} does not hold while the DE accuracy guarantee does (See \ref{req:TDEdiffRange} in Definition \ref{def:TDE}).
\footnote{To see that, consider a stream with prefix frequency vector $u$ and suffix frequency vector $w$ s.t. the norm of the frequency vectors between these two is roughly the norm of $u$ (and so is the norm of $u+w$) while the support of $u$ and $u+w$ is disjoint.}
The problem is that in such a scenario the $\DE$ estimators that are used by the framework are not accurate, while the framework may try to use their estimations.
In order to capture such a scenario in a stream, we define the following:
\begin{definition}[Suffix violation of $\mathcal{F}$ ]\label{def:F2Violation} 
Let $\gamma \in (0,1)$. For $\mathcal{F}$, for some time $e \in [t]$ where the stream  $\mathcal{S}_{t}$ of length $t$ is partitioned, denote $\mathcal{S}_{e}$ as the prefix of that partition and by $\mathcal{S}_{t}^{e}$ its suffix. Then time $e$ is a $\gamma$-{\em suffix violation} for $\mathcal{F}$ if the following holds:
\begin{enumerate}
    \item $|\mathcal{F}(\mathcal{S}_{t}) - \mathcal{F}(\mathcal{S}_{e}) | \leq \gamma \cdot \mathcal{F}(\mathcal{S}_{e})$, and  
    \item $\mathcal{F}(\mathcal{S}_{t}^{e}) > \gamma\cdot \mathcal{F}(\mathcal{S}_{e})$
\end{enumerate}
\end{definition}
Note, that on times $t\in [m]$ that are {\em $F_2$ suffix violations} the $\DE$ construction from \cite{woodruff2020tight} has no accuracy guarantee, and so our framework cannot use it for its estimation.
We wish to characterize the input stream w.r.t the number of such violations. For that, we present the notion of a {\em twist number} of a stream (also defined in Section \ref{sec:introduction}, Definition \ref{def:twistIntro}):
\begin{definition}[Twist number]\label{def:twist} 
The {\em $(\alpha,m)$-twist number} of a stream $\SSS$ w.r.t.\ a functionality $\FFF$, denoted as $\TN_{\alpha,m}(\SSS)$, is the maximal $\mu\in[m]$ such that $\SSS$ can be partitioned into $2\mu$ disjoint segments $\SSS = \PPP_0 \circ \VVV_0 \circ\dots \circ \PPP_{\TN-1} \circ \VVV_{\TN-1}$ (where $\{\PPP_i\}_{i\in[\TN]}$ may be empty) s.t.\ for every $i\in[\mu]$:
\begin{enumerate}
    \item $\FFF(\VVV_i) > \alpha \cdot \FFF(\PPP_0 \circ \VVV_0 \circ \dots\circ \VVV_{i-1} \circ \PPP_i)$
    \item $|\FFF(\PPP_0 \circ\VVV_0 \dots \circ \PPP_i \circ \VVV_i) - \FFF(\PPP_0\circ\VVV_0 \circ \dots \circ \PPP_i)| \leq \alpha\cdot \FFF(\PPP_0 \circ \VVV_0 \circ \dots \circ \PPP_i)$
\end{enumerate}
\end{definition}

\paragraph{An extension for the turnstile model: Algorithm description.} At a high level, the extension is wrapping our framework.
It monitors the output of the framework and checks whether it is accurate, in which case it forwards it as the output. If the framework is not accurate then it must be due to a previous input that did not admit some of the frameworks $\DE$ input requirement. In such case, the monitor sends a {\em phase reset command} to the framework, and outputs the same as the framework after its reset.
This accuracy assertion is done by running additional estimators (strong trackers) that are used as a validation to the framework output. These monitor estimators are correct on all turnstile input.
The extension is presented in algorithm \ref{alg:ViolationMonitor}. We also describe the exact modification needed in algorithm \ref{alg:ROEF} in order to receive external phase reset commands in \ref{alg:PhaseResetCommand}.

%%%%%%%%%%%%%%%%%%%%%%%%%%%%%%%%
% Algorithm: violation monitor
%%%%%%%%%%%%%%%%%%%%%%%%%%%%%%%%
\begin{algorithm*}[ht]
\caption{\texttt{Guardian($\SSS, \alpha,\delta,\lambda,\TN,\mathsf{E}_{\ST},\mathsf{E}_{\TDE}$)}}
\makeatletter\def\@currentlabel{\texttt{Guardian}}\makeatother
\label{alg:ViolationMonitor}

{\bf Input:} A stream $\mathcal{S} = \{\langle s_t, \Delta_t\rangle \}_{t\in [m]}$ accuracy parameter $\alpha$, failure probability $\delta$, a bound on the flip number $\lambda$ and a bound on the number of input violations $\TN$.

\hdashrule[0.5ex][x]{\linewidth}{0.5pt}{1.5mm}

{\bf Initialization:}
\begin{enumerate}[leftmargin=25pt,rightmargin=0pt,itemsep=0pt,topsep=0pt]
	\item Set $P_{\STmon} = O(\mu)$, $\eps_{\STmon} = O(1/\sqrt{P_{\STmon}\log \delta})$,  $\mathsf{K}_{\STmon} = \Omega \left( 
\frac{1}{\eps_{\STmon}}\log\left( \frac{m}{\delta}\right)
\right)$.

	\item Start estimators $\bar{\mathsf{E}}_{\STmon}$ with $\alpha_{\STmon}=O(\alpha)$, Set $\mathsf{T}_{\STmon} = \mathsf{K}_{\STmon}/2 + \Lap(2\cdot \eps_{\STmon}^{-1})$.
	
	\item Initialize algorithm \ref{alg:ROEF} with $(\alpha/2,\delta/2,\hat{\lambda},\mathsf{E}_{\ST},\mathsf{E}_{\TDE})$ where $\hat{\lambda} = O(\lambda + \mu \cdot \alpha^{-1})$.
\end{enumerate}

\hdashrule[0.5ex][x]{\linewidth}{0.5pt}{1.5mm}

{\bf For ${\boldsymbol{t\in[m]}}$:}
\begin{enumerate}
    \item Get the update $\langle s_t, \Delta_t \rangle$ from $\mathcal{S}$, feed into all estimators $\bar{\mathsf{E}}_{\STmon}$ and receive the estimations $z_{\STmon}^{k}$.
    \item $\Next \leftarrow \texttt{RobustDE}(\langle s_t, \Delta_t, 0 \rangle)$

    \item If $\left| \left\{ k\in [\mathsf{K}_{\STmon}] : |z_{\STmon}^{k} - \Next| \geq (3/4) \alpha  \cdot \Next \right\} \right| + \Lap(4\cdot\eps^{-1}_{\STmon})> \mathsf{T}_{\STmon}$ then \label{algstep:monViolationCond}
    \begin{enumerate}
        \item Redraw $\mathsf{T}_{\STmon}$
        \item $\Next \leftarrow \texttt{RobustDE}(\langle 0, 0, 1 \rangle)$ \qquad \qquad \qquad  \qquad{} \quad{}{} {\bf \% Phase reset command}
    \end{enumerate}
    \item $\Output \leftarrow \Next$
\end{enumerate}

\end{algorithm*}
\paragraph{An extension for turnstile model: Analysis structure.} Algorithm \texttt{Guardian} analysis is composed of two components. 
In the first (Section \ref{sec:boundResetCommands}) we show that on a $\mu$ bounded twist number stream there can be at most $\mu$ phase reset commands that are sent to algorithm \texttt{RobustDE} (Lemma \ref{lem:maxMonTrigger}). In addition, in that first component we also prove that the output of the algorithm \texttt{Guardian} is always accurate (Lemma \ref{lem:monOutputAccuracy}). 
The second component (Section \ref{sec:extensionSpace}) consist of calculating the resulting space bounds of the extended framework due to receiving $\mu$ phase reset commands and the additional space of the monitor (Theorem \ref{thm:ExtensionSpace}).
By instantiating known constructions of a strong tracker and a difference estimator for $F_2$ in Theorem \ref{thm:ExtensionSpace} we establish the resulting space bounds for $F_2$ in the turnstile model in Theorem \ref{thm:F2Space}.

\subsection{Bounding the number of phase reset commands.}\label{sec:boundResetCommands}
Next we show that for a $\TN$ bounded $(\gamma, m)$-twist number streams, algorithm \texttt{Guardian} captures at most $\TN$ $\gamma$-suffix violations (and so it issues that many phase reset commands to \texttt{RobustDE}).
That is established on Lemma \ref{lem:maxMonTrigger}.
Since algorithm \texttt{Guardian} uses {\em oblivious} estimators, 
we also prove that its output validation is correct in the {\em adaptive} input setting.
That is done by using a technique from \cite{hassidim2020adversarially}. That is, first we prove that algorithm \texttt{Guardian} is DP (Lemma \ref{lem:privacyMonitor}). Then, we use tools from DP to argue that the validation estimators are accurate (Lemma \ref{lem:AccurateMonEstimations}), and finally we show that this accuracy is leveraged for the correct validation (Lemma \ref{lem:monOutputAccuracy}). 

As we mentioned, we achieve robustness for estimators $\bar{\mathsf{E}}_{\STmon}$ via DP. The following lemma states that algorithm \texttt{Guardian} preserves privacy w.r.t the random strings of the estimators $\bar{\mathsf{E}}_{\STmon}$.

\begin{lemma}\label{lem:privacyMonitor}
Let $\mathcal{R}_{\STmon}$ be the random bit-strings dataset of  $\bar{\mathsf{E}}_{\STmon}$. Then algorithm \ref{alg:ViolationMonitor} satisfies $(\eps,\delta^{\prime})$-DP w.r.t.\ a dataset $\mathcal{R}_{\STmon}$ by configuring $\eps_{\STmon}=O\left(\eps/\sqrt{P_{\STmon}\log(1/\delta^{\prime})}\right)$.
\end{lemma}
\begin{proof}[Proof sketch] We focus on the time sequences that begin after a time the condition in \ref{algstep:monViolationCond} is True, and ends in the consecutive time that the condition in \ref{algstep:monViolationCond} is True.
Denote by $P_{\STmon}$ at the number of such time-sequences.
Throughout every such time sequence, we access the dataset $\mathcal{R}_{\STmon}$ via the sparse vector technique (See Algorithm of \ref{thm:AboveThreshold}).
We calibrate the privacy parameters of this algorithm to be $\eps_{\STmon}=O\left(\eps/\sqrt{P_{\STmon}\log(1/\delta^{\prime})}\right)$ such that, by using composition theorems across all of the $P_{\STmon}$ sequences, our algorithm satisfies $(\eps,\delta^{\prime})$-differential privacy w.r.t.\ $\mathcal{R}_{\STmon}$.
\end{proof}

In the following lemmas we assume that all the noises (up to 2m draws of $\Lap(O(\eps_{\STmon}^{-1}))$ noise) are smaller in absolute value from $\frac{4}{\eps_{\STmon}}\log\left( \frac{m}{\delta^{\STmon}}  \right)$ which is the case with probability at least $1-\delta^{\STmon}$.
First lemma is an adaptation of technique from Lemma 3.2 \cite{hassidim2020adversarially} that uses differential privacy to assert that the estimations of $\bar{\mathsf{E}}_{\STmon}$ are accurate.
\begin{lemma}[Accurate Estimations (Lemma 3.2 \cite{hassidim2020adversarially})]\label{lem:AccurateMonEstimations}
Let $\mathsf{E}(\mathcal{S})$ have (an oblivious) guarantee that all of its estimates are accurate with accuracy parameter $\alpha_{\mathsf{E}}$
with probability at least $\frac{9}{10}$. 
Then for sufficiently small $\eps$, if algorithm \ref{alg:ViolationMonitor} is $(\eps,\delta^{\prime})$-DP w.r.t.\ the random bits of the estimators $\{\mathsf{E}^k\}_{k\in\mathsf{K}}$, then with probability at least $1-\frac{\delta^{\prime}}{\eps}$, for time $t$ we have:
$$|\{k\in [\mathsf{K}] : |z^{k}-\mathcal{F}(t)|< \alpha_{\mathsf{E}}\cdot\mathcal{F}(t) \}| \geq (8/10) \mathsf{K}$$
Where $z^k \leftarrow \mathsf{E}^{k}(\mathcal{S})$ for a set of size $\mathsf{K} \geq \frac{1}{\eps^2}\log\left( \frac{2\eps}{\delta^{\prime}} \right) $ of the oblivious estimator $\mathsf{E}(\mathcal{S})$
\end{lemma}

\begin{proof}
[proof (A simplified version of \ref{lem:AccurateEstimations})]
For time $t \in [m]$ let $\mathcal{S}_t=\left( \langle s_1, \Delta_1 \rangle , \dots,  \langle s_t, \Delta_t \rangle \right)$ be the prefix of the input stream $\mathcal{S}$ for that time. 
Let $z_t\leftarrow \mathsf{E}(r,\mathcal{S}_t)$ be the estimation returned by the oblivious streaming algorithm $\mathsf{E}$ after the $t$ stream update, when its executed with random string $r$ on the input stream $\mathcal{S}_t$. Consider the following function:
$f_{\langle \mathcal{S}_t, \pi(t) \rangle}(r) = \mathbbm{1}\left\{ z_t \in \left( 1 \pm \alpha_{\mathsf{E}} \right)\cdot \mathcal{F}(\mathcal{S}_t)  \right\}$. 
Since algorithm \texttt{Guardian} is $(\eps, \delta^{\prime})$-DP, then by the generalization properties of differential privacy (see Theorem \ref{thm:PrivacyImpGen}), assuming that $\mathsf{K} \geq \frac{1}{\eps^2}\log\left( \frac{2\eps}{\delta^{\prime}} \right)$, with probability at least $1-\frac{\delta^{\prime}}{\eps}$, the following holds for time $t$:
$$
\left| 
\E_{r} \left[f_{\langle \mathcal{S}_t, \pi(t) \rangle}(r) \right] - 
\frac{1}{\mathsf{K}} \sum_{k\in [\mathsf{K}]}f_{\langle \mathcal{S}_t, \pi(t) \rangle}(r_k)
\right| \leq 10\eps
$$
We continue with the analysis assuming that this is the case. Now observe that $\E_{r} \left[f_{\langle \mathcal{S}_t, \pi(t) \rangle}(r) \right] \geq 9/10$ by the utility guarantees of $\mathsf{E}$
(because when the stream is fixed its answers are accurate to within a multiplicative error of $(1\pm \alpha_{\mathsf{E}})$ with probability at least $9/10$). 
Thus for $\eps\leq \frac{1}{100}$, for at least of $8/10$ of the executions of $\mathsf{E}$ we have $f_{\langle \mathcal{S}_t, \pi(t) \rangle}(r_k)=1$ which means 
the estimations $z_t$ returned from these executions are accurate. That is, we have that at least $8\mathsf{K}/10$ of the estimations $\left\{ z^k_t \right\}_{k\in [\mathsf{K}]}$ satisfy the accuracy of the estimators.

\end{proof}

\begin{lemma}[Maximal number of monitor triggers]\label{lem:maxMonTrigger}
For an input stream with a $(\gamma_0, m)$-twist number $\TN$, algorithm \ref{alg:ViolationMonitor} is sending at most $\TN$ reset commands for algorithm \ref{alg:ROEF}.
\end{lemma}

\begin{proof} We show, that for an execution with $\mu$ phase reset commands, the input stream $\SSS$ has a $(\gamma_0,m)$-twist number of at least $\mu$. That implies the statement.

Let $r_0< r_1< \dots < r_{\mu-1}$ be the times in which algorithm \texttt{RobustDE} has issued a phase reset command. 
We focus on the time segment $(r_{i-1}, r_{i}]$ for some $i\in [\mu]$ (and in the time segment $[0,r_0]$ in the case of $i=0$).
That is in time $r_i$ we have:
$$\left| \left\{ k\in [\mathsf{K}_{\STmon}] : |z_{\STmon}^{k} - \Next| \geq (3/4) \alpha\cdot \Next  \right\} \right| >
\frac{\mathsf{K}_{\STmon}}{2} - 
\frac{4}{\eps_{\STmon}}\log\left( \frac{2m}{\delta^{N}}  \right)  >
\frac{4\cdot \mathsf{K}_{\STmon}}{10}
$$
Where the first inequality holds in the event of the bounded noises and the second inequality holds by asserting that:
$
\mathsf{K}_{\STmon} = \Omega \left( 
\frac{1}{\eps_{\STmon}}\log\left( \frac{m}{\delta^{N}}\right)
\right)
= \Omega \left( \frac{1}{\eps}
\sqrt{P_{\STmon} \cdot \log \left(\frac{1}{\delta^{\prime}} \right) }\log\left( \frac{m}{\delta^{N}}\right)
\right)
$.
So, for at least $40\%$ of the estimations $z_{\STmon}^{k}$ of the estimators $\bar{\mathsf{E}}_{\STmon}$ it holds that $|z_{\STmon}^{k} - \Next| \geq (3/4) \alpha \cdot \Next$, and in the same time, by Lemma \ref{lem:AccurateMonEstimations} we have that at least $80\%$ of the estimators are accurate. That is for at least one estimation $z_{\STmon}^{k}$ both of above statements hold and we have (for any $\alpha \in (0,14/15)$):
$$
|\Next - \mathcal{F}(t)| \geq  (1/2)\alpha\cdot \mathcal{F}(t) \text{.}
$$
That is, algorithm \texttt{RobustDE} accuracy guarantee does not hold in time $r_i$, as it is not $(1/2)\alpha$ accurate.
In addition, in time $r_{i-1}$ there was also a phase reset (or in case that $i=0$, $r_{-1}=0$). And so, if there has been any suffix violations it has not effect after the phase reset of time $r_{i-1}$.
Therefore it follows that (w.h.p) there must be some time segment $[e,t]$ s.t.\ $r_{i-1}\leq e < t\leq r_{i}$ and in addition some level of estimators s.t.\ these estimators were not accurate, causing algorithm \texttt{RobustDE} accuracy guarantee to break.
Denote $\VVV_i$ as the input stream $\SSS$ in times $[e,t]$, then previous conclusion is that $\VVV_i$ is a $\gamma^{\prime}$-suffix violation for some $\gamma^{\prime}\geq \gamma_0$. 
That is, for each issued phase reset command, we have in $\SSS$ at least one $\gamma_0$-suffix violation which imply that in such a scenario the input stream must have a $(\gamma_0,m)$-twist number of at least $\mu$.
\end{proof}

\paragraph{Extension output is accurate.} We now show that the output of algorithm \texttt{Guardian} is accurate in all time $t\in [m]$.

\begin{lemma}\label{lem:monOutputAccuracy} If $P_\STmon > \mu$ then with probability at least $1-\delta$ algorithm \ref{alg:ViolationMonitor} output admit for all time $t\in [m]$: 
$$|\Output - \mathcal{F}(t)| \leq \alpha \cdot \mathcal{F}(t)$$
\end{lemma}

\begin{proof} We relate to two cases w.r.t condition \ref{algstep:monViolationCond}. 
If the condition is True, then the output is given after a phase reset command. In that case it was computed by the ST level estimators that are used in the new phase and are not affected by $\gamma$-suffix violations. And so, the output is $(1/2)\alpha$-accurate according to the configured accuracy of algorithm \texttt{RobustDE}. %
In the complement case where the condition is False, we have the following:
$$\left| \left\{ k\in [\mathsf{K}_{\STmon}] : |z_{\STmon}^{k} - \Next| < (3/4) \alpha\cdot \Next  \right\} \right| \geq
\frac{\mathsf{K}_{\STmon}}{2} -
\frac{4}{\eps_{\STmon}}\log\left( \frac{2m}{\delta^{N}}  \right)  \geq
\frac{4\cdot \mathsf{K}_{\STmon}}{10}
$$
Where the first inequality holds in the event of the bounded noises and the second inequality holds by asserting that 
$
\mathsf{K}_{\STmon} = \Omega \left( 
\frac{1}{\eps_{\STmon}}\log\left( \frac{m}{\delta^{N}}\right)
\right)
= \Omega \left( \frac{1}{\eps}
\sqrt{P_{\STmon} \cdot \log \left(\frac{1}{\delta^{\prime}} \right) }\log\left( \frac{m}{\delta^{N}}\right)
\right)
$.
That is, we have at least $40\%$ of the estimations $z_{\STmon}^{k}$ that are $(3/4)\alpha$ close to $\Next$. At the same time, by Lemma \ref{lem:AccurateMonEstimations}, at least $80\%$ of the estimators are $\alpha_{\STmon}$-accurate thus there exists an estimator that admit both. And so, by setting $\alpha_{\STmon} = (1/10)\alpha$ we have that (for any $\alpha \in (0,1)$):
$$
|\Next - \mathcal{F}(t)| \leq \alpha\cdot \mathcal{F}(t)
$$
We now address the failure probability $\delta$. Recall that  all noises in algorithm \texttt{Guardian} (we have at most $m$ draws of  $\Lap(2/\eps_{\STmon})$ and $m$ draws of $\Lap(4/\eps_{\STmon})$ noises) are bounded by $\frac{4}{\eps_{\STmon}}\log\left( \frac{2m}{\delta^{N}}  \right)$ w.p. at least $1-\delta^{N}$. Then by setting $\delta^{N} = \delta/4$, we have that the noises in algorithm \texttt{Guardian} are bounded as required w.p. at lest $\delta/4$. 
In addition, Lemma \ref{lem:AccurateMonEstimations} statement holds w.p. at least $1-\delta^{\prime}/100$. Configuring $\delta^{\prime} = \delta/(400m)$ yields that this lemma statement holds for all $t\in [m]$ w.p. at least $1-\delta/4$.
In addition, we configure the failure probability of \texttt{RobustDE} for $\delta/2$. 
That is, we have that w.p. at least $1-\delta$ all algorithm \texttt{Guardian} outputs are accurate in all $t\in [m]$.
\end{proof}

\subsection{Space complexity of the framework extension.}\label{sec:extensionSpace}
It remains to account for the space complexity of \texttt{RobustDE} with at most $\TN$ additional phase reset commands received externally from the \texttt{Guardian} algorithm.
The adaptation that is needed in \texttt{RobustDE} in order to facilitate external phase reset command is presented in \ref{alg:PhaseResetCommand} (we present only the relevant lines).
%
%%%%%%%%%%%%%%%%%%%%%%%%%%%%%%%%
% Algorithm: violation monitor
%%%%%%%%%%%%%%%%%%%%%%%%%%%%%%%%
\begin{algorithm*}[ht]
\caption{\texttt{PhaseResetCommand}}
\makeatletter\def\@currentlabel{\texttt{PhaseResetCommand}}\makeatother
\label{alg:PhaseResetCommand}

{\bf For ${\boldsymbol{t\in[m]}}$:}
\begin{enumerate}
    \setcounter{enumi}{7}
    \item Get the update \red{$\langle s_t, \Delta_t, b_t \rangle$} from $\mathcal{S}$ and feed $\langle s_t, \Delta_t \rangle$ into all estimators.

    \item If $\left| \left\{ k\in [\mathsf{K}_{\STwrapper}] : z_{\STwrapper}^{k} \notin \Big(\frac{1}{\Gamma}\cdot \mathsf{Z}_{\ST} \;,\; \Gamma\cdot\mathsf{Z}_{\ST} \Big)  \right\} \right| + \Lap(\eps^{-1}_{\STwrapper})> \mathsf{T}_{\STwrapper}$ \red{or $b_t = 1$}:\\
    Set $\tau=0$, redraw $\mathsf{T}_{\STwrapper}$.
    
    \item $\dots$
\end{enumerate}

\end{algorithm*}

\paragraph{External phase reset command in algorithm \ref{alg:ROEF}.} In order for \texttt{Guardian} to be able to trigger a phase reset command in algorithm \texttt{RobustDE}, we add an input to the stream, namely $b_t$, that signals an {\em external} phase reset command.
This input $b_t$ has an effect on the functionality of line \ref{algStep:STwrapper} and can initiate a phase reset.
That is, in \ref{algStep:STwrapper}, the condition is triggering initiation of a new phase (regardless of $\tau$ state) and in the extended version this initiation can also be triggered externally by the received input $b_t=1$.

Each external reset command comes with cost in terms of additional output modifications. 
As these additional output modifications require additional estimators in the framework to support them, we calculate a new sufficient value for the input parameter $\lambda$ of Algorithm  \texttt{RobustDE}. This parameter in the not-extended framework is bounding the flip number of the input stream. 
We calculate a new value for that parameter, denoted by $\hat{\lambda}$.
That value is sufficient to support in the extended framework a stream with a flip number of $\lambda$ and in addition, $\mu$ external reset commands.

\paragraph{Calibrating $\hat{\lambda}$.} Recall that in the analysis of \texttt{RobustDE} we calculate bounds for the number of output modification that are associated with each of the estimators levels, $C_j$ (see Lemma \ref{lem:outputUpdatesBounds}). 
It then follows for that analysis that configuring the capping parameter of each level, $P_j$, to be larger then $C_j$ (Corollary \ref{cor:NoCapping}) ensures no capping.
These bounds are stated w.r.t a bound on an $(O(\alpha),m)$-flip number bound $\lambda$ that is an input to the algorithm, and hold for the framework {\em without} external phase reset commands.
Since the extension introduces such external phase reset commands, the previous analysis needs to be adapted.
That is, we need to show new bounds for the number of output modification per estimators level $C_j$ for the extended framework w.r.t a stream that has a bounded $(\alpha^{\prime},m)$-flip number and $(\gamma_0, m)$-twist number.
We do that as follows: calculate a new input for the framework $\hat{\lambda}=f(\lambda, \mu)$ s.t.\ the computed parameters of the framework $P_j(\hat{\lambda})$ will be sufficient for no-capping-state for a $\lambda$ bounded $(\alpha^{\prime},m)$-flip number and $\mu$ bounded $(\gamma_0, m)$-twist number streams. 
The following lemma calculate such calibration of $\hat{\lambda}$:

\begin{lemma}[Calibration of $\hat{\lambda}$] \label{lem:lambdaCalibration}
Let $\SSS$ be a stream with $(\alpha^{\prime},m)$-flip number and $(\gamma_0, m)$-twist number bounded by $\lambda$ and $\mu$ correspondingly. Then,
$$
C_j \leq O\left(\frac{\hat{\lambda}}{2^j}\right)\text{,}
$$
where $\hat{\lambda} = O(\lambda + \mu\cdot \alpha^{-1})$,  $\alpha^{\prime} = (1/2)\cdot \StepSize(\alpha) = O(\alpha)$, $\gamma_0 =  \frac{1+\alpha_{\ST}}{1-\alpha_{\ST}} \Gamma^2 \cdot 2 \cdot \alpha = O(\alpha)$.
\end{lemma}

\begin{proof} First we look on some segment of the stream $\SSS$ corresponding to times between two consecutive phase resets (either an internal phase reset or a phase reset command received from \texttt{Guardian}). On each such segment we bound its $(\alpha^{\prime},m)$-flip number and calculate the resulting number of phases in that segment. Then we sum the total number of phases within all these segments. Finally, we bound the number of output modifications associated with each level from the bound of the number of phases.

\paragraph{Total number of phases.} By Lemma \ref{lem:maxMonTrigger} we have that there are at most $\mu$ reset commands issued from \texttt{Guardian} for $\SSS$. In addition, there are at most $\kappa = O(\alpha \lambda_{\alpha^{\prime}}(\SSS))$ internal resets (see Lemma \ref{lem:outputUpdatesBounds}). Denote by $\hat{\mu} = \mu + \kappa$ the number of phase resets in algorithm \texttt{RobustDE} (internal and external).
Let $\{r_i\}_{i \in [\hat{\mu}]}$, $r_i\in [m]$ be a set of times in which the $\hat{\mu}$ phase reset were executed.
For $i\in [\hat{\mu}]$, let $\SSS_i$ by the sub stream of $\SSS$ in times $[r_i, r_{i+1})$ (where $\SSS_{\hat{\mu}-1}$ is on times $[r_{\hat{\mu}-1},m-1]$). 
Also denote by $\phi, \phi_i$ the number of phases in $\SSS, \SSS_i$ correspondingly.
The following holds:
\begin{align*}
\phi = \sum_{i\in [\hat{\mu}]} \phi_i 
&\stackrel{1}{=} \sum_{i\in [\hat{\mu}]} \left\lceil \frac{C(\SSS_i)}{\PhaseSize+1} \right\rceil 
\leq \sum_{i\in [\hat{\mu}]} \left( \frac{C(\SSS_i)}{\PhaseSize+1} + 1 \right) 
\stackrel{2}{\leq} \hat{\mu} + \frac{\sum_{i\in [\hat{\mu}]}  \lambda_{\alpha^{\prime}}(\SSS_i)+1}{\PhaseSize+1}  \\
&\stackrel{3}{\leq} \hat{\mu} + \frac{\lambda_{\alpha^{\prime}}(\SSS) + 3\hat{\mu}}{\PhaseSize+1}  
\leq \mu + \kappa + \frac{\lambda_{\alpha^{\prime}}(\SSS) + 3(\mu + \kappa)}{\PhaseSize+1} 
= O\left(\mu + \frac{\lambda}{\PhaseSize} \right)
\end{align*}
where (1) is true since on each segment $\SSS_i$ there is no phase reset and we start a new phase every $\PhaseSize+1$ number of steps (see the proof of Lemma \ref{lem:outputUpdatesBounds}), 
(2) holds since for every output modification the value of $\mathcal{F}$ progresses by at least factor of $(1/2)\cdot \StepSize \geq \alpha^{\prime}$ (see the proof of Lemma \ref{lem:outputUpdatesBounds}),
(3) is true by Lemma \ref{lem:substreamFlipNumber}.

\paragraph{Output per level.} In every phase there are at most $\PhaseSize$ number of output modifications. And so (See the proof of Lemma \ref{lem:outputUpdatesBounds}), for $j\in [\beta]$, the number of output modifications associated with  level $j$ estimators is $O(\PhaseSize/2^{j})$. That is:
$$
C_j(\SSS) = \phi \cdot O\left(\frac{\PhaseSize}{2^{j}}\right) = O\left(\mu + \frac{\lambda}{\PhaseSize} \right)\cdot  O\left(\frac{\PhaseSize}{2^{j}}\right)= 
O\left(\frac{\mu\cdot \alpha^{-1} + \lambda}{2^{j}}\right) = 
O\left(\frac{\hat{\lambda}}{2^{j}}\right)
$$

\end{proof}
An immediate Corollary is that calibrating the input $\hat{\lambda} = \Omega(\mu \alpha^{-1} + \lambda)$, algorithm \texttt{RobustDE} will not get to capping state. That is since algorithm \texttt{RobustDE} is setting the parameters $P_j = \Omega(\hat{\lambda}/2^{j})$ for an input $\hat{\lambda}$, resulting in $P_j > C_j(\SSS)$ as required.

\begin{corollary}[No capping in extended \ref{alg:ROEF}.] 
Let $\SSS$ be a stream with $(\alpha^{\prime},m)$-flip number and $(\gamma_0, m)$-twist number bounded by $\lambda$ and $\mu$ correspondingly. Calibrating $\hat{\lambda}$, the input of $\ref{alg:ROEF}$, to $\hat{\lambda} = \Omega(\mu\cdot \alpha^{-1} + \lambda)$ is sufficient to ensure $\ref{alg:ROEF}$ will not get into capping state.
\end{corollary}
We now present the resulting space bounds of the extended framework. 

\begin{theorem}[Extended framework for Adversarial Streaming - Space]\label{thm:ExtensionSpace}
Provided that there exist: 
\begin{enumerate} 
	\item An oblivious streaming algorithm $\mathsf{E}_{\ST}$ for functionality $\mathcal{F}$, that guarantees that with probability at least $9/10$ all of it's estimates are accurate to within a multiplicative error of $(1\pm \alpha_{\ST})$ with space complexity of $S_{\ST}(\alpha_{\ST}, \frac{1}{10}, n,m)$
	\item For every $\gamma$ there is a $(\gamma,\alpha_{\DE},\frac{1}{10})$-$\DE$ for $\mathcal{F}$ using space $\gamma \cdot S_{\DE}(\alpha_{\DE},\frac{1}{10},n,m)$.
\end{enumerate}  
Then there exist an adversarially robust streaming algorithm for functionality $\mathcal{F}$ that for any stream $\SSS$ with a bounded flip number $\lambda_{\alpha^{\prime},m}(\SSS)< \lambda$ and a bounded twist number $\mu_{\gamma_0, m}(\SSS) < \mu$ (where $\alpha^{\prime}, \gamma_0 = O(\alpha)$), s.t.\ with probability at least $1-\delta$ its output is accurate to within a multiplicative error of $(1\pm \alpha)$ for all times $t\in [m]$, and has a space complexity of
$$
O\left(
\sqrt{\alpha\cdot \lambda + \mu} \cdot 
\polylog\text{}_{\ALG}
\right)\cdot 
\left[
S_{\ST} + 
S_{\DE}
\right]  \text{.}
$$
Where:
\begin{enumerate}
	\item $S_{\ST} =  S_{\ST}(O(\alpha),\frac{1}{10},n,m)$.
	\item $S_{\DE} =  S_{\DE}(O(\alpha/\log(\alpha^{-1})),\frac{1}{10\hat{\lambda}},n,m)$, for $\hat{\lambda} = O(\lambda + \mu \cdot \alpha^{-1})$
	\item $\polylog_{\ALG} 
 = \polylog(\lambda, \mu, \alpha^{-1}, \delta^{-1}, m, n)$.
\end{enumerate}
\end{theorem}
\begin{proof} In order to use the framework of Algorithm $\texttt{RobustDE}$ for functionality $\mathcal{F}$ it is necessary (by Theorem \ref{thm:FrameworkSpace}) to have for every $\gamma,p$ a $(\gamma,\alpha_{\TDE},p,\frac{1}{10})$-$\TDE$ for $\mathcal{F}$ using space $\gamma \cdot S_{\TDE}(\alpha_{\TDE},p,n,m)$. 
By Corollary \ref{cor:TDEfromDE}, it is possible to construct a $\TDE$ from a $\DE$ (that has the same accuracy guarantee) with space of $ S_{\TDE}(\gamma,\alpha,\delta,p,n,m) = 2\cdot S_{\DE}(\gamma,\alpha,\delta/p,n,m)$.
Thus having a $(\gamma,\alpha_{\DE},1/10)$-$\DE$ with space of $\gamma\cdot S_{\DE}(\alpha_{\DE},1/10,n,m)$ imply a $(\gamma,\alpha_{\TDE}, p, 1/10)$-$\TDE$ with space of $S_{\TDE}=2\cdot S_{\DE}(\gamma, \alpha_{\TDE}, 1/(10\cdot p),n,m)$ with the same accuracy guarantee.

\paragraph{Sufficient parameter calibration.} By Lemma \ref{lem:lambdaCalibration}, calibrating $\hat{\lambda} = \Omega(\lambda + \mu\cdot \alpha^{-1})$ is sufficient to ensure that algorithm \texttt{RobustDE} will not get to capping state. 
(in addition in Lemma \ref{lem:lambdaCalibration} the required accuracy constant of the flip number is required to be $\alpha^{\prime} = (1/2)\cdot \StepSize(\alpha) \leq (1/2)\cdot \alpha/(2\Gamma) $ that is $\alpha^{\prime} = O(\alpha)$.)
If in addition we configure $P_{\STmon} > \mu$, then by Lemmas \ref{lem:privacyMonitor}, \ref{lem:monOutputAccuracy} we have that the output of \texttt{Guardian} is $\alpha$-accurate in all times $t\in [m]$.

\paragraph{Space of \ref{alg:ViolationMonitor}.} Space of \texttt{Guardian} alone is accounted with the space of a $\mathsf{K}_{\STmon}$ number of $\mathsf{E}_{\STmon}$ estimators. These are strong trackers with accuracy $\alpha_{\STmon} = (1/10)\cdot \alpha = O(\alpha)$.
\begin{align*}
\Space(\texttt{Guardian}(\mathsf{E}_{\ST}, \mathsf{E}_{\DE}, \lambda, \mu, \alpha, \delta, n,m)) &= \mathsf{K}_{\STmon}\cdot \Space(\mathsf{E}_{\STmon})\\
&=O\left( 
\sqrt{P_{\STmon} \cdot \log \left(\frac{1}{\delta^{\prime}} \right) }\log\left( \frac{m}{\delta^{N}}\right)
\right) \cdot S_
{\ST}\\
&=O\left( 
\sqrt{\mu} \cdot \log ^{1.5}\left( \frac{m}{\delta}\right)
\right) \cdot S_
{\ST}
\end{align*}
Since $\delta^{\prime} = O(\delta/m)$, $\delta^{N} = O(\delta)$, $P_{\STmon} = O(\mu)$.

\paragraph{Space of $\ref{alg:ROEF}$.} We have that  $S_{\TDE}=2\cdot S_{\DE}(\gamma, \alpha_{\TDE}, 1/(10\cdot p),n,m)$ and $\hat{\lambda} = O(\lambda + \mu\cdot \alpha^{-1})$.
And so, by plugging in $\hat{\lambda}$, $S_{\TDE}(\alpha,\delta,p,n,m) = 2\cdot S_{\DE}(\alpha,\delta/p,n,m)$ in \ref{thm:FrameworkSpace} we get the required bounds:
\begin{align*}
\Space(\texttt{RobustDE}(\mathsf{E}_{\ST}, \mathsf{E}_{\DE}, \hat{\lambda}, \alpha, \delta, n,m)) &= O\left(
\sqrt{\alpha\cdot \hat{\lambda}} \cdot 
\polylog\text{}_{\ALG}
\right)\cdot 
\left[
S_{\ST} + 
S_{\TDE}
\right] \\
&= O\left(
\sqrt{\alpha\cdot \lambda + \mu} \cdot 
\polylog\text{}_{\ALG}
\right)\cdot 
\left[
S_{\ST} + 
S_{\DE}
\right] 
\end{align*}
Where 
\begin{enumerate}
	\item $S_{\ST} = S_{\ST}(\alpha_{\ST},1/10,n,m) = S_{\ST}(O(\alpha),1/10,n,m)$.
	\item $S_{\DE} = S_{\DE}(\alpha_{\TDE},1/(10\hat{\lambda}),n,m) = S_{\DE}(O(\alpha/\log(\alpha^{-1})),1/(10\hat{\lambda}),n,m)$.
	\item $\polylog_{\ALG} =  \left[ 
\log\left( \frac{m}{\delta^{*}}\right) + \log\left( \frac{\hat{\lambda}}{\alpha\delta^{*}} \log(n)\right)
\right] 
\sqrt{\log\left(\frac{m}{\delta^{*}}\right)}
 = \polylog(\lambda + \mu\cdot \alpha^{-1}, \alpha^{-1}, \delta^{-1}, m, n)$.
    \item $\hat{\lambda} = O(\lambda + \mu \cdot \alpha^{-1})$, $\delta^{*} = \delta/\log(\alpha^{-1})$.
\end{enumerate}

\paragraph{Total space of extension.} It remain to calculate the resulting bounds of algorithms \texttt{RobustDE}, \texttt{Guardian}. Since $\log^{1.5}(m/\delta) = O(\polylog_{\ALG})$, then the space of \texttt{Guardian} is subsumed in the space of \texttt{RobustDE}.
\end{proof}
To apply our extended framework to $F_2$, we first cite constructions of a strong tracker and of a difference estimator for $F_2$, and then calculate the overall space complexity that results from our framework.

\begin{theorem}[Oblivious strong tracker for $F_2$ \cite{alon1999space,DBLP:conf/soda/ThorupZ04}]\label{thm:F2STSpace} There exists a strong tracker for $F_2$ functionality s.t.\ for every stream $S$ of length $m$ outputs on every time step $t\in [m]$ an $\alpha$-accurate estimation $z_t\in (1\pm \alpha)\cdot F_2(S)$ with probability at least $9/10$ and has space complexity of 
$
O
\left(
\frac{1}{\alpha^{2}}\log m
\left( 
\log n + \log m
\right) 
\right)
$
\end{theorem}

\begin{theorem}[Oblivious \DE\text{} for $F_2$ \cite{woodruff2020tight}]\label{thm:DESpace} There exists a $(\gamma, \alpha, \delta)$-difference estimator for $F_2$ that uses space of 
$O\left (\gamma\cdot \frac{\log n}{\alpha^{2}} \left(\log \frac{1}{\alpha} + \log \frac{1}{\delta} \right) \right)$
\end{theorem}

\begin{theorem}[$F_2$ Robust estimation]\label{thm:F2Space}
There exists an adversarially robust $F_2$ estimation algorithm for turnstile streams of length $m$ with a bounded  $(\alpha^{\prime}, m)$-flip number and $(\gamma_0, m)$-twist number  with parameters $\lambda$ and $\mu$ correspondingly (where $\alpha^{\prime}, \gamma_0 = O(\alpha)$), that guarantees $\alpha$-accuracy with probability at least $1-1/m$ in all time $t\in [m]$ with space complexity of
$$\tilde{O}\left(\frac{\sqrt{\alpha\lambda+\mu}}{\alpha^{2}}\log^{3.5}(m)\right)\text{.}$$
where $\tilde{O}$ stands for omitting $\polylog(\alpha^{-1})$ factors.
\end{theorem}

\begin{proof} By Theorem \ref{thm:F2STSpace}, there exists a $(\alpha_{\ST}, \frac{1}{10})$-strong tracker for functionality $F_2$ with space complexity of $S_{\ST}(\alpha_{\ST}, \frac{1}{10}, n,m) = O\left( \alpha_{\ST}^{-2} \log m \left( \log n + \log m \right) \right) $. 
For $m=\poly(n)$ we get 
$S_{\ST} = O(\alpha^{-2} \log^2(m))$.
By theorem \ref{thm:DESpace}, there exists a $(\gamma, \alpha_{\DE}, \frac{1}{10})$-difference estimator for functionality $F_2$ with space complexity of $\gamma\cdot S_{\DE}(\alpha_{\DE}, \delta, n,m)$ where $S_{\DE} = O\left (\alpha_{\DE}^{-2} \log n \left(\log \alpha_{DE}^{-1} + \log \delta^{-1} \right) \right)$.
Then by Theorem \ref{thm:ExtensionSpace} we have:
\begin{align*}
&Space(F_2\text{-Extension}) = O\left(
\sqrt{\alpha\cdot \lambda + \mu} \cdot 
\polylog\text{}_{\ALG}
\right)\cdot 
\left[ S_{\ST} + S_{\DE} \right] \\
&\stackrel{1}{=} \tilde{O}\left(
\sqrt{\alpha\cdot \lambda + \mu} \cdot 
\left[ 
\log\left( \frac{m}{\delta}\right) + \log\left( \frac{\lambda + \mu\alpha^{-1}}{\alpha\delta} \log(n)\right)
\right] 
\sqrt{\log\left(\frac{m}{\delta}\right)}
\cdot 
\left[ S_{\ST} +  S_{\DE} \right]
\right)\\
&\stackrel{2}{=} \tilde{O}\left(
\sqrt{\alpha\cdot \lambda + \mu} \cdot 
\log^{1.5}\left( \frac{m}{\delta}\right) 
\left[ \alpha^{-2} \log^2(m)) + 
\alpha^{-2}\log(m) \log (\delta^{-1}) \right]
\right)
\end{align*}
where (1) is by plugging in $\polylog_{\ALG}$ and omitting factors of $\polylog \alpha^{-1}$, (2) is by again omitting factors of $\polylog \alpha^{-1}$, noting that $\lambda, \mu < m$ and by assuming that $n = \poly(m)$. Now, setting $\delta = 1/m$ we get:
$$
Space(F_2\text{-Extension}) = \tilde{O} \left(
\frac{\sqrt{\alpha\cdot \lambda + \mu}}{\alpha^2}\log^{3.5}(m)
\right)
$$
\end{proof}

\end{document}